\documentclass[a4paper,USenglish,autoref]{lipics-v2021}

\usepackage{algorithm}
\usepackage[noend]{algpseudocode}
\usepackage{booktabs}
\usepackage{enumitem}

\newcommand{\bigo}[1]{\ensuremath{\mathcal{O}(#1)}}
\renewcommand{\Pr}[1]{\ensuremath{\text{Pr}\left[#1\right]}}
\newcommand{\true}{\ensuremath{\textsc{true}}}
\newcommand{\false}{\ensuremath{\textsc{false}}}
\newcommand{\nil}{\ensuremath{\textsc{null}}}
\newcommand{\Gtri}{\ensuremath{G_{\Delta}}}

\newcommand{\Connected}{\ensuremath{\textsc{Connected}}}
\newcommand{\Read}{\ensuremath{\textsc{Read}}}
\newcommand{\Write}{\ensuremath{\textsc{Write}}}
\newcommand{\Contract}{\ensuremath{\textsc{Contract}}}
\newcommand{\Expand}{\ensuremath{\textsc{Expand}}}
\newcommand{\Push}{\ensuremath{\textsc{Push}}}
\newcommand{\Pull}{\ensuremath{\textsc{Pull}}}
\newcommand{\Lock}{\ensuremath{\textsc{Lock}}}
\newcommand{\Unlock}{\ensuremath{\textsc{Unlock}}}

\newcommand{\alg}{\ensuremath{\mathcal{A}}}
\newcommand{\sched}{\ensuremath{\mathcal{S}}}
\newcommand{\shape}{\ensuremath{\texttt{shape}}}
\newcommand{\lock}{\ensuremath{\texttt{lock}}}

\newcommand{\activity}{\ensuremath{\texttt{act}}}
\newcommand{\awaken}{\ensuremath{\texttt{awaken}}}
\newcommand{\xflag}{\ensuremath{\texttt{flag}}}
\newcommand{\algHex}{\textsf{Hexagon-Formation}}
\newcommand{\state}{\ensuremath{\texttt{state}}}
\newcommand{\parent}{\ensuremath{\texttt{parent}}}
\newcommand{\dir}{\ensuremath{\texttt{dir}}}

\makeatletter
\algrenewcommand\ALG@beginalgorithmic{\small}
\algrenewcommand\alglinenumber[1]{\scriptsize #1:}
\makeatother

\algdef{SE}[SUBALG]{Indent}{EndIndent}{}{\algorithmicend\ }%
\algtext*{Indent}
\algtext*{EndIndent}
\makeatletter
\newcommand{\multiline}[1]{%
  \begin{tabularx}{\dimexpr\linewidth-\ALG@thistlm}[t]{@{}X@{}}
    #1
  \end{tabularx}
}
\makeatother

\newtheorem{convention}{Convention}

\newif\ifcomment
\commenttrue    

\newif\iffigabbrv
\figabbrvtrue    
\figabbrvfalse   
\newcommand{\figtext}{\iffigabbrv Fig.\else Figure\fi}

\title{The Canonical Amoebot Model: Algorithms and Concurrency Control}
\titlerunning{The Canonical Amoebot Model}

\author{Joshua J. Daymude}{School of Computing and Augmented Intelligence\\Biodesign Center for Biocomputing, Security and Society\\Arizona State University, Tempe, AZ, USA}{jdaymude@asu.edu}{https://orcid.org/0000-0001-7294-5626}{NSF (CCF-1733680), U.S.\ ARO (MURI W911NF-19-1-0233), the Momental Foundation's Mistletoe Research Fellowship, and the ASU Biodesign Institute.}
\author{Andr\'ea W. Richa}{School of Computing and Augmented Intelligence, Arizona State University, Tempe, AZ, USA}{aricha@asu.edu}{https://orcid.org/0000-0003-3592-3756}{NSF (CCF-1733680) and U.S.\ ARO (MURI W911NF-19-1-0233).}
\author{Christian Scheideler}{Department of Computer Science, Paderborn University, Paderborn, Germany}{scheideler@upb.de}{https://orcid.org/0000-0002-5278-528X}{DFG Project SCHE 1592/6-1.}
\authorrunning{J.\ J.\ Daymude, A.\ W.\ Richa, and C.\ Scheideler}

\Copyright{Joshua J.\ Daymude, Andr\'ea W.\ Richa, and Christian Scheideler}

\ccsdesc{Theory of computation~Distributed computing models}
\ccsdesc{Theory of computation~Concurrency}
\ccsdesc{Theory of computation~Self-organization}

\keywords{Programmable matter, self-organization, distributed algorithms, concurrency}

\acknowledgements{We thank Nicola Santoro and Paola Flocchini for their constructive feedback and Kristian Hinnenthal for his contributions to a preliminary version of this work.}

\hideLIPIcs
\nolinenumbers

\begin{document}

\maketitle

\begin{abstract}
    The \textit{amoebot model} abstracts active programmable matter as a collection of simple computational elements called \textit{amoebots} that interact locally to collectively achieve tasks of coordination and movement.
    Since its introduction at SPAA 2014, a growing body of literature has adapted its assumptions for a variety of problems; however, without a standardized hierarchy of assumptions, precise systematic comparison of results under the amoebot model is difficult.
    We propose the \textit{canonical amoebot model}, an updated formalization that distinguishes between core model features and families of assumption variants.
    A key improvement addressed by the canonical amoebot model is \textit{concurrency}.
    Much of the existing literature implicitly assumes amoebot actions are isolated and reliable, reducing analysis to the sequential setting where at most one amoebot is active at a time.
    However, real programmable matter systems are concurrent.
    The canonical amoebot model formalizes all amoebot communication as message passing, leveraging adversarial activation models of concurrent executions.
    Under this granular treatment of time, we take two complementary approaches to \textit{concurrent algorithm design}.
    We first establish a set of \textit{sufficient conditions} for algorithm correctness under any concurrent execution, embedding concurrency control directly in algorithm design.
    We then present a \textit{concurrency control framework} that uses locks to convert amoebot algorithms that terminate in the sequential setting and satisfy certain conventions into algorithms that exhibit equivalent behavior in the concurrent setting.
    As a case study, we demonstrate both approaches using a simple algorithm for \textit{hexagon formation}.
    Together, the canonical amoebot model and these complementary approaches to concurrent algorithm design open new directions for distributed computing research on programmable matter.
\end{abstract}

\section{Introduction} \label{sec:intro}

The vision of \textit{programmable matter} is to realize a material that can dynamically alter its physical properties in a programmable fashion, controlled either by user input or its own autonomous sensing of its environment~\cite{Toffoli1991-programmablematter}.
Towards a formal characterization of the minimum capabilities required by individual modules of programmable matter to achieve a given system behavior, many abstract models have been proposed over the last several decades~\cite{Angluin2006-computationnetworks,Chalk2018-freezingsimulates,Chirikjian1994-kinematicsmetamorphic,DAngelo2020-asynchronoussilent,Flocchini2019-distributedcomputing,Michail2016-simpleefficient,Patitz2014-introductiontilebased,Piranda2018-designingquasispherical,Woods2013-activeselfassembly}.
We focus on the \textit{amoebot model}~\cite{Daymude2019-computingprogrammable,Derakhshandeh2014-amoebotba} which is motivated by micro- and nano-scale robotic systems with strictly limited computational and locomotive capabilities~\cite{Blackiston2021-cellularplatform,Hines2017-softactuators,Kriegman2020-scalablepipeline,Liu2021-autoperforationtwodimensional,Xie2019-reconfigurablemagnetic,Yang2019-syntheticcells}.
The amoebot model abstracts active programmable matter as a collection of simple computational elements called \textit{amoebots} that utilize local interactions to collectively achieve tasks involving coordination, movement, and reconfiguration.
Since its introduction at SPAA 2014, the amoebot model has been used to study both fundamental problems---such as leader election~\cite{Bazzi2019-stationarydeterministic,DAngelo2020-asynchronoussilent,Daymude2017-improvedleader,Derakhshandeh2015-leaderelection,DiLuna2020-shapeformation,Dufoulon2021-efficientdeterministic,Emek2019-deterministicleader,Gastineau2019-distributedleader,Gastineau2020-leaderelection} and shape formation~\cite{Cannon2016-markovchain,Derakhshandeh2015-algorithmicframework,Derakhshandeh2016-universalshape,DiLuna2020-shapeformation,DiLuna2020-mobileram,Nokhanji2020-linereconfiguration}---as well as more complex behaviors including object coating~\cite{Daymude2018-runtimeuniversal,Derakhshandeh2017-universalcoating}, convex hull formation~\cite{Daymude2020-convexhull}, bridging~\cite{AndresArroyo2018-stochasticapproach}, spatial sorting~\cite{Cannon2019-localstochastic}, and fault tolerance~\cite{Daymude2021-bioinspiredenergy,DiLuna2018-linerecovery}.

With this growing body of amoebot model literature, it is evident that the model has evolved---and, to some extent, fractured---during its lifetime as assumptions were updated to support individual results, capture more realistic settings, or better align with other models of programmable matter.
This makes it difficult to conduct any systematic comparison between results under the amoebot model (see, e.g., the overlapping but distinct features used for comparison of leader election algorithms in~\cite{Bazzi2019-stationarydeterministic} and~\cite{Emek2019-deterministicleader}), let alone between amoebot model results and those of related models (e.g., those from the established \textit{autonomous mobile robots} literature~\cite{Flocchini2019-distributedcomputing}).
To address the ways in which the amoebot model has outgrown its original rigid formulation, we propose the \textit{canonical amoebot model} that includes a standardized, formal hierarchy of assumptions for its features to better facilitate comparison of its results.
Moreover, such standardization will more gracefully support future model generalizations by distinguishing between core features and assumption variants.

A key area of improvement addressed by the canonical amoebot model is \textit{concurrency}.
The original model treats concurrency at a high level, implicitly assuming an isolation property that prohibits concurrent amoebot actions from interfering with each other.
Furthermore, amoebots are usually assumed to be \textit{reliable}; i.e., they cannot crash or exhibit Byzantine behavior.
Under these simplifying assumptions, most existing algorithms are analyzed for correctness and runtime as if they are executed \textit{sequentially}, with at most one amoebot acting at a time.
Notable exceptions include the recent work of Di Luna et al.~\cite{DiLuna2020-shapeformation,DiLuna2018-linerecovery,DiLuna2020-mobileram} and Nokhanji and Santoro~\cite{Nokhanji2020-linereconfiguration} that adopt ideas from the ``look-compute-move'' paradigm used in autonomous mobile robots to bring the amoebot model closer to a realistic, concurrent setting.
Our canonical amoebot model furthers these efforts by formalizing all communication and cooperation between amoebots as message passing while also addressing the complexity of potential conflicts caused by amoebot movements.
This careful formalization allows us to use standard adversarial activation models from the distributed computing literature to describe concurrency~\cite{Altisen2019-introductiondistributed}.

This fine-grained treatment of concurrency in the canonical amoebot model lays the foundation for the design and analysis of \textit{concurrent amoebot algorithms}.
Concurrency adds significant design complexity, allowing concurrent amoebot actions to mutually interfere, conflict, affect outcomes, or fail in ways far beyond what is possible in the sequential setting.
As a tool for controlling concurrency, we introduce a \Lock\ operation in the canonical amoebot model enabling amoebots to attempt to gain exclusive access to their neighborhood.

We then take two complementary approaches to concurrent amoebot algorithm design: a direct approach that embeds concurrency control directly into the algorithm's design without requiring locks, and an indirect approach that relies on the \Lock\ operation to mitigate issues of concurrency.
In the first approach, we establish a set of \textit{general sufficient conditions} for amoebot algorithm correctness under any adversary---sequential or asynchronous, fair or unfair---using the \textit{hexagon formation problem} (see, e.g.,~\cite{Daymude2019-computingprogrammable,Derakhshandeh2015-algorithmicframework}) as a case study.
Our \algHex\ algorithm demonstrates that locks are not necessary for correctness even under an unfair, asynchronous adversary.
However, this algorithm's asynchronous correctness relies critically on its actions succeeding despite any concurrent action executions, which may be a difficult property to obtain in general.

For our second approach, we present a \textit{concurrency control framework} using the \Lock\ operation that, given an amoebot algorithm that terminates under any sequential execution and satisfies some basic conventions, produces an algorithm that exhibits equivalent behavior under any asynchronous execution.
This framework establishes a general design paradigm for concurrent amoebot algorithms: one can first design an algorithm with correct behavior in the simpler sequential setting and then, by ensuring it satisfies our framework's conventions, automatically obtain a correct algorithm for the asynchronous setting.
The convenience of this approach comes at the cost of limiting the full generality of the canonical amoebot model to comply with the framework's conventions.
Nevertheless, we prove that the \algHex\ algorithm satisfies these conventions and thus is compatible with the framework.

\subparagraph*{Our Contributions.}

We summarize our contributions as follows.
\begin{itemize}
    \item The \textit{canonical amoebot model}, an updated formalization that treats amoebot actions at the fine-grained level of message passing and distinguishes between core model features and hierarchies of assumption variants (Section~\ref{sec:model}).
    
    \item General \textit{sufficient conditions} for amoebot algorithm correctness under any adversary and an algorithm for \textit{hexagon formation} that satisfies these conditions (Section~\ref{sec:hexagon}).
    
    \item A \textit{concurrency control framework} that converts amoebot algorithms that terminate under any sequential execution and satisfy certain conventions into algorithms that exhibit equivalent behavior under any asynchronous execution (Section~\ref{sec:framework}), and an application of this framework to the algorithm for hexagon formation (Section~\ref{subsec:conventions}).
\end{itemize}

\subparagraph*{Relationship to Prior Versions.}

This work improves over its conference version published at DISC 2021~\cite{Daymude2021-canonicalamoebot} in several aspects.
First, this work contains all details and proofs that were omitted due to conference space constraints, including the message passing implementations of amoebot operations.
Second, whereas the original publication treated the newly added \Lock\ and \Unlock\ operations as black boxes, this work suggests a possible implementation based on the recent algorithm for local mutual exclusion in dynamic networks~\cite{Daymude2022-localmutual}.
Third, this work improves the usability of the concurrency control framework.
Of the three algorithm conventions required for compatibility with the framework in~\cite{Daymude2021-canonicalamoebot}, the most difficult to understand and verify is \textit{monotonicity}.
In fact, it was not known whether any amoebot algorithm involving movement could satisfy monotonicity, posing a serious limitation to the framework's use.
This work replaces monotonicity with a more general and more easily-understood convention, \textit{expansion-robustness}, without changing the framework's guarantees.
Finally, this work proves that the algorithm for hexagon formation (Section~\ref{sec:hexagon}) is expansion-robust, thus identifying the first algorithm involving movement that is compatible with the concurrency control framework and resolving the previously open question.

\subsection{Related Work} \label{subsec:relwork}

There are many theoretical models of programmable matter, ranging from the non-spatial \textit{population protocols}~\cite{Angluin2006-computationnetworks} and \textit{network constructors}~\cite{Michail2016-simpleefficient} to the tile-based models of \textit{DNA computing} and \textit{molecular self-assembly}~\cite{Chalk2018-freezingsimulates,Patitz2014-introductiontilebased,Woods2013-activeselfassembly}.
Most closely related to the amoebot model studied in this work is the well-established literature on \textit{autonomous mobile robots}, and in particular those using discrete, graph-based models of space (see Chapter 1 of~\cite{Flocchini2019-distributedcomputing} for a recent overview).
Both models assume anonymous individuals that can actively move, lacking a global coordinate system or common orientation, and having strictly limited computational and sensing capabilities.
In addition, stronger capabilities assumed by the amoebot model also appear in more recent variants of mobile robots, such as persistent memory in the \textit{$\mathcal{F}$-state} model~\cite{Barrameda2008-deploymentasynchronous,Flocchini2016-rendezvousconstant} and limited communication capabilities in \textit{luminous robots}~\cite{Das2012-powerlights,Das2016-autonomousmobile,DiLuna2017-mutualvisibility}.

There are also key differences between the amoebot model and the standard assumptions for mobile robots, particularly around their treatment of physical space, the structure of individuals' actions, and concurrency.
First, while the discrete-space mobile robots literature abstractly envisions robots as agents occupying nodes of a graph---allowing multiple robots to occupy the same node---the amoebot model assumes \textit{physical exclusion} that ensures each node is occupied by at most one amoebot at a time, inspired by the real constraints of self-organizing micro-robots and colloidal state machines~\cite{Blackiston2021-cellularplatform,Hines2017-softactuators,Kriegman2020-scalablepipeline,Liu2021-autoperforationtwodimensional,Xie2019-reconfigurablemagnetic,Yang2019-syntheticcells}.
Physical exclusion introduces conflicts of movement (e.g., two amoebots concurrently moving into the same space) that must be handled carefully in algorithm design.

Second, mobile robots are assumed to operate in \textit{look-compute-move} cycles, where they take an instantaneous snapshot of their surroundings (look), perform internal computation based on the snapshot (compute), and finally move to a neighboring node determined in the compute stage (move).
While it is reasonable to assume robots may instantaneously snapshot their surroundings due to all information being visible, the amoebot model---and especially the canonical version presented in this work---treats all inter-amoebot communication as asynchronous message passing, making snapshots nontrivial. Moreover, amoebots have \textit{read and write} operations allowing them to access or update variables stored in the persistent memories of their neighbors that do not fit cleanly within the look-compute-move paradigm.

Finally, the mobile robots literature has a well-established and carefully studied hierarchy of \textit{adversarial schedulers} capturing assumptions on concurrency that the amoebot model has historically lacked.
In fact, other than notable recent works that adapt look-compute-move cycles and a semi-synchronous scheduler from mobile robots to the amoebot model~\cite{DiLuna2020-shapeformation,DiLuna2018-linerecovery,DiLuna2020-mobileram,Nokhanji2020-linereconfiguration}, most amoebot literature assumes only sequential activations.
A key contribution of our canonical amoebot model presented in this work is a hierarchy of concurrency and fairness assumptions similar in spirit to that of mobile robots, though our underlying message passing design and lack of explicit action structure require different formalizations.

\section{The Canonical Amoebot Model} \label{sec:model}

\begin{table}[t]
    \centering
    \caption{Summary of assumption variants in the canonical amoebot model, each organized from most to least general.
    Variants marked with $*$ have been considered in existing literature, and variants marked with $\dag$ are the focus of the algorithmic results in this work.}
    \label{tab:variants}
    \begin{tabular}{llp{96mm}}
        \toprule
        & \textbf{Variant} & \textbf{Description} \\
        \midrule
        \parbox[t]{3mm}{\multirow{2}{*}{\rotatebox[origin=c]{90}{Space}}} & General$^*$ & $G$ is any infinite, undirected graph. \\
        & Geometric$^{*,\dag}$ & $G = \Gtri$, the triangular lattice. \\
        \midrule
        \parbox[t]{3mm}{\multirow{4}{*}{\rotatebox[origin=c]{90}{Orientation}}} & Assorted$^{*,\dag}$ & Assorted direction and chirality. \\
        & Common Chirality$^*$ & Assorted direction but common chirality. \\
        & Common Direction & Common direction but assorted chirality. \\
        & Common & Common direction and chirality. \\
        \midrule
        \parbox[t]{3mm}{\multirow{4}{*}{\rotatebox[origin=c]{90}{Memory}}} & Oblivious & No persistent memory. \\
        & Constant-Size$^{*,\dag}$ & Memory size is $\bigo{1}$. \\
        & Finite & Memory size is $\bigo{f(n)}$, some function of the system size. \\
        & Unbounded & Memory size is unbounded. \\
        \midrule
        \parbox[t]{3mm}{\multirow{5}{*}{\rotatebox[origin=c]{90}{Concurrency}}} & Asynchronous$^\dag$ & Any amoebots can be simultaneously active. \\
        & Synchronous$^*$ & Any amoebots can simultaneously execute a single action per discrete step. Each step has an evaluation phase and an execution phase. \\
        & $k$-Isolated & No amoebots within hop distance $k$ can be simultaneously active. \\
        & Sequential$^*$ & At most one amoebot is active per time. \\
        \midrule
        \parbox[t]{3mm}{\multirow{3}{*}{\rotatebox[origin=c]{90}{Fairness}}} & Unfair$^\dag$ & Some enabled amoebot is eventually activated. \\
        & Weakly Fair$^*$ & Every continuously enabled amoebot is eventually activated. \\
        & Strongly Fair & Every amoebot enabled infinitely often is activated infinitely often. \\
        \bottomrule
    \end{tabular}
\end{table}

We introduce the \textit{canonical amoebot model} as an update to the model's original formulation~\cite{Daymude2019-computingprogrammable,Derakhshandeh2014-amoebotba}.
This update has two main goals.
First, we model all amoebot actions and operations using message passing, leveraging this finer level of granularity for a formal treatment of concurrency.
Second, we clearly delineate which assumptions are fixed features of the model and which have stronger and weaker variants, providing unifying terminology for future amoebot model research.
Unless variants are explicitly listed, the following description of the canonical amoebot model details its core, fixed assumptions.
The variants are summarized in Table~\ref{tab:variants}; we anticipate that this list will grow as future research develops new adaptations and generalizations of the model.

In the canonical amoebot model, programmable matter consists of individual, homogeneous computational elements called \textit{amoebots}.
The structure of an amoebot system is represented as a subgraph of an infinite, undirected graph $G = (V,E)$ where $V$ represents all relative positions an amoebot can occupy and $E$ represents all atomic movements an amoebot can make.
Each node in $V$ can be occupied by at most one amoebot at a time.
There are many potential variants with respect to space; the most common is the \textit{geometric} variant that assumes $G = \Gtri$, the triangular lattice (\figtext~\ref{fig:modellattice}).

\begin{figure}[t]
    \centering
    \begin{subfigure}{.3\textwidth}
    	\centering
    	\includegraphics[width=\textwidth]{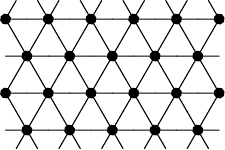}
    	\caption{\centering}
    	\label{fig:modellattice}
    \end{subfigure} \hfill
    \begin{subfigure}{.3\textwidth}
    	\centering
    	\includegraphics[width=\textwidth]{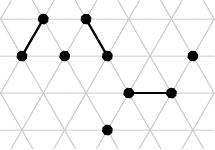}
    	\caption{\centering}
    	\label{fig:modelparticles}
    \end{subfigure} \hfill
    \begin{subfigure}{.3\textwidth}
    	\centering
    	\includegraphics[width=\textwidth]{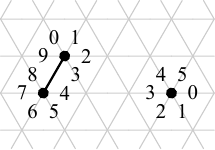}
    	\caption{\centering}
    	\label{fig:modellabels}
    \end{subfigure}
    \caption{The Canonical Amoebot Model.
    (a) A section of the triangular lattice $\Gtri$ used in the geometric variant; nodes of $V$ are shown as black circles and edges of $E$ are shown as black lines.
    (b) Expanded and contracted amoebots; $\Gtri$ is shown in gray, and amoebots are shown as black circles.
    Amoebots with a black line between their nodes are expanded.
    (c) Two amoebots that agree on their chirality but not on their direction, using different offsets for their clockwise-increasing port labels.}
    \label{fig:model}
\end{figure}

An amoebot has two \textit{shapes}: \textsc{contracted}, meaning it occupies a single node in $V$, or \textsc{expanded}, meaning it occupies a pair of adjacent nodes in $V$ (\figtext~\ref{fig:modelparticles}).
For a contracted amoebot, the unique node it occupies is considered its \textit{head}; for an expanded amoebot, the node it has most recently come to occupy (due to movement) is considered its head and the other is its \textit{tail}.
Each amoebot keeps a collection of ports---one for each edge incident to the node(s) it occupies---that are labeled consecutively according to its own local, persistent \textit{orientation}.
For any space variant where $G$ is a planar graph (i.e., those that can be thought of as ``two-dimensional''), an amoebot's orientation depends on its \textit{direction}---i.e., which incident edge it perceives as ``north''---and its \textit{chirality}, or sense of clockwise and counter-clockwise rotation.
Different variants may assume that amoebots share one, both, or neither of their directions and chiralities in common (see Table~\ref{tab:variants}); \figtext~\ref{fig:modellabels} gives an example of the \textit{common chirality} variant where amoebots share a sense of clockwise rotation but have different directions.

Two amoebots occupying adjacent nodes are said to be \textit{neighbors}.
Although each amoebot is \textit{anonymous}, lacking a unique identifier, we assume an amoebot can locally identify its neighbors using their port labels.
In particular, we assume that amoebots $A$ and $B$ connected via ports $p_A$ and $p_B$ each know one another's orientations and labels for $p_A$ and $p_B$.
If $A$ is expanded, we also assume $B$ knows the direction $A$ is expanded in with respect to its own local direction, and vice versa.
This is sufficient for an amoebot to reconstruct which adjacent nodes are occupied by the same neighbor and to translate its local orientation into those of its neighbors, but is not so strong so as to collapse the hierarchy of orientation assumptions.
More details on an amoebot's anatomy are given in Section~\ref{subsec:modelanatomy}.

An amoebot's functionality is partitioned between a higher-level \textit{application layer} and a lower-level \textit{system layer}.
Algorithms controlling an amoebot's behavior are designed from the perspective of the application layer.
The system layer is responsible for an amoebot's core functions and exposes a limited programming interface of \textit{operations} to the application layer that can be used in amoebot algorithms.
The operations are defined in Section~\ref{subsec:modeloperations} and their organization into algorithms is described in Section~\ref{subsec:modelalgs}.
Throughout, we assume amoebots execute their algorithms \textit{reliably}, without crash or Byzantine faults.\footnote{As we discuss in Section~\ref{sec:discuss}, designing \textit{fault tolerant} algorithms is an important research direction for programmable matter. We leave the formalization of different fault models under the canonical amoebot model for future work.}
Although theoretical models usually abstract away from a system layer, we describe it in detail to justify the interface to the application layer since amoebots are not a standard computing platform.
In future publications, one may abstract from the system layer and focus only on the interface.

\subsection{Amoebot Anatomy} \label{subsec:modelanatomy}

Each amoebot has memory whose size is a model variant; the standard assumption is \textit{constant-size} memory.
An amoebot's memory consists of two parts: a persistent \textit{public memory} that is read-writeable by the system layer but only accessible to the application layer via communication operations (see Section~\ref{subsubsec:operationscomms}), and a volatile \textit{private memory} that is inaccessible to the system layer but read-writable by the application layer.
The public memory of an amoebot $A$ contains (\textit{i}) the shape of $A$, denoted $A.\shape \in \{\textsc{contracted}, \textsc{expanded}\}$, (\textit{ii}) the lock state of $A$, denoted $A.\lock$ (see Section~\ref{subsubsec:operationsconcurrency}), and (\textit{iii}) any variables used in the algorithm being run by the application layer.
An amoebot's private memory can be modified by the application layer as needed.

Neighboring amoebots (i.e., those occupying adjacent nodes) form \textit{connections} via their ports facing each other.
An amoebot's system layer receives instantaneous feedback whenever a new connection is formed or an existing connection is broken.
Communication between connected neighbors is achieved via \textit{message passing}.
To facilitate message passing communication, each of an amoebot's ports has a FIFO \textit{outgoing message buffer} managed by the system layer that can store up to a fixed (constant) number of messages waiting to be sent to the neighbor incident to the corresponding port.
If two neighbors disconnect due to some movement, their system layers immediately flush the corresponding message buffers of any pending messages.
Otherwise, we assume that any pending message is sent to the connected neighbor in FIFO order in finite time.
Incoming messages are processed as they are received.

\subsection{Amoebot Operations} \label{subsec:modeloperations}

Operations provide the application layer with a programming interface for controlling the amoebot's behavior; the application layer calls operations and the system layer executes them.
We assume the execution of an operation is \textit{blocking} for the application layer; that is, the application layer can only call one operation at a time.
We formally define the communication, movement, and concurrency control operations and their execution details in Sections~\ref{subsubsec:operationscomms}--\ref{subsubsec:operationsconcurrency}; see Table~\ref{tab:operations} for a summary and Appendix~\ref{app:pseudocode} for complete distributed pseudocode.
As we will show in Section~\ref{subsubsec:operationscomplexity}, each operation is carefully designed so that any operation execution terminates
in finite time (Observation~\ref{obs:operationstime}) and, at any time, there are at most a constant number of messages being sent or received between any pair of neighboring amoebots as a result of any set of concurrent operation executions (Observation~\ref{obs:operationsspace}).
Combined with the blocking and reliability assumptions, these design principles prohibit outgoing message buffer overflow and deadlocks in operation executions.

\begin{table}[t]
    \centering
    \caption{Summary of operations exposed by an amoebot's system layer to its application layer.}
    \label{tab:operations}
    \begin{tabular}{lp{107mm}}
        \toprule
        \textbf{Operation} & \textbf{Return Value on Success} \\
        \midrule
        $\Connected(p)$ & \true\ iff a neighboring amoebot is connected via port $p$ \\
        $\Connected()$ & $[c_0, \ldots, c_{k-1}] \in \{N_1, \ldots, N_k, \false\}^k$ where $c_p = N_i$ if $N_i$ is the locally identified neighbor connected via port $p$ and $c_p = \false$ otherwise \\
        $\Read(p, x)$ & The value of $x$ in the public memory of this amoebot if $p = \bot$ or of the neighbor incident to port $p$ otherwise \\
        $\Write(p, x, x_{val})$ & Confirmation that the value of $x$ was updated to $x_{val}$ in the public memory of this amoebot if $p = \bot$ or of the neighbor incident to port $p$ otherwise \\
        \midrule
        $\Contract(v)$ & Confirmation of the contraction out of node $v \in \{\textsc{head}, \textsc{tail}\}$ \\
        $\Expand(p)$ & Confirmation of the expansion into the node incident to port $p$ \\
        $\Pull(p)$ & Confirmation of the pull handover with the neighbor incident to port $p$ \\
        $\Push(p)$ & Confirmation of the push handover with the neighbor incident to port $p$ \\
        \midrule
        $\Lock()$ & Port labels corresponding to the amoebots that were locked \\
        $\Unlock(\mathcal{L})$ & Confirmation that the amoebots of $\mathcal{L}$ were unlocked \\
        \bottomrule
    \end{tabular}
\end{table}

\subsubsection{Communication Operations} \label{subsubsec:operationscomms}

An amoebot checks for the presence of neighbors using the \Connected\ operations and exchanges information with its neighbors using the \Read\ and \Write\ operations.
When the application layer calls $\Connected(p)$, the system layer simply returns \true\ if there is a neighbor connected via port $p$ and \false\ otherwise.
The application layer may instead call $\Connected()$ to obtain a full snapshot of its current port connectivity.
Specifically, the system layer returns an array $[c_0, \ldots, c_{k-1}]$ mapping the amoebot's $k$ ports to local identifiers for its neighbors, of which there can be at most $k$.
If there is no neighbor connected via port $p$, then $c_p = \false$; otherwise, $c_p = N_i$ where $N_i \in \{N_1, \ldots, N_k\}$ locally identifies the neighbor connected via port $p$ (see \figtext~\ref{fig:operationsconn}).
Note that, depending on amoebots' shapes and the geometry of the space variant, multiple ports may connect to the same neighbor $N_i$.

\begin{figure}[t]
    \centering
    \begin{subfigure}{0.45\textwidth}
    	\centering
    	\includegraphics[scale=1.4]{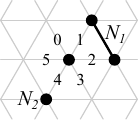}
    	\caption{\centering $[\false, N_1, N_1, \false, N_2, \false]$}
    	\label{fig:operationsconna}
    \end{subfigure}
    \hfill
    \begin{subfigure}{0.53\textwidth}
    	\centering
    	\includegraphics[scale=1.4]{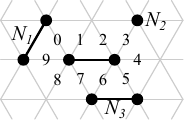}
    	\caption{\centering $[N_1, \false, \false, N_2, \false, N_3, N_3, N_3, \false, N_1]$}
    	\label{fig:operationsconnc}
    \end{subfigure}
    \caption{Different neighborhood configurations for an amoebot $A$ and their corresponding $\Connected()$ return values.
    The ports of $A$ are shown with their labels and neighboring amoebots are shown with their local identifiers according to $A$.}
    \label{fig:operationsconn}
\end{figure}

The application layer calls $\Read(p, x)$ to issue a request to read the value of a variable $x$ in the public memory of the neighbor connected via port $p$.
Analogously, the application layer calls $\Write(p, x, x_{val})$ to issue a request to update the value of a variable $x$ in the public memory of the neighbor connected via port $p$ to a new value $x_{val}$.
If $p = \bot$, an amoebot's own public memory is accessed instead of a neighbor's.

\begin{figure}[t]
    \centering
    \begin{subfigure}{0.53\textwidth}
    	\centering
    	\includegraphics[width=\textwidth]{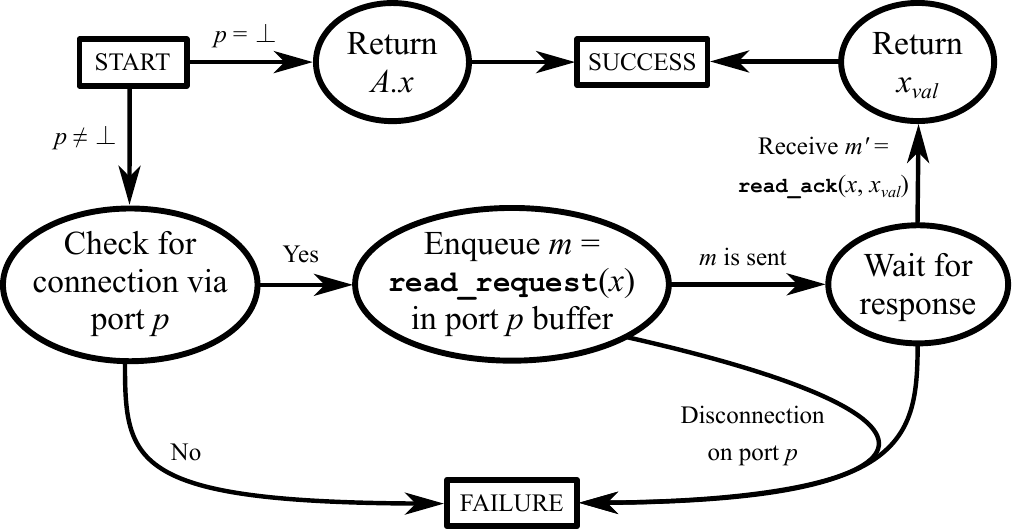}
    	\caption{\centering $\Read(p, x)$}
    	\label{fig:operationsread}
    \end{subfigure}
    \hfill
    \begin{subfigure}{0.46\textwidth}
    	\centering
    	\includegraphics[width=\textwidth]{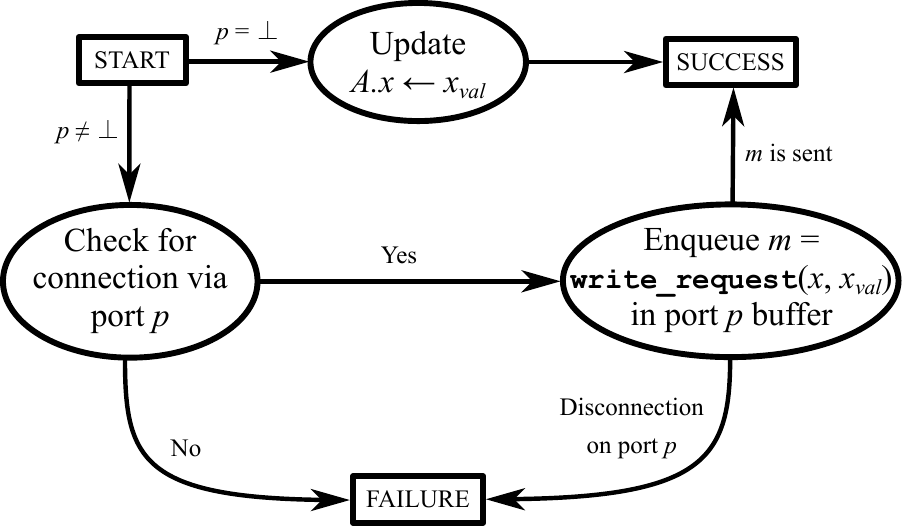}
    	\caption{\centering $\Write(p, x, x_{val})$}
    	\label{fig:operationswrite}
    \end{subfigure}
    \caption{Execution flows of the \Read\ and \Write\ operations for the calling amoebot $A$.}
    \label{fig:operationscomms}
\end{figure}

Suppose that the application layer of an amoebot $A$ calls $\Read(p, x)$, illustrated in \figtext~\ref{fig:operationsread}.
If $p = \bot$, the system layer simply returns the value of $x$ in the public memory of $A$ to the application layer and this \Read\ succeeds.
Otherwise, the system layer checks if there is a neighbor connected via port $p$: if so, the system layer enqueues $m = \texttt{read\_request}(x)$ in the message buffer on $p$; otherwise, this \Read\ fails.
Let $B$ be the neighbor connected to $A$ via port $p$ and let $p'$ be its corresponding port.
Eventually, $m$ is sent in FIFO order and the system layer of $B$ receives it, prompting it to access variable $x$ with value $x_{val}$ in its public memory and enqueue $m' = \texttt{read\_ack}(x, x_{val})$ in the message buffer on $p'$.
Message $m'$ is eventually sent in FIFO order by $B$ and received by the system layer of $A$, prompting it to unpack $x_{val}$ and return it to the application layer, successfully completing this \Read.
If $A$ and $B$ are disconnected (i.e., due to a movement) any time after $A$ enqueues message $m$ but before $A$ receives message $m'$, this \Read\ fails.

A $\Write(p, x, x_{val})$ operation is executed analogously, though it does not need to wait for an acknowledgement after its write request is sent (see \figtext~\ref{fig:operationswrite}).

\subsubsection{Movement Operations} \label{subsubsec:operationsmoves}

The application layer can direct the system layer to initiate movements using the four movement operations \Contract, \Expand, \Pull, and \Push.
An expanded amoebot can \Contract\ into either node it occupies; a contracted amoebot can \Expand\ into an unoccupied adjacent node.
Neighboring amoebots can coordinate their movements in a \textit{handover}, which can occur in one of two ways.
A contracted amoebot $A$ can \Push\ an expanded neighbor $B$ by expanding into a node occupied by $B$, forcing it to contract.
Alternatively, an expanded amoebot $B$ can \Pull\ a contracted neighbor $A$ by contracting, forcing $A$ to expand into the neighbor it is vacating.

\subparagraph*{Contract.}

Suppose that the application layer of an amoebot $A$ calls $\Contract(v)$, where $v \in \{\textsc{head}, \textsc{tail}\}$ (see \figtext~\ref{fig:operationscontract}).
The system layer of $A$ first determines if this contraction is valid: if $A.\shape \neq \textsc{expanded}$ or $A$ is currently involved in a handover, this \Contract\ fails.
Otherwise, the system layer releases all connections to neighboring amoebots via ports on node $v$ and begins contracting out of node $v$.
Once the contraction completes, the system layer updates $A.\shape \gets \textsc{contracted}$, successfully completing this \Contract.

\begin{figure}[t]
    \centering
    \begin{subfigure}{0.41\textwidth}
    	\centering
    	\includegraphics[width=\textwidth]{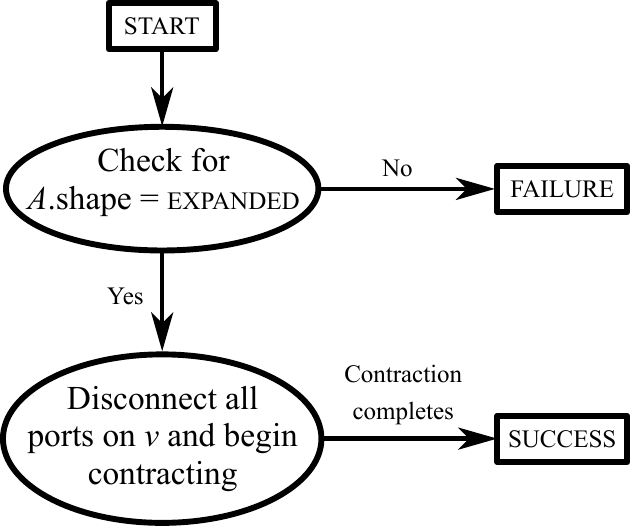}
    	\caption{\centering $\Contract(v)$}
    	\label{fig:operationscontract}
    \end{subfigure}
    \hfill
    \begin{subfigure}{0.56\textwidth}
    	\centering
    	\includegraphics[width=\textwidth]{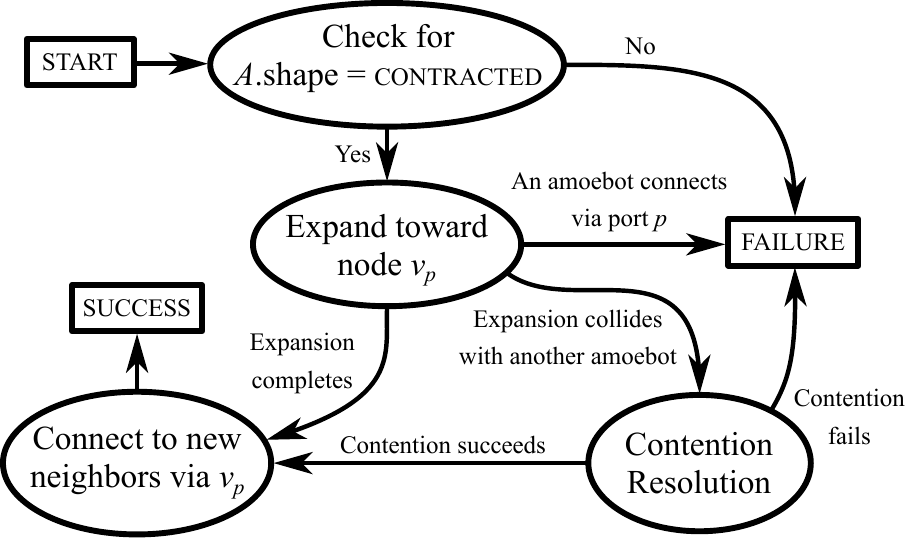}
    	\caption{\centering $\Expand(p)$}
    	\label{fig:operationsexpand}
    \end{subfigure} \\ \medskip
    \begin{subfigure}{\textwidth}
    	\centering
    	\includegraphics[scale=0.52]{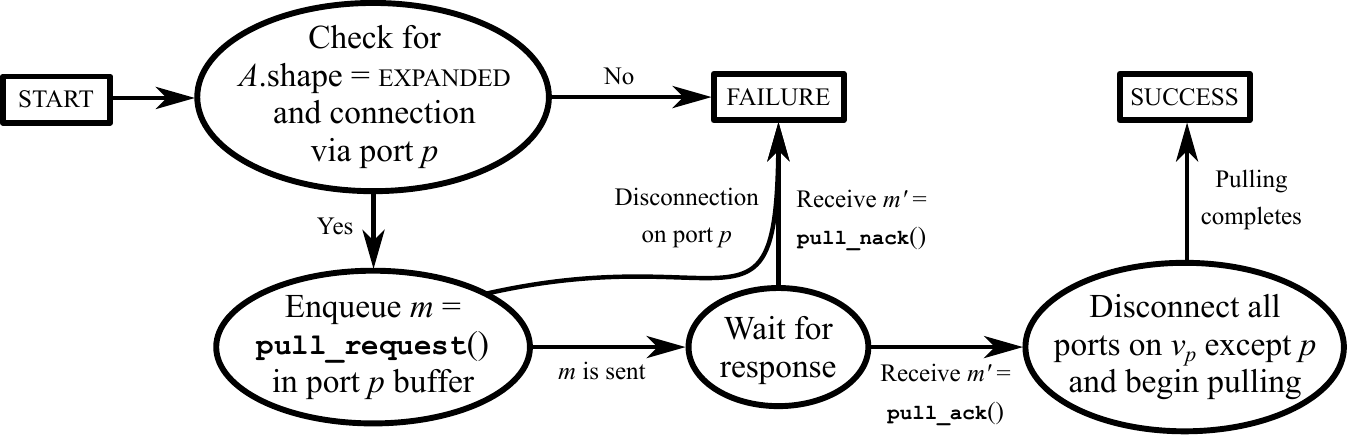}
    	\caption{\centering $\Pull(p)$}
    	\label{fig:operationspull}
    \end{subfigure} \\
    \medskip
    \begin{subfigure}{\textwidth}
    	\centering
    	\includegraphics[scale=0.52]{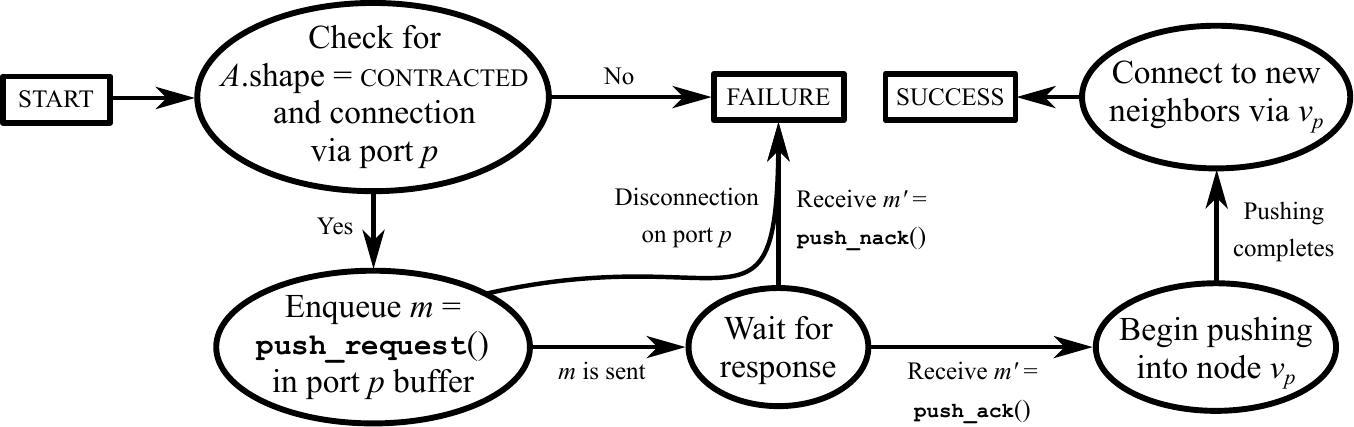}
    	\caption{\centering $\Push(p)$}
    	\label{fig:operationspush}
    \end{subfigure}
    \caption{Execution flows of the movement operations for the calling amoebot $A$.}
    \label{fig:operationsmoves}
\end{figure}

\subparagraph*{Expand.}

Suppose an amoebot $A$ calls $\Expand(p)$ for one of its ports $p$ (see \figtext~\ref{fig:operationsexpand}); let $v_p$ denote the node $A$ is expanding into.
If $A.\shape \neq \textsc{contracted}$, $A$ is already involved in a handover, or $v_p$ is already occupied by another amoebot, this \Expand\ fails.
Otherwise, $A$ begins its expansion into node $v_p$.
Once this expansion completes, the system layer establishes connections with all neighbors adjacent to $v_p$ and updates $A.\shape \gets \textsc{expanded}$, successfully completing this \Expand.
However, $A$ may \textit{collide} with other amoebots while expanding into $v_p$.
We assume that the system layer can detect when a collision has occurred and, on collision, performs \textit{contention resolution} such that exactly one contending amoebot succeeds in completing its expansion into $v_p$ while all others fail within finite time.
We abstract away from the details of this contention resolution mechanism for the sake of clarity, but give one possible implementation in Appendix~\ref{app:expandcontend} to demonstrate its feasibility.

\subparagraph*{Pull and Push.}

Suppose an amoebot $A$ calls $\Pull(p)$ for one of its ports $p$ (see \figtext~\ref{fig:operationspull}); let $v_p$ denote the node $A$ intends to vacate in this pull handover.
If $A.\shape \neq \textsc{expanded}$, $A$ is already involved in a handover, or $A$ is not connected to a neighbor via port $p$, this \Pull\ fails.
Otherwise, the system layer of $A$ enqueues $m = \texttt{pull\_request}()$ in the message buffer on port $p$.
Let $B$ be the neighbor connected to $A$ via port $p$.
Eventually, message $m$ is sent in FIFO order and the system layer of $B$ receives it.
If $B$ is not involved in another movement and $B.\shape = \textsc{contracted}$, its system layer prepares message $m' = \texttt{pull\_ack}()$; otherwise, it sets $m' = \texttt{pull\_nack}()$.
In either case, the system layer of $B$ enqueues $m'$ in the message buffer on its port facing $A$.
If $A$ and $B$ are disconnected any time after $A$ enqueues message $m$ but before $A$ receives message $m'$, this \Pull\ fails; otherwise, message $m'$ is eventually sent in FIFO order by $B$ and received by the system layer of $A$.
If $m' = \texttt{pull\_nack}()$, this \Pull\ fails.
Otherwise, if $m' = \texttt{pull\_ack}()$, $A$ disconnects from all ports on node $v_p$ (except for $p$) and $A$ and $B$ begin their coordinated handover of node $v_p$.
When $A$ completes its contraction, it updates $A.\shape \gets \textsc{contracted}$; analogously, when $B$ completes its expansion, it updates $B.\shape \gets \textsc{expanded}$ and establishes connections to its new neighbors adjacent to node $v_p$.
This successfully completes this \Pull.

A $\Push(p)$ operation is executed analogously (see \figtext~\ref{fig:operationspush}).

\subsubsection{Concurrency Control Operations} \label{subsubsec:operationsconcurrency}

The amoebot model's concurrency control operations \Lock\ and \Unlock\ encapsulate a variant of the classical mutual exclusion problem in which an amoebot attempts to gain exclusive control over itself and the amoebots in its neighborhood.
Achieving this behavior in the system layer's setting of asynchronous message passing with dynamic neighbor connections is non-trivial.
Daymude et al.\ recently solved this problem in their algorithm for ``local mutual exclusion''~\cite{Daymude2022-localmutual} where nodes in a dynamic graph seek to acquire exclusive locks over themselves and their ``persistent'' neighbors, i.e., nodes that remain connected to them over the time interval of the lock request.
Here, we focus on the properties that \Lock\ and \Unlock\ must satisfy and refer the interested reader to~\cite{Daymude2022-localmutual} for one possible implementation.

Each amoebot $A$ stores a variable $A.\lock \in \{\bot, -1, \ldots, \Delta - 1\}$, where $\Delta$ is the maximum number of neighbors an amoebot can have based on the assumed space variant, that is equal to $\bot$ if $A$ is unlocked, $-1$ if $A$ has locked itself, and $i \in \{0, \ldots, \Delta - 1\}$ if $A$ is locked by its neighbor connected via port $i$.
An amoebot $A$ calls $\Lock()$ to issue a lock request to itself and the neighbors it has at the start of this execution.
To succeed, this \Lock\ operation must lock $A$ and every \textit{persistent} neighbor of $A$ that remained connected to $A$ throughout its \Lock\ execution, setting their \lock\ variables accordingly.
On success, the \Lock\ operation returns the \textit{lock set} $\mathcal{L}$ of port labels corresponding to the amoebots $A$ has locked.
We assume that a \Lock\ operation either succeeds or fails in finite time.
An amoebot calls $\Unlock(\mathcal{L}')$ to release its locks on itself or any neighbors connected via port labels in $\mathcal{L}'$, resetting their \lock\ variables to $\bot$; this operation always succeeds.

Any implementation of these operations must ensure that any set of \Lock\ and \Unlock\ executions satisfies: (\textit{i}) \textit{mutual exclusion}, meaning that the amoebots' lock sets must be disjoint at all times, and (\textit{ii}) \textit{deadlock freedom}, meaning that if a \Lock\ operation is initiated at time $t$, then some \Lock\ execution succeeds after time $t$.
The local mutual exclusion algorithm of~\cite{Daymude2022-localmutual} satisfies both of these properties; in fact, it even satisfies the stronger property of \textit{lockout freedom}, guaranteeing that every \Lock\ execution eventually succeeds.

\subsubsection{Operation Time and Space Complexity} \label{subsubsec:operationscomplexity}

With the communication, movement, and concurrency control operations defined, we now briefly characterize their time and space complexity.
Recall that we assume amoebots execute reliably, without crash or Byzantine faults.
The \Connected\ operations are effectively instantaneous as the system layer has immediate access to the physical information about its port connectivity; moreover, these operations do not involve any messages.
The complexity of the \Lock\ and \Unlock\ operations depend on their implementation; e.g., the local mutual exclusion algorithm of~\cite{Daymude2022-localmutual} guarantees termination in finite time and that at most two messages are in transit between any pair of neighbors at any time.
For the remaining operations, recall from Sections~\ref{subsec:modelanatomy} and~\ref{subsubsec:operationsmoves} that we assume (\textit{i}) messages pending in an outgoing message buffer are each sent to the connected neighbor in FIFO order in finite time and are immediately flushed on disconnection, and (\textit{ii}) every physical movement completes in finite time.
We first consider the execution of each operation independently.
\begin{itemize}
    \item A \Read\ operation by an amoebot $A$ from its own public memory is immediate and does not involve any messages.
    If instead $A$ reads from the public memory of a neighbor $B$, at most two messages are used---a \texttt{read\_request} sent from $A$ to $B$ followed by a \texttt{read\_ack} sent from $B$ to $A$---that are each delivered in finite time by (\textit{i}).
    Successful termination occurs when $A$ receives \texttt{read\_ack} while a disconnection between $A$ and $B$ results in immediate failure with any related messages being flushed.

    \item A \Write\ operation by an amoebot $A$ to its own public memory is immediate and does not involve any messages.
    If instead $A$ writes to the public memory of a neighbor $B$, one message is used: a \texttt{write\_request} sent from $A$ to $B$ that is delivered in finite time by (\textit{i}).
    Successful termination occurs when $A$ sends \texttt{write\_request} while a disconnection between $A$ and $B$ results in immediate failure with any related messages being flushed.

    \item A \Contract\ operation does not involve any messages.
    It either fails at its start or succeeds after releasing its connections and completing its contraction, which must occur in finite time by (\textit{ii}).

    \item An \Expand\ operation does not involve any messages.
    It either fails at its start or is able to begin its expansion.
    If there are no collisions, (\textit{ii}) guarantees the expansion completes in finite time; otherwise, the contention resolution mechanism guarantees that exactly one contending amoebot succeeds while all others fail within finite time.

    \item A \Pull\ operation by an amoebot $A$ with a neighbor $B$ involves at most two messages---a \texttt{push\_request} sent from $A$ to $B$ and either a \texttt{push\_ack} or a \texttt{push\_nack} sent from $B$ to $A$---that are delivered in finite time by (\textit{i}).
    The contraction of $A$ and expansion of $B$ must complete in finite time by (\textit{ii}).
    Any failures can only happen earlier.

    \item A \Push\ operation is symmetric to a \Pull\ and thus satisfies the same properties.
\end{itemize}
This immediately reveals the following observation regarding time complexity.

\begin{observation} \label{obs:operationstime}
    Any execution of an operation in the canonical amoebot model terminates---either successfully or in failure---in finite time.
\end{observation}

By the blocking assumption, the application layer of each amoebot can execute at most one operation per time.
The above discussion shows that each operation has at most one message in transit per time.
Thus, there can be at most two messages in any outgoing message buffer at any time, e.g., in the situation where an amoebot $A$ is executing an operation that involves sending a message $m_1$ to a neighbor $B$ while $B$ is concurrently executing an operation that requires $A$ to send a message $m_2$ back to $B$ in response to some prior message sent from $B$.
This yields the following observation, demonstrating that constant-size buffers suffice to avoid overflow.

\begin{observation} \label{obs:operationsspace}
    At any time, there are at most a constant number of messages in transit (i.e., being sent or received) between any pair of neighboring amoebots as a result of any set of operation executions.
\end{observation}

\subsection{Amoebot Actions, Algorithms and Executions} \label{subsec:modelalgs}

Following the message passing literature, we specify distributed algorithms in the amoebot model as sets of \textit{actions} to be executed by the application layer, each of the form:
\[\langle label\rangle: \langle guard\rangle \to \langle operations\rangle\]
An action's \textit{label} specifies its name.
Its \textit{guard} is a Boolean predicate determining whether an amoebot $A$ can execute it based on the connected ports of $A$---i.e., which nodes adjacent to $A$ are (un)occupied---and information from the public memories of $A$ and its neighbors.
An action is \textit{enabled} for an amoebot $A$ if its guard is true for $A$, and an amoebot is \textit{enabled} if it has at least one enabled action.
An action's \textit{operations} specify the finite sequence of operations and computation in private memory to perform if this action is executed.
The control flow of this computation may optionally include \textit{randomization} to generate random values and \textit{error handling} to address any operation executions resulting in failure.

Each amoebot executes its own algorithm instance independently and reliably, without crash or Byzantine faults.
An amoebot is said to be \textit{active} if its application layer is executing an action and is \textit{inactive} otherwise.
An amoebot can begin executing an action if and only if it is inactive; i.e., an amoebot can execute at most one action at a time.
On becoming active, an amoebot $A$ first evaluates which of its actions $\alpha_i : g_i \to ops_i$ are enabled.
Since each guard $g_i$ is based only on the connected ports of $A$ and the public memories of $A$ and its neighbors, each $g_i$ can be evaluated using the \Connected\ and \Read\ operations.
If no action is enabled, $A$ returns to inactive; otherwise, $A$ chooses an enabled action $\alpha_i$ and executes the operations and private computation specified by $ops_i$.
Recall from Section~\ref{subsec:modeloperations} that each operation is guaranteed to terminate (either successfully or with a failure) in finite time.
Thus, since $A$ is reliable and $ops_i$ consists of a finite sequence of operations and finite computation, each action execution is also guaranteed to terminate in finite time after which $A$ returns to inactive.
An action execution \textit{fails} if any of its operations' executions result in a failure that is not addressed with error handling and \textit{succeeds} otherwise.

\begin{figure}[t]
    \centering
    \begin{subfigure}{.49\textwidth}
        \centering
        \includegraphics[width=\textwidth]{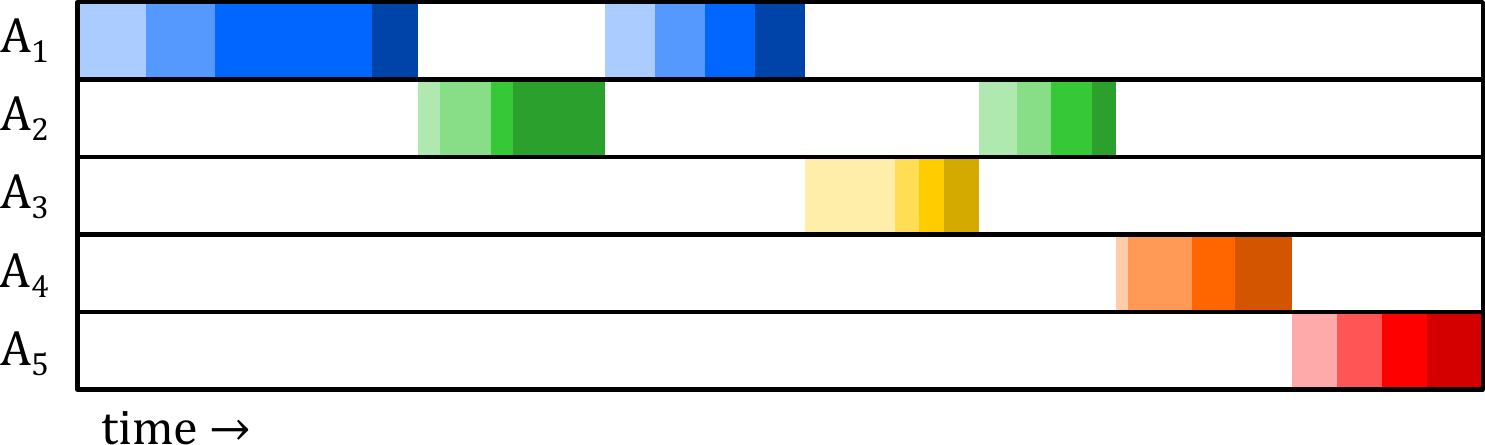}
        \caption{\centering Sequential}
        \label{fig:adversaries:sequential}
    \end{subfigure} \hfill
    \begin{subfigure}{.49\textwidth}
        \centering
        \includegraphics[width=\textwidth]{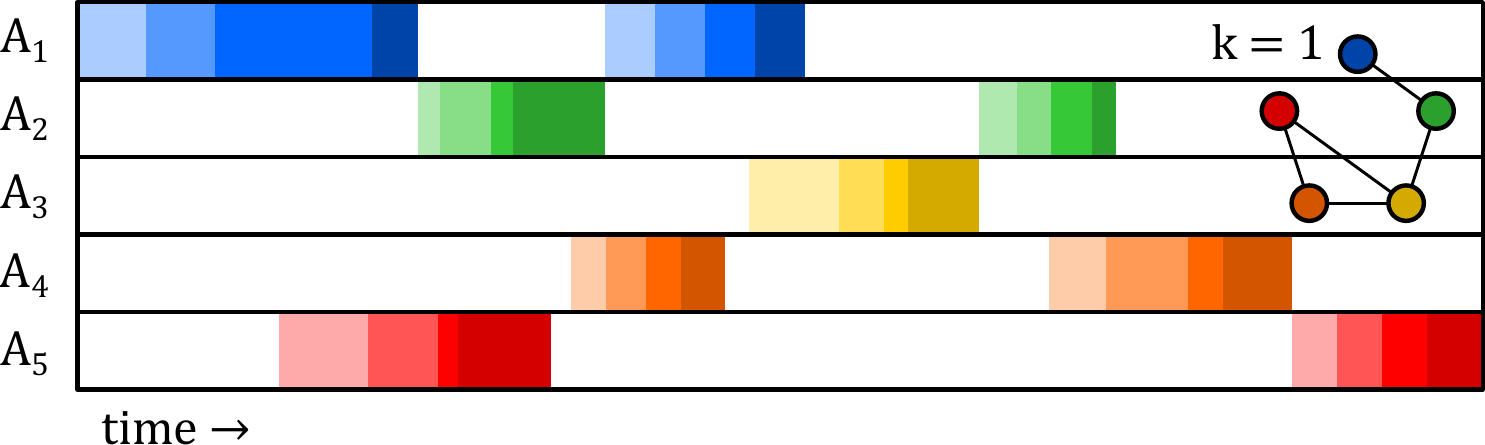}
        \caption{\centering $k$-Isolated (for $k = 1$)}
        \label{fig:adversaries:kisolated}
    \end{subfigure} \\ \bigskip
    \begin{subfigure}{.49\textwidth}
        \centering
        \includegraphics[width=\textwidth]{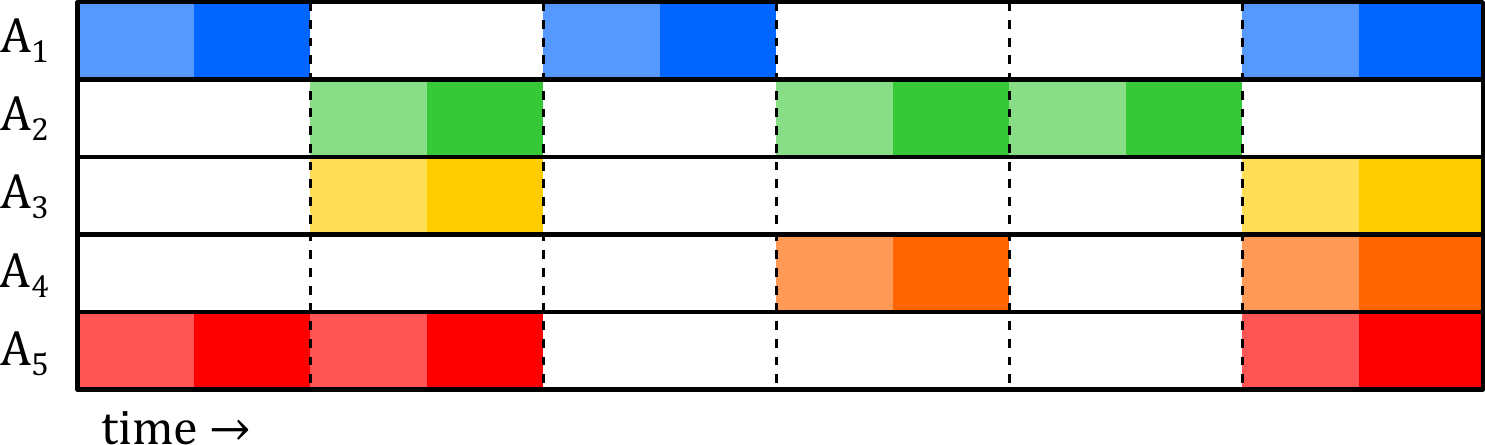}
        \caption{\centering Synchronous}
        \label{fig:adversaries:synchronous}
    \end{subfigure} \hfill
    \begin{subfigure}{.49\textwidth}
        \centering
        \includegraphics[width=\textwidth]{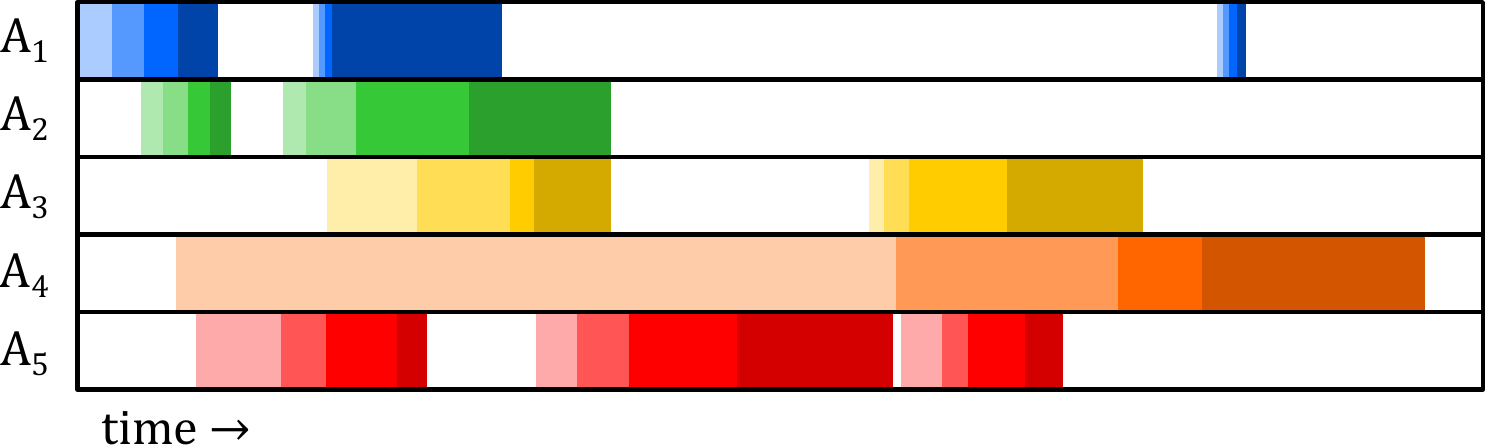}
        \caption{\centering Asynchronous}
        \label{fig:adversaries:asynchronous}
    \end{subfigure}
    \caption{Adversary concurrency variants.
    Amoebots are shown in rows and their actions over time are shown as colored boxes, subdivided into operations (gradient of colors).}
    \label{fig:adversaries}
\end{figure}

As is standard in the distributed computing literature (see, e.g.,~\cite{Altisen2019-introductiondistributed}), we assume an \textit{adversary} (or \textit{daemon}) controls the timing of amoebot activations, the choice of enabled actions to execute, and the timing of action executions.
The power of an adversary is determined by its \textit{concurrency} and \textit{fairness}.
We distinguish between four concurrency variants: \textit{sequential}, in which at most one amoebot can be active at a time (\figtext~\ref{fig:adversaries:sequential}); \textit{$k$-isolated}, in which no two amoebots occupying nodes of $G$ within hop distance $k$ can be simultaneously active, but any others can (\figtext~\ref{fig:adversaries:kisolated}); \textit{synchronous}, in which time is discretized into ``steps'' and in each step any set of amoebots can simultaneously execute one action each (\figtext~\ref{fig:adversaries:synchronous}); and \textit{asynchronous}, in which any set of amoebots can be simultaneously active (\figtext~\ref{fig:adversaries:asynchronous}).
For synchronous concurrency, we further assume that each step is partitioned into an \textit{evaluation phase} when all active amoebots evaluate their guards followed by an \textit{execution phase} when all active amoebots with enabled actions execute the corresponding operations.
Fairness restricts how often the adversary must activate enabled amoebots.
We distinguish between three fairness variants: \textit{strongly fair}, in which every amoebot that is enabled infinitely often is activated infinitely often; \textit{weakly fair}, in which every continuously enabled amoebot is eventually activated; and \textit{unfair}, in which the adversary may activate any enabled amoebot.
An algorithm execution is said to \textit{terminate} if eventually all amoebots are inactive and disabled; note that since an amoebot can only become enabled based on some other amoebot's action, termination is permanent.

We evaluate an amoebot algorithm's time complexity in terms of \textit{rounds}, which informally represent the time for the slowest continuously enabled amoebot to execute a single action.
Let $t_i$ denote the time at which round $i \in \{0, 1, 2, \ldots\}$ starts, where $t_0 = 0$, and let $\mathcal{E}_i$ denote the set of amoebots that are enabled or already executing an action at time $t_i$.
Round $i$ completes at the earliest time $t_{i+1} > t_i$ by which every amoebot in $\mathcal{E}_i$ either completed an action execution or became disabled at some time in $(t_i, t_{i+1}]$.
Depending on the adversary's concurrency, action executions may span more than one round.

In this paper, we focus on unfair sequential and asynchronous adversaries.
In the sequential setting, there is at most one active amoebot per time; thus, its guard evaluations and subsequent operation executions must be correct.
In the asynchronous setting, however, concurrent movements and memory updates can cause discrepancies between the adversary's instantaneous view of enabled actions and an amoebot's real-time evaluation of its guards, potentially allowing disabled actions to be executed or enabled actions to be skipped.
Moreover, concurrency can cause operations to fail due to conflicts.
We address these issues in two ways, justifying the formulation of algorithms in terms of actions:
In Section~\ref{sec:hexagon}, we present an algorithm whose actions are carefully designed to ensure correct execution under any adversary;
in Section~\ref{sec:framework}, we present a concurrency control framework that uses locks to ensure correct guard evaluation and operation execution even in the asynchronous setting.

\section{Asynchronous Hexagon Formation Without Locks}
\label{sec:hexagon}

We use the \textit{hexagon formation} problem as a concrete case study for algorithm design, pseudocode, and analysis in the canonical amoebot model.
Our \algHex\ algorithm (Algorithm~\ref{alg:hexagon}) assumes geometric space, assorted orientation, and constant-size memory (Table~\ref{tab:variants}) and is formulated in terms of actions as specified in Section~\ref{subsec:modelalgs}.
Our analysis of \algHex\ reveals a set of sufficient conditions for any amoebot algorithm's correctness under an unfair asynchronous adversary: (\textit{i}) correctness under an unfair sequential adversary, (\textit{ii}) enabled actions remaining enabled despite concurrent action executions, and (\textit{iii}) executions of enabled actions remaining successful and unaffected by concurrent action executions.
Any concurrent execution of an algorithm satisfying (\textit{ii}) and (\textit{iii}) can be shown to be serializable, which combined with sequential correctness establishes correctness under an unfair asynchronous adversary, the most general of all possible adversaries.
Notably, we prove that our \algHex\ algorithm satisfies these sufficient conditions without using locks, demonstrating that while locks are useful tools for designing correct amoebot algorithms under concurrent adversaries, they are not always necessary.

\begin{figure}[t]
    \centering
    \begin{subfigure}{.3\textwidth}
        \centering
        \includegraphics[width=\textwidth]{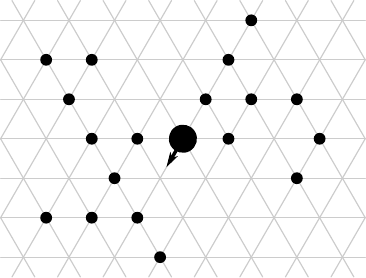}
        \caption{\centering}
        \label{fig:hexformationa}
    \end{subfigure} \hfill
    \begin{subfigure}{.3\textwidth}
        \centering
        \includegraphics[width=\textwidth]{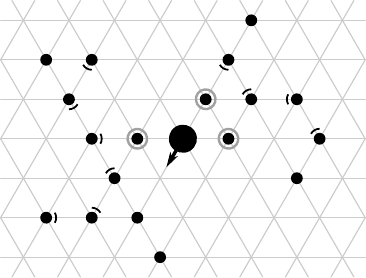}
        \caption{\centering}
        \label{fig:hexformationb}
    \end{subfigure} \hfill
    \begin{subfigure}{.3\textwidth}
        \centering
        \includegraphics[width=\textwidth]{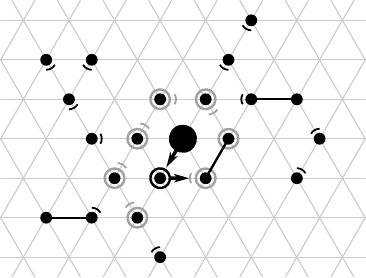}
        \caption{\centering}
        \label{fig:hexformationc}
    \end{subfigure} \\ \bigskip
    \begin{subfigure}{.3\textwidth}
        \centering
        \includegraphics[width=\textwidth]{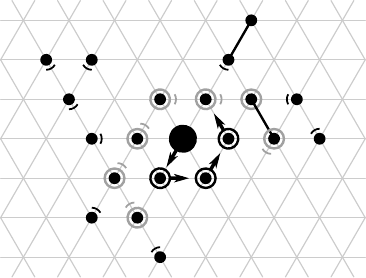}
        \caption{\centering}
        \label{fig:hexformationd}
    \end{subfigure} \hfill
    \begin{subfigure}{.3\textwidth}
        \centering
        \includegraphics[width=\textwidth]{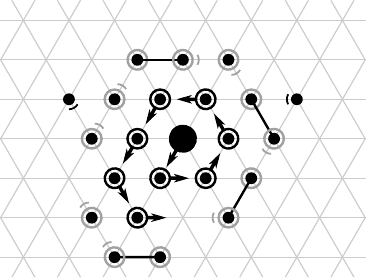}
        \caption{\centering}
        \label{fig:hexformatione}
    \end{subfigure} \hfill
    \begin{subfigure}{.3\textwidth}
        \centering
        \includegraphics[width=\textwidth]{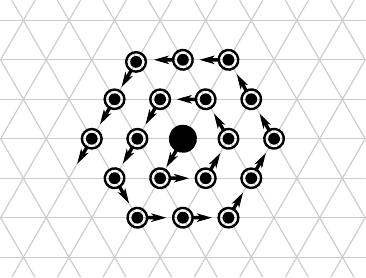}
        \caption{\centering}
        \label{fig:hexformationf}
    \end{subfigure}
    \caption{An example run of \algHex\ with $19$ amoebots.
    (a) All amoebots are initially idle (black dots), with the exception of a unique seed amoebot (large black dot).
    (b) Amoebots adjacent to the seed become roots (gray circles), and followers form parent-child relationships (black arcs) with roots and other followers.
    (c)--(f) Roots traverse the forming hexagon clockwise, becoming retired (black circles) when reaching the position marked by the last retired amoebot.}
    \label{fig:hexformation}
\end{figure}

The hexagon formation problem tasks an arbitrary, connected system of initially contracted amoebots with forming a regular hexagon (or as close to one as possible, given the number of amoebots in the system).
We assume that there is a \textit{unique seed amoebot} in the system and all other amoebots are initially \textit{idle}; note that the seed amoebot immediately collapses the hierarchy of orientation assumptions since it can impose its own local orientation on the rest of the system.
Following the sequential algorithm given by Derakhshandeh et al.~\cite{Daymude2019-computingprogrammable,Derakhshandeh2015-algorithmicframework}, the basic idea of our \algHex\ algorithm is to form a hexagon by extending a spiral of amoebots counter-clockwise from the seed (see \figtext~\ref{fig:hexformation}).

In addition to the shape variable assumed by the amoebot model, each amoebot $A$ keeps variables $A.\state \in \{\textsc{seed}, \textsc{idle}, \textsc{follower}, \textsc{root}, \textsc{retired}\}$, $A.\parent \in \{\nil, 0, \ldots, 9\}$, and $A.\dir \in \{\nil, 0, \ldots, 9\}$ in public memory.
The amoebot system first self-organizes as a spanning forest rooted at the seed amoebot using their \parent\ ports.
Follower amoebots follow their parents until reaching the surface of retired amoebots that have already found their place in the hexagon.
They then become roots, traversing the surface of retired amoebots clockwise.
Once they connect to a retired amoebot's \dir\ port, they also retire and set their \dir\ port to the next position of the hexagon.
Algorithm~\ref{alg:hexagon} describes \algHex\ in terms of actions.
W.l.o.g., we assume that if multiple actions are enabled for an amoebot, the enabled action with smallest index is executed.\footnote{Observe that any amoebot algorithm could directly implement this assumption by replacing each guard $g_i$ of action $\alpha_i$ with the guard $g_i \wedge \bigwedge_{j=1}^{i-1} (\neg g_j)$.}
In action guards, we use $N(A)$ to denote the neighbors of amoebot $A$ and say that an amoebot $A$ has a \textit{tail-child} $B$ if $B$ is connected to the tail of $A$ via port $B.\parent$.

\begin{algorithm}[t]
    \caption{\algHex\ for Amoebot $A$} \label{alg:hexagon}
    \begin{algorithmic}[1]
        \State $\alpha_1 : (A.\state \in \{\textsc{idle}, \textsc{follower}\}) \wedge (\exists B \in N(A) :  B.\state \in \{\textsc{seed}, \textsc{retired}\}) \to$
        \Indent
            \State \Write$(\bot, \parent, \nil)$.
            \State \Write$(\bot, \state, \textsc{root})$.
            \State \Write$(\bot, \dir, \Call{GetNextDir}{\text{counter-clockwise}})$. \Comment{See Algorithm~\ref{alg:hexagonhelpers}.}
        \EndIndent
        \State $\alpha_2 : (A.\state = \textsc{idle}) \wedge (\exists B \in N(A) : B.\state \in \{\textsc{follower}, \textsc{root}\}) \to$
        \Indent
            \State Find a port $p$ for which $\Connected(p) = \true$ and $\Read(p, \state) \in \{\textsc{follower}, \textsc{root}\}$.
            \State \Write$(\bot, \parent, p)$.
            \State \Write$(\bot, \state, \textsc{follower})$.
        \EndIndent
        \State $\alpha_3 : (A.\shape = \textsc{contracted}) \wedge (A.\state = \textsc{root}) \wedge (\forall B \in N(A) : B.\state \neq \textsc{idle})$ \\
        \hspace{.6cm} $\wedge~(\exists B \in N(A) : (B.\state \in \{\textsc{seed}, \textsc{retired}\}) \wedge (B.\dir$ is connected to $A)) \to$
        \Indent
            \State \Write$(\bot, \dir, \Call{GetNextDir}{\text{clockwise}})$.
            \State \Write$(\bot, \state, \textsc{retired})$.
        \EndIndent
        \State $\alpha_4 : (A.\shape = \textsc{contracted}) \wedge (A.\state = \textsc{root}) \wedge ($the node adjacent to $A.\dir$ is empty$) \to$
        \Indent
            \State Let $p \gets \Read(\bot, \dir)$.
            \State \Expand$(p)$.
        \EndIndent
        \State $\alpha_5 : (A.\shape = \textsc{expanded}) \wedge (A.\state \in \{\textsc{follower}, \textsc{root}\}) \wedge (\forall B \in N(A) : B.\state \neq \textsc{idle})$ \\
        \hspace{.6cm} $\wedge~(A$ has a tail-child $B : B.\shape = \textsc{contracted}) \to$
        \Indent
            \If {$\Read(\bot, \state) = \textsc{root}$} \Write$(\bot, \dir, \Call{GetNextDir}{\text{counter-clockwise}})$.
            \EndIf
            \State Find a port $p \in \Call{TailChildren}{ }$ s.t.\ $\Read(p, \shape) = \textsc{contracted}$. \Comment{See Algorithm~\ref{alg:hexagonhelpers}.}
            \State Let $p'$ be the label of the tail-child's port that will be connected to $p$ after the pull handover.
            \State \Write$(p, \parent, p')$.
            \State \Pull$(p)$.
        \EndIndent
        \State $\alpha_6 : (A.\shape = \textsc{expanded}) \wedge (A.\state \in \{\textsc{follower}, \textsc{root}\}) \wedge (\forall B \in N(A) : B.\state \neq \textsc{idle})$ \\
        \hspace{.6cm} $\wedge~(A$ has no tail-children$) \to$
        \Indent
            \If {$\Read(\bot, \state) = \textsc{root}$} \Write$(\bot, \dir, \Call{GetNextDir}{\text{counter-clockwise}})$.
            \EndIf
            \State \Contract$(\textsc{tail})$.
        \EndIndent
    \end{algorithmic}
\end{algorithm}

\begin{algorithm}[t]
    \caption{Helper Functions for \algHex} \label{alg:hexagonhelpers}
    \begin{algorithmic}[1]
        \Function{GetNextDir}{$c$} \Comment{$c \in \{\text{clockwise}, \text{counter-clockwise}\}$}
            \State Let $p$ be any head port.
            \State \textbf{try:}
            \Indent
                \While {$\neg\Connected(p) \vee (\Read(p, \state) \not\in \{\textsc{seed}, \textsc{retired}\})$} \label{alg:hexagonhelpers:while1}
                    \State $p \gets$ the next head port in orientation $c$.
                \EndWhile
            \EndIndent
            \State \textbf{catch} \texttt{disconnect-failure} \textbf{do} $p \gets$ the next head port in orientation $c$; go to Step~\ref{alg:hexagonhelpers:while1}.
            \State \textbf{try:}
            \Indent
                \While {$\Connected(p) \wedge (\Read(p, \state) \in \{\textsc{seed}, \textsc{retired}\})$} \label{alg:hexagonhelpers:while2}
                    \State $p \gets$ the next head port in orientation $c$.
                \EndWhile
            \EndIndent
            \State \textbf{catch} \texttt{disconnect-failure} \textbf{do} $p \gets$ the next head port in orientation $c$; go to Step~\ref{alg:hexagonhelpers:while2}.
            \State \Return $p$.
        \EndFunction
        \Function{TailChildren}{ }
            \State Let $P \gets \emptyset$.
            \For {each tail port $p$}
                \State \textbf{try:}
                \Indent
                    \If {$\Connected(p) \wedge (\Read(p, \parent) \text{ points to } A)$}
                        \State $P \gets P \cup \{p\}$.
                    \EndIf
                \EndIndent
                \State \textbf{catch} \texttt{disconnect-failure} \textbf{do} nothing.
            \EndFor
            \State \Return $P$.
        \EndFunction
    \end{algorithmic}
\end{algorithm}

We begin our analysis of the \algHex\ algorithm by showing it is correct under an unfair sequential adversary.
Although the related algorithm of Derakhshandeh et al.\ has already been analyzed in the sequential setting~\cite{Daymude2019-computingprogrammable,Derakhshandeh2015-algorithmicframework}, \algHex\ must be proved correct with respect to its action formulation.

\begin{lemma} \label{lem:hexsequential}
    Any unfair sequential execution of the \algHex\ algorithm terminates with the amoebot system forming a hexagon.
\end{lemma}
\begin{proof}
    We first show that the system remains connected throughout the execution.
    Recall that the amoebot system is assumed to be initially connected.
    A disconnection can only result from a movement, and in particular, a contraction.
    Expansions only enlarge the set of nodes occupied by the system and handovers only change which amoebot occupies the handover node, not the fact that the node remains occupied.
    So it suffices to consider $\alpha_6$, the only action involving a \Contract\ operation.
    Action $\alpha_6$ only allows an expanded follower or root amoebot to contract its tail if it has no idle neighbors or neighbors pointing at its tail as their parent.
    The only other possible tail neighbors are the seed, roots, or retired amoebots; however, all of these neighbors are guaranteed to be connected to the forming hexagon structure.
    Thus, the system remains connected throughout the algorithm's execution.
    
    Now, suppose to the contrary that the \algHex\ algorithm has terminated---i.e., no amoebot has an enabled action---but the system does not form a hexagon.
    By inspection of action $\alpha_3$, the retired amoebots form a hexagon extending counter-clockwise from the seed.
    Thus, for the system to not form a hexagon, there must exist some amoebot that is neither the seed nor retired.
    
    First of all, there cannot be any idle amoebots remaining in the system; in particular, we argue that so long as there are idle amoebots in the system, there exists an idle amoebot for which $\alpha_1$ or $\alpha_2$ is enabled, and thus the algorithm cannot have terminated.
    Suppose to the contrary that there are idle amoebots in the system but none of them have non-idle neighbors, yielding $\alpha_1$ and $\alpha_2$ disabled.
    Then the idle amoebots must be disconnected from the rest of the system, since we assumed that the system contains a unique seed amoebot initially, a contradiction of connectivity.
    Thus, if the algorithm has terminated, all idle amoebots must have already become roots or followers.
    
    For all root or follower amoebots to be disabled, we have the following chain of observations:
    \begin{enumerate}[label=(\textit{\alph*}), leftmargin=0.75cm]
        \item \label{lem:hexsequential:a1} No follower can have a seed or retired neighbor; otherwise, action $\alpha_1$ would be enabled for that follower.
        
        \item \label{lem:hexsequential:a3} Since we have already established that there are no idle amoebots in the system, there must not be a contracted root occupying the next hexagon node; otherwise, action $\alpha_3$ would be enabled for that root.
        
        \item \label{lem:hexsequential:a4} Every contracted root amoebot must have its clockwise traversal of the forming hexagon's surface blocked by another amoebot; otherwise, action $\alpha_4$ would be enabled for some contracted root.
        Moreover, since there are no followers on the hexagon's surface by \ref{lem:hexsequential:a1} and no contracted root has yet reached the next hexagon node by \ref{lem:hexsequential:a3}, each contracted root must be blocked by another root.
        
        \item \label{lem:hexsequential:a5a6} By \ref{lem:hexsequential:a4}, there must exist at least one expanded root amoebot $A$.
        Since actions $\alpha_5$ and $\alpha_6$ must be disabled for $A$ by supposition---and, again, there are no idle amoebots remaining in the system---$A$ must have one or more tail-children that are all expanded.
        
        \item \label{lem:hexsequential:tailchildren} By the same argument, actions $\alpha_5$ and $\alpha_6$ can only be disabled for the expanded tail-children of $A$ if they also each have at least one tail-child, all of which are expanded.
    \end{enumerate}
    
    The chain of expanded tail-children established by \ref{lem:hexsequential:a5a6} and \ref{lem:hexsequential:tailchildren} cannot continue ad infinitum since the amoebot system is finite.
    There must eventually exist an expanded root or follower amoebot that either has a contracted tail-child or no tail-children, enabling $\alpha_5$ or $\alpha_6$, respectively.
    In all cases, we reach a contradiction: so long as the amoebot system does not yet form a hexagon, there must exist an amoebot with an enabled action.
    The execution of any enabled action brings the system monotonically closer to forming a hexagon: turning idle amoebots into followers, bringing followers to the hexagon's surface, turning followers into roots, bringing roots closer to their final position, and finally turning roots into retired amoebots.
    Therefore, regardless of the unfair sequential adversary's choice of enabled amoebot to activate, the system is guaranteed to reach and terminate in a configuration forming a hexagon, as desired.
\end{proof}

We next consider unfair asynchronous executions, the most general of all possible concurrency assumptions.
The \algHex\ algorithm maintains the following invariants:
\begin{enumerate}[label=(\textit{\roman*}), leftmargin=0.75cm]
    \item \label{inv:hex:state} The \state\ variable of an amoebot $A$ can only be updated by $A$ itself.
    This follows from actions $\alpha_1$, $\alpha_2$, and $\alpha_3$.
    
    \item \label{inv:hex:followparent} Only follower amoebots have non-\nil\ \parent\ variables.
    An idle amoebot sets its own \parent\ variable when it becomes a follower.
    While an amoebot $A$ is a follower, the only amoebot that can update $A.\parent$ is the amoebot indicated by $A.\parent$.
    Finally, when a follower becomes a root, it updates its own \parent\ variable to \nil, after which its \parent\ variable never changes again.
    This follows from actions $\alpha_1$, $\alpha_2$, and $\alpha_5$.
    
    \item \label{inv:hex:dir} Only root and retired amoebots have non-\nil\ \dir\ variables.
    The \dir\ variable of an amoebot $A$ can only be updated by $A$ itself.
    Once a \dir\ variable is set by a retired amoebot, it never changes again.
    This follows from actions $\alpha_1$, $\alpha_3$, $\alpha_5$, and $\alpha_6$.
    
    \item \label{inv:hex:nomove} Seed, idle, and retired amoebots are always contracted and never move.
    Moreover, seed and retired amoebots never change their state.
    
    
    \item \label{inv:hex:shape} The \shape\ variable of a root or expanded follower $A$ can only be updated by a movement operation initiated by $A$ itself, while the \shape\ variable of a contracted follower $A$ can only be updated by a \Pull\ operation initiated by the neighboring amoebot connected via $A.\parent$.
    This follows from actions $\alpha_4$, $\alpha_5$, and $\alpha_6$.
    
    \item \label{inv:hex:idleconnect} No amoebot can disconnect from an idle neighbor.
    Moreover, a root will not change its state if it has an idle neighbor.
    This follows from actions $\alpha_3$, $\alpha_5$, and $\alpha_6$.
    
    \item \label{inv:hex:movepath} Root amoebots traverse the surface of the forming hexagon clockwise while follower amoebots are pulled by their parents.
    This follows from actions $\alpha_1$, $\alpha_4$, $\alpha_5$, and $\alpha_6$.
\end{enumerate}

In general, asynchronous executions may cause amoebots to incorrectly evaluate their action guards.
Nevertheless, in the following two lemmas, we show that \algHex\ has the key property that whenever an amoebot thinks an action is enabled, it remains enabled and will execute successfully, even when other actions are executed concurrently.

\begin{lemma} \label{lem:hexenabled}
    For any asynchronous execution of the \algHex\ algorithm, if an action $\alpha_i$ is enabled for an amoebot $A$, then $\alpha_i$ stays enabled for $A$ until $A$ executes an action.
\end{lemma}
\begin{proof}
    We use the invariants to prove the claim on an action-by-action basis.
    \begin{enumerate}[label=$\alpha_{\arabic*}:$, leftmargin=0.75cm]
        \item If $A$ evaluates the guard of $\alpha_1$ as \true, then it must be an idle or follower amoebot with a seed or retired neighbor.
        Invariant~\ref{inv:hex:state} ensures that $A$ remains an idle or follower amoebot, and Invariant~\ref{inv:hex:nomove} ensures its seed or retired neighbor does not move or change state.
        
        \item If $A$ evaluates the guard of $\alpha_2$ as \true, then it must be an idle amoebot with a follower or root neighbor.
        Invariant~\ref{inv:hex:state} ensures that $A$ remains an idle amoebot, and Invariant~\ref{inv:hex:idleconnect} ensures that its neighbors remain connected to $A$ while $A$ is idle.
        A follower neighbor of $A$ can concurrently change its state to root by $\alpha_1$; however, a root neighbor of $A$ will not change its state while $A$ is idle by Invariant~\ref{inv:hex:idleconnect}.
        
        \item If $A$ evaluates the guard of $\alpha_3$ as \true, then it must be a contracted root with no idle neighbors and a seed or retired neighbor that indicates that the node $A$ occupies is the next hexagon node.
        Invariants~\ref{inv:hex:state} and~\ref{inv:hex:shape} ensure that $A$ remains a contracted root, Invariant~\ref{inv:hex:nomove} ensures that $A$ cannot gain any idle neighbors, and Invariants~\ref{inv:hex:dir} and~\ref{inv:hex:nomove} ensure that the seed or retired neighbor continues to indicate the node $A$ occupies as the next hexagon node.
        
        \item If $A$ evaluates the guard of $\alpha_4$ as \true, then it must be a contracted root with no neighbor connected via $A.\dir$.
        Invariants~\ref{inv:hex:state} and~\ref{inv:hex:shape} ensure that $A$ remains a contracted root, and Invariant~\ref{inv:hex:movepath} ensures that no amoebot but $A$ can move into the node adjacent to $A.\dir$.
        
        \item If $A$ evaluates the guard of $\alpha_5$ as \true, then it must be an expanded follower or root with no idle neighbors and some contracted tail-child.
        Invariants~\ref{inv:hex:state} and~\ref{inv:hex:shape} ensure that $A$ remains an expanded follower or root, Invariant~\ref{inv:hex:nomove} ensures that $A$ cannot gain any idle neighbors, and Invariants~\ref{inv:hex:followparent} and~\ref{inv:hex:shape} ensure that any contracted tail-child of $A$ remains so.
        
        \item If $A$ evaluates the guard of $\alpha_6$ as \true, then it must be an expanded follower or root with no idle neighbors and no tail-children.
        Invariants~\ref{inv:hex:state} and~\ref{inv:hex:shape} ensure that $A$ remains an expanded follower or root, and Invariants~\ref{inv:hex:followparent} and~\ref{inv:hex:nomove} ensure that $A$ cannot gain any idle neighbors or tail-children.
    \end{enumerate}
    
    Therefore, any action that $A$ evaluates as enabled must remain enabled, as claimed.
\end{proof}

\begin{lemma} \label{lem:hexsuccess}
    For any asynchronous execution of the \algHex\ algorithm, any execution of an enabled action is successful and unaffected by any concurrent action executions.
\end{lemma}
\begin{proof}
    We once again consider each action individually.
    \begin{enumerate}[label=$\alpha_{\arabic*}:$, leftmargin=0.75cm]
        \item Action $\alpha_1$ first executes two \Write\ operations to $A$'s own public memory which cannot fail.
        It then executes a helper function \textsc{GetNextDir}(counter-clockwise) which involves a sequence of \Connected\ and \Read\ operations.
        \Connected\ operations always succeed, so it suffices to consider the \Read\ operations.
        While it is possible that \Read\ operations issued to follower or root neighbors may fail if those neighbors disconnect, these failures are caught by error handling and thus do not cause the action to fail.
        Moreover, the critical \Read\ operations issued to seed or retired neighbors that the function depends on for calculating the correct direction must succeed by the guard of $\alpha_1$ and Lemma~\ref{lem:hexenabled}.
        Once this direction is computed, $\alpha_1$ then executes a \Write\ operation to $A$'s own memory which cannot fail.
        
        \item Action $\alpha_2$ first executes \Connected\ and \Read\ operations to find a follower or root neighbor. Such a neighbor must exist and the corresponding \Read\ operations must succeed by the guard of $\alpha_2$ and Lemma~\ref{lem:hexenabled}.
        Action $\alpha_2$ then executes two \Write\ operations to $A$'s own public memory which cannot fail.
        
        \item Action $\alpha_3$ first executes helper function \textsc{GetNextDir}(clockwise) which must succeed by an argument analogous to that of $\alpha_1$.
        Once this direction is computed, $\alpha_3$ executes two \Write\ operations to $A$'s own public memory which cannot fail.
        
        \item Action $\alpha_4$ executes an \Expand\ operation toward port $A.\dir$ which must succeed because $A$ is contracted and the node adjacent to $A.\dir$ must remain unoccupied, as ensured by the guard of $\alpha_4$ and Lemma~\ref{lem:hexenabled}.
        
        \item Action $\alpha_5$ first executes a conditional based on a \Read\ operation issued to $A$'s own public memory which cannot fail.
        It then executes helper function \textsc{GetNextDir}(counter-clockwise) which must succeed by an argument analogous to that of $\alpha_1$.
        The computed direction is then used in a \Write\ operation to $A$'s own public memory which cannot fail.
        Action $\alpha_5$ then executes a helper function \textsc{TailChildren}() which, like \textsc{GetNextDir}, involves \Connected\ and \Read\ operations.
        It must succeed for similar reasons: any failed \Read\ operations are caught by error handling, and the critical \Read\ operations issued to tail-children must succeed by the guard of $\alpha_5$ and Lemma~\ref{lem:hexenabled}.
        Once the ports connected to tail-children are computed, \Read\ operations are executed to find a contracted tail-child $B$ which once again must succeed by the guard of $\alpha_5$ and Lemma~\ref{lem:hexenabled}.
        Finally, $\alpha_5$ executes a \Write\ to the public memory of $B$ and performs a \Pull\ handover with $B$; both operations must succeed because $B$ remains connected to $A$ and cannot be involved in another movement by Invariant~\ref{inv:hex:shape}.
        
        \item Action $\alpha_6$ first executes the same conditional operation as $\alpha_5$ and thus succeeds for an analogous reason.
        It then executes a single \Contract\ operation which must succeed because $A$ is expanded, as ensured by the guard of $\alpha_6$ and Lemma~\ref{lem:hexenabled}.
    \end{enumerate}
    
    Therefore, any execution of an enabled action must be successful and unaffected by concurrent action executions, as claimed.
\end{proof}

We next show that the \algHex\ algorithm is serializable.
We denote the execution of an action $\alpha$ by an amoebot $A$ in an execution of the algorithm as a pair $(A, \alpha)$.

\begin{lemma} \label{lem:hexserializable}
    For any asynchronous execution of the \algHex\ algorithm, there exists a sequential ordering of its action executions producing the same final configuration.
\end{lemma}
\begin{proof}
    Argue by induction on $i$, the number of action executions in the asynchronous execution of $\algHex$.
    Clearly, if $i = 1$, the asynchronous execution of a single action is also a sequential execution, and we are done.
    So suppose that any asynchronous execution of $\algHex$ consisting of $i \geq 1$ action executions can be serialized, and consider any asynchronous execution $\sched$ consisting of $i + 1$ action executions.
    One can partially order the action executions $(A, \alpha$) of $\sched$ according to the wall-clock time the asynchronous adversary activated $A$; note that this is only used for this analysis, and the wall-clock time is never available to the amoebots.
    Let $(A^*, \alpha^*)$ be the action execution with the latest activation time; if there are multiple such executions because the asynchronous adversary activated multiple amoebots simultaneously, choose any such execution.
    If $(A^*, \alpha^*)$ was removed from $\sched$ to produce a new asynchronous execution $\sched^-$, we have by Lemmas~\ref{lem:hexenabled} and~\ref{lem:hexsuccess} that the remaining $i$ action executions must still be enabled and successful since all other action executions either terminated before $(A^*, \alpha^*)$ was initiated or were concurrent with it.
    By the induction hypothesis, there must exist a sequential ordering of the $i$ action executions in $\sched^-$ producing the same final configuration as $\sched^-$.
    Append $(A^*, \alpha^*)$ to the end of this sequential execution to produce $\sched^*$, a sequential execution of $i + 1$ action executions.
    Any actions that were concurrent with $(A^*, \alpha^*)$ in $\sched$ have now terminated before $(A^*, \alpha^*)$ in $\sched^*$.
    However, by Lemmas~\ref{lem:hexenabled} and~\ref{lem:hexsuccess}, this does not change the fact that $\alpha^*$ is enabled for $A^*$ and its execution is successful and produces the same outcome in $\sched^*$.
    Therefore, we conclude that there exists a sequential ordering of the action executions of $\sched$ producing the same final configuration.
\end{proof}

Finally, we show that the \algHex\ algorithm is correct under an unfair asynchronous adversary.

\begin{lemma} \label{lem:hexconcurrent}
    Any unfair asynchronous execution of the \algHex\ algorithm terminates with the amoebot system forming a hexagon.
\end{lemma}
\begin{proof}
    First suppose to the contrary that there exists an asynchronous execution of \algHex\ that does not terminate; i.e., there are an infinite number of executions of enabled actions.
    By Lemmas~\ref{lem:hexenabled} and~\ref{lem:hexsuccess}, any such action execution must succeed and do exactly what it would have in a sequential execution where there are no other concurrent action executions.
    But Lemma~\ref{lem:hexsequential} implies that there can only be a finite number of successful action executions before no amoebot has any enabled actions left, a contradiction.
    So all asynchronous executions of \algHex\ must terminate.

    Now suppose to the contrary that there exists an asynchronous execution of \algHex\ that has terminated but the system does not form a hexagon.
    By Lemma~\ref{lem:hexserializable}, there must exist a sequential execution that also produces this non-hexagon final configuration.
    However, this is a contradiction of Lemma~\ref{lem:hexsequential} which states that every sequential execution of \algHex\ must terminate with the system forming a hexagon.
\end{proof}

Our analysis culminates in the following theorem.

\begin{theorem} \label{thm:hexagon}
    Assuming geometric space, assorted orientations, and constant-size memory, the \algHex\ algorithm solves the hexagon formation problem under any adversary.
\end{theorem}

Serializability and correctness under an asynchronous adversary (Lemmas~\ref{lem:hexserializable} and~\ref{lem:hexconcurrent}) follow directly from Lemmas~\ref{lem:hexsequential}--\ref{lem:hexsuccess}, independent of the specific details of \algHex.
Thus, Lemmas~\ref{lem:hexsequential}--\ref{lem:hexsuccess} establish a set of general sufficient conditions for amoebot algorithm correctness under an asynchronous adversary.
We are optimistic that other existing amoebot algorithms, once translated into action formulations, will also satisfy these conditions.

\section{A General Framework for Concurrency Control} \label{sec:framework}

In the sequential setting where only one amoebot is active at a time, operation failures are necessarily the fault of the algorithm designer: e.g., attempting to \Read\ on a disconnected port, attempting to \Expand\ when already expanded, etc.
Barring these design errors, it suffices to focus only on the correctness of the algorithm---i.e., whether the algorithm's actions always produce the desired system behavior under any sequential execution---not whether the individual actions themselves execute as intended.
This is the focus of most existing amoebot works~\cite{AndresArroyo2018-stochasticapproach,Cannon2019-localstochastic,Cannon2016-markovchain,Daymude2020-convexhull,Daymude2017-improvedleader,Derakhshandeh2015-algorithmicframework,Derakhshandeh2016-universalshape,Derakhshandeh2017-universalcoating,Derakhshandeh2015-leaderelection,Dufoulon2021-efficientdeterministic,Emek2019-deterministicleader,Gastineau2019-distributedleader}.

Our present focus is on asynchronous executions, where concurrent action executions can mutually interfere, affect outcomes, and cause failures far beyond those of simple designer negligence.
Ensuring algorithm correctness in spite of concurrency thus appears to be a significant burden for the algorithm designer, especially for problems that are challenging even in the sequential setting due to the constraints of constant-size memory, assorted orientation, and strictly local interactions.
What if there was a way to ensure that correct, sequential amoebot algorithms could be lifted to the asynchronous setting without sacrificing correctness?
This would give the best of both worlds: the relative ease in design from the sequential setting and the correct execution in a more realistic concurrent setting.

In this section, we introduce and rigorously analyze a framework for transforming an algorithm $\alg$ that works correctly for every sequential execution into an algorithm $\alg'$ that works correctly for every asynchronous execution.
We prove that our framework achieves this goal so long as the original algorithms satisfy certain \textit{conventions}.
These conventions limit the full generality of the amoebot model in order to provide a common structure to the algorithms.
In Section~\ref{subsec:conventions}, we define these conventions and prove that they are satisfied by both the \algHex\ algorithm of Section~\ref{sec:hexagon} and a broad class of \textit{stationary} amoebot algorithms.
This implies that these algorithms are immediately compatible with our concurrency control framework, which we detail in Section~\ref{subsec:framework} and rigorously analyze in Section~\ref{subsec:serialanalysis}.

\subsection{Algorithm Conventions for Concurrency Control} \label{subsec:conventions}

The first convention requires that all actions of the given algorithm are executed successfully under a sequential adversary.
For sequential executions, the \textit{system configuration} is defined as the mapping of amoebots to the node(s) they occupy and the contents of each amoebot's public memory.
Certainly, this configuration is well-defined whenever all amoebots are inactive, and we call a configuration \textit{legal} whenever the requirements of our amoebot model are met, i.e., every position is occupied by at most one amoebot, each amoebot is either contracted or expanded, its \shape\ variable corresponds to its physical shape, and its \lock\ variable corresponds to its lock state.
Whenever we talk about a system configuration in the following, we assume that it is legal.

\begin{convention}[Validity] \label{conv:valid}
    All actions $\alpha$ of an amoebot algorithm $\alg$ should be \underline{valid}, i.e., for all system configurations in which $\alpha$ is enabled for some amoebot $A$, the execution of $\alpha$ by $A$ should be successful whenever all other amoebots are inactive.
\end{convention}

The second convention defines a common structure for an algorithm's actions by controlling the order and number of operations they perform.
This structure is similar in spirit to the \textit{look-compute-move} paradigm used in the mobile robots literature (see, e.g.,~\cite{Flocchini2019-distributedcomputing}), though the canonical amoebot model's underlying message passing communication adds additional complexity.
Moreover, the instantaneous snapshot performed in the mobile robots' look phase is not trivially realizable by amoebots whose public memories are included in neighborhood configurations (Section~\ref{subsec:relwork}).

\begin{convention}[Phase Structure] \label{conv:phases}
    Each action of an amoebot algorithm $\alg$ should structure its operations as: (1) a \underline{compute phase}, during which an amoebot performs a finite amount of computation and a finite sequence of \Connected, \Read, and \Write\ operations, and (2) a \underline{move phase}, during which an amoebot performs at most one movement operation decided upon in the compute phase.
    In particular, no action should use \Lock\ or \Unlock\ operations.
\end{convention}

The third and final convention, \textit{expansion-robustness}, allows us to map asynchronous executions of algorithms produced by the concurrency control framework to related sequential executions.
A key challenge in achieving this mapping for concurrent executions of amoebot algorithms is the possibility of one or more amoebots expanding into the neighborhood of an amoebot $A$ that has already started executing an action of its own.
These newly expanded neighbors were not present when $A$ evaluated its action guard, and thus may cause the execution to exhibit different behavior than in the sequential setting---or worse, fail altogether.
Informally, an expansion-robust algorithm is ambivalent to these concurrent expansions, guaranteeing correct behavior regardless.

\begin{algorithm}[t]
    \caption{Expansion-Robust Variant $\alg^E$ of Algorithm $\alg$ for Amoebot $A$} \label{alg:expandrobust}
    \begin{algorithmic}[1]
        \Statex \textbf{Input}: Algorithm $\alg = \{[\alpha_i : g_i \to ops_i] : i \in \{1, \ldots, m\}\}$ satisfying Conventions~\ref{conv:valid} and~\ref{conv:phases}.
        \State Set $\alpha_0^E : (\exists$ port $p$ of $A : A.\xflag_p = \true) \to$ \Write$(\bot, \xflag_p, \false)$.
        \For {each action $[\alpha_i : g_i \to ops_i] \in \alg$}
            \State Set $g_i^E \gets g_i$ with $N(A)$ replaced by $N^E(A)$ and connections defined w.r.t.\ $N^E(A)$.
            \State Set $ops_i^E \gets$ ``Do:
            \Indent
                \For {each port $p$ of $A$} \Write$(\bot, \xflag_p, \false)$.  \Comment{Reset own expand flags.}  \label{alg:expandrobust:resetown}
                \EndFor
                \For {each unique neighbor $B \in \Connected()$}  \Comment{Reset neighbor's expand flags.}
                    \For{each port $p$ of $B$} \Write$(B, \xflag_p, \false)$. \label{alg:expandrobust:resetnbr}
                    \EndFor
                \EndFor
                \State Execute each operation of $ops_i$ with connections defined w.r.t.\ $N^E(A)$.
                \If {a \Pull\ or \Push\ operation was executed with neighbor $B$}
                    \For {each new port $p$ of $A$ not connected to $B$} \Write$(\bot, \xflag_p, \true)$.  \label{alg:expandrobust:handover1}
                    \EndFor
                    \For {each new port $p$ of $B$ not connected to $A$} \Write$(B, \xflag_p, \true)$.  \label{alg:expandrobust:handover2}
                    \EndFor
                \ElsIf {an \Expand\ operation was successfully executed}
                    \For {each new port $p$ of $A$} \Write$(\bot, \xflag_p, \true)$.  \label{alg:expandrobust:expand}
                    \EndFor
                \ElsIf {an \Expand\ operation failed in its execution} undo $ops_i$.''
                \label{alg:expandrobust:undo}
                \EndIf
            \EndIndent
        \EndFor
        \State \Return $\alg^E = \{[\alpha_i^E : g_i^E \to ops_i^E] : i \in \{0, \ldots, m\}\}$.
    \end{algorithmic}
\end{algorithm}

Formally, let $\alg$ be any amoebot algorithm satisfying Conventions~\ref{conv:valid} and~\ref{conv:phases} and consider its expansion-robust variant $\alg^E$ defined as follows.
Each amoebot $A$ executing $\alg^E$ additionally stores in public memory \textit{expand flags} $A.\xflag_p \in \{\true, \false\}$ for each of its ports $p$ that are initially set to \false.
These expand flags communicate when an amoebot has newly expanded into another amoebot's neighborhood.
Each action $\alpha_i : g_i \to ops_i$ in $\alg$ translates to an action $\alpha_i^E : g_i^E \to ops_i^E$ in $\alg^E$, as detailed in Algorithm~\ref{alg:expandrobust}.\footnote{For the sake of clarity and brevity, we abuse \Connected, \Read, and \Write\ notation slightly by referring directly to the neighboring amoebots and not to the ports which they are connected to.}
The main difference is that while an amoebot $A$ executes actions with respect to its full neighborhood $N(A)$ in algorithm $\alg$, it does so only with respect to its \textit{established neighborhood} $N^E(A) = \{B \in N(A) : \exists \text{ port $p$ of $B$ connected to $A$ s.t.\ } B.\xflag_p = \false\}$ in algorithm $\alg^E$, effectively ignoring its newly expanded neighbors until its next action execution. 
The expansion-robustness convention can now be stated as follows:

\begin{convention}[Expansion-Robustness] \label{conv:expandrobust}
    An amoebot algorithm $\alg$ should be \underline{expansion-robust}, meaning that for any (legal) initial system configuration $C_0$ of $\alg$,
    \begin{enumerate}
        \item Termination. If all sequential executions of $\alg$ starting in $C_0$ terminate, all sequential executions of $\alg^E$ starting in $C_0^E$ (i.e., $C_0$ with all \false\ expand flags) also terminate.
        
        \item Correctness. If some sequential execution of $\alg^E$ starting in $C_0^E$ terminates in a configuration $C^E$, there exists a sequential execution of $\alg$ starting in $C_0$ that terminates in $C$ (i.e., $C^E$ without expand flags).
    \end{enumerate}
\end{convention}

It is worth emphasizing that $\alg^E$ is simply a useful artifact for formalizing expansion-robustness and, as such, we are only interested in its behavior in the sequential setting.
Several operations in $\alg^E$, such as the ``undo'' on Line~\ref{alg:expandrobust:undo} of Algorithm~\ref{alg:expandrobust}, may not even be possible in a concurrent setting, but this is inconsequential for our purposes.

We now demonstrate that Conventions~\ref{conv:valid}--\ref{conv:expandrobust} are not too limiting; i.e., there do exist algorithms that satisfy these conventions and thus are compatible with our concurrency control framework.
Of the three conventions, expansion-robustness (Convention~\ref{conv:expandrobust}) is the most technically difficult to verify.
However, \textit{stationary} amoebot algorithms $\alg$---i.e., those that do not perform any movement operations, including many of the existing algorithms for leader election~\cite{Bazzi2019-stationarydeterministic,Daymude2017-improvedleader,Derakhshandeh2015-leaderelection,DiLuna2020-shapeformation,Gastineau2019-distributedleader} and the recent algorithm for energy distribution~\cite{Daymude2021-bioinspiredenergy}---are trivially expansion-robust since no amoebot ever moves and thus $\alg$ and $\alg^E$ are identical.

\begin{observation} \label{obs:stationarycompatible}
    Any stationary amoebot algorithm satisfies Convention~\ref{alg:expandrobust}.
\end{observation}

In the conference version of this work~\cite{Daymude2021-canonicalamoebot}, it remained an open question whether any amoebot algorithm involving movement satisfied all three conventions.
By replacing the prior version's \textit{monotonicity} convention with expansion-robustness, we identify the \algHex\ algorithm of Section~\ref{sec:hexagon} as such an algorithm.

\begin{theorem} \label{thm:hexcompatible}
    The \algHex\ algorithm satisfies Conventions~\ref{conv:valid}--\ref{conv:expandrobust}.
\end{theorem}
\begin{proof}
    It is easy to verify that \algHex\ satisfies Conventions~\ref{conv:valid} and~\ref{conv:phases} by inspection.
    Hence, it remains to show \algHex\ is expansion-robust (Convention~\ref{conv:expandrobust}).
    To prove correctness, we will show that whenever an action $\alpha_i^E \in \algHex^E$ (other than $\alpha_0^E$) is enabled for an amoebot $A$ w.r.t.\ $N^E(A)$, action $\alpha_i \in \algHex$ is enabled for $A$ w.r.t.\ $N(A)$.
    Moreover, we show that the executions of $\alpha_i$ and $\alpha_i^E$ by $A$ are identical except for the handling of expand flags.
    This immediately implies that every sequential execution of $\algHex^E$ represents an identical sequential execution of \algHex\ (after removing the executions of $\alpha_0^E$), proving correctness.
    Observe that expand flags are only set to \true\ in $\algHex^E$ as a result of an \Expand, \Pull, or \Push\ operation, which are specific to $\alpha_4^E$ and $\alpha_5^E$; thus, we need only focus on amoebots executing these actions.
    
    Suppose a contracted root amoebot $A$ executes $\alpha_4^E$, expanding towards $A.\dir$ along the surface of the forming hexagon into the neighborhood of another amoebot $B$ it was not already connected to.
    If $B$ is a root, then it can be easily verified by inspecting the guards and operations of $\alpha_3$, $\alpha_4$, $\alpha_5$, and $\alpha_6$ that $A$---which must be ``behind'' $B$ in the clockwise traversal of the hexagon's surface---has no bearing on which actions are enabled for $B$ nor on their execution.
    If $B$ is a follower, then its parent must be some amoebot other than $A$ because $A$ only just became its neighbor.
    Thus, in either of these cases, the fact that $A \in N(B) \setminus N^E(B)$ is inconsequential.
    Finally, if $B$ is idle, then $A \not\in N^E(B)$ prohibits $B$ from choosing $A$ as its parent in $\alpha_2^E$ while the same choice is allowed in $\alpha_2$.
    If $B$ chooses some amoebot $C \neq A$ as its parent while executing $\alpha_2^E$, then certainly an execution of $\alpha_2$ by $B$ in the same configuration could have made the same choice.
    Otherwise, if $A$ is the only neighbor of $B$, then we know $\alpha_0^E$ is continuously enabled for $A$ while $A \not\in N^E(B)$.
    Amoebot $A$ cannot disconnect from $B$ while $B$ is idle by the guards of $\alpha_3$, $\alpha_5$, and $\alpha_6$, so eventually an execution of $\alpha_0^E$ resets the expand flags of $A$, allowing $B$ to choose $A$ as its parent just as in its corresponding execution of $\alpha_2$.
    
    Now suppose an expanded follower or root amoebot $A$ executes $\alpha_5^E$, pulling some follower tail-child $B$ in a handover.
    Consider any new port $p$ of $B$ for which $B.\xflag_p = \true$ after this handover occurs.
    If there is no neighboring amoebot connected to port $p$, then only a contracted root could expand into that position, as already covered in the analysis of $\alpha_4^E$.
    So suppose that an amoebot $C$ is connected to $B$ via port $p$.
    The guard of $\alpha_5$ would have prohibited $A$ from performing the pull handover with $B$ if $A$ had any idle neighbors, so $C$ cannot be idle.
    Inspection of the guards and operations of actions $\alpha_1$, $\alpha_5$, and $\alpha_6$ show that $C$ is irrelevant to $B$ if $C$ is a root.
    So it remains to consider if $C$ is a follower.
    If $C.\parent$ points to any node other than the new head of $B$, then $B \in N(C) \setminus N^E(C)$ is inconsequential to $C$.
    Otherwise, if $C.\parent$ refers to the new head of $B$, the fact that $B \not\in N^E(C)$ is once again inconsequential to $C$ because children never initiate interactions with their parents.
    Therefore, in all cases, the correctness condition of expansion-robustness follows.
    
    To prove termination, suppose to the contrary that there exists a sequential execution $\sched^E$ of $\algHex^E$ starting in a legal initial configuration $C_0^E$ that contains an infinite number of action executions.
    By our correctness analysis, we know that $\sched^E$ must correspond to an identical sequential execution $\sched$ of \algHex, modulo executions of $\alpha_0^E$.
    In Section~\ref{sec:hexagon}, we proved that all sequential executions of \algHex, $\sched$ included, must terminate (Lemma~\ref{lem:hexsequential}).
    Thus, $\sched^E$ must contain an infinite number of executions of $\alpha_0^E$.
    But this is impossible, as there are a finite number of amoebots and each of them has a finite number of expand flags to reset with $\alpha_0^E$, a contradiction.
    Thus, the termination condition is satisfied, and \algHex\ is expansion-robust.
\end{proof}

With the validity, phase structure, and expansion-robustness conventions established, we now turn to the description and analysis of the concurrency control framework.

\subsection{The Concurrency Control Framework} \label{subsec:framework}

Our \textit{concurrency control framework} (Algorithm~\ref{alg:framework}) takes as input any amoebot algorithm $\alg = \{[\alpha_i : g_i \to ops_i] : i \in \{1, \ldots, m\}\}$ satisfying Conventions~\ref{conv:valid}--\ref{conv:expandrobust} and produces a corresponding algorithm $\alg' = \{[\alpha' : g' \to ops']\}$ composed of a single action $\alpha'$.
The core idea of our framework is to carefully incorporate locks in $\alpha'$ as a wrapper around the actions of $\alg$, ensuring that $\alg'$ only produces outcomes in concurrent settings that $\alg$ can produce in the sequential setting.
With locks, action guards that in general can only be evaluated reliably in the sequential setting can now also be evaluated reliably in concurrent settings.

To avoid any deadlocks that locking may cause, our framework adds an \textit{activity bit} variable $A.\activity \in \{\true, \false\}$ to the public memory of each amoebot $A$ indicating if any changes have occurred in the memory or neighborhood of $A$ since it last attempted to execute an action.
The single action $\alpha'$ of $\alg'$ has guard $g' = (A.\activity = \true)$, ensuring that $\alpha'$ is only enabled for an amoebot $A$ if changes in its memory or neighborhood may have caused some actions of $\alg$ to become enabled.
As will become clear in the presentation of the framework, \Write\ and movement operations may enable actions of $\alg$ not only for the neighbors of the acting amoebot, but also for the neighbors of those neighbors (i.e., in the 2-neighborhood of the acting amoebot).
The acting amoebot cannot directly update the activity bits of amoebots in its 2-neighborhood, so it instead sets its neighbors' \textit{awaken bits} $A.\awaken \in \{\true, \false\}$ to indicate that they should update their neighbors' activity bits in their next action.
Initially, $A.\activity = \true$ and $A.\awaken = \false$ for all amoebots $A$.

\begin{algorithm}[t]
\caption{Concurrency Control Framework for Amoebot $A$} \label{alg:framework}
\begin{algorithmic}[1]
    \Statex \textbf{Input}: Algorithm $\alg = \{[\alpha_i : g_i \to ops_i] : i \in \{1, \ldots, m\}\}$ satisfying Conventions~\ref{conv:valid}--\ref{conv:expandrobust}.
    \State Set $g' \gets (A.\activity = \true)$ and $ops' \gets$ ``Do: \label{alg:framework:guard}
    \Indent
        \State \textbf{try:} Set $\mathcal{L} \gets \Lock()$ to attempt to lock $A$ and its persistent neighbors. \label{alg:framework:lock}
        \State \textbf{catch} \texttt{lock-failure} \textbf{do} abort. \label{alg:framework:lockabort}
        \If {\Read$(\bot, \awaken) = \true$} \label{alg:framework:awaken}
            \ForAll {amoebots $B \in \mathcal{L}$} \Write$(B, \activity, \true)$.
            \EndFor
            \State \Write$(\bot, \awaken, \false)$, $\Unlock(\mathcal{L})$, and abort. \label{alg:framework:awakenabort}
        \EndIf
        \ForAll {actions $[\alpha_i : g_i \to ops_i] \in \alg$} \label{alg:framework:guards}
            \State Perform \Connected\ and \Read\ operations to evaluate guard $g_i$ w.r.t.\ $\mathcal{L}$.\label{alg:framework:read}
            \State Evaluate $g_i$ in private memory to determine if $\alpha_i$ is enabled. \label{alg:framework:evaluate}
        \EndFor
        \If {no action is enabled} \Write$(\bot, \activity,\false)$, $\Unlock(\mathcal{L})$, and abort.\label{alg:framework:disabledabort}
        \EndIf
        \State Choose an enabled action $\alpha_i \in \alg$ and perform its compute phase in private memory. \label{alg:framework:compute1}
        \State \multiline{Let $W_i$ be the set of \Write\ operations and $M_i$ be the movement operation in $ops_i$ based on its compute phase; set $M_i \gets \nil$ if there is none. \label{alg:framework:compute2}} 
        \If {$M_i$ is \Expand\ (say, from node $u$ into node $v$)}
            \State \textbf{try:} Perform the \Expand\ operation and \Write$(\bot, \awaken, \true)$. \label{alg:framework:expand} \label{alg:framework:expandawaken}
            \State \textbf{catch} \texttt{expand-failure} \textbf{do} $\Unlock(\mathcal{L})$ and abort. \label{alg:framework:expandabort}
        \EndIf
        \ForAll {amoebots $B \in \mathcal{L}$} \Write$(B, \activity, \true)$. \label{alg:framework:active}
        \EndFor
        \ForAll {\Write$(B, x, x_{val}) \in W_i$} \Write$(B, x, x_{val})$ and \Write$(B, \awaken, \true)$. \label{alg:framework:writes}
        \EndFor
        \If {$M_i$ is \nil\ or \Expand} 
        \Unlock$(\mathcal{L})$. \label{alg:framework:nomove}
        \ElsIf {$M_i$ is \Contract\ (say, from nodes $u,v$ into node $u$)}
            \State \Unlock\ each amoebot in $\mathcal{L}$ that is adjacent to node $v$ but not to node $u$. \label{alg:framework:contractunlock1}
            \State Perform the \Contract\ operation. \label{alg:framework:contract}
            \State \Unlock\ each remaining amoebot in $\mathcal{L}$. \label{alg:framework:contractunlock2}
        \ElsIf {$M_i$ is \Push\ (say, $A$ is pushing $B$)}
            \State \Write$(\bot, \awaken, \true)$ and \Write$(B, \awaken, \true)$. \label{alg:framework:pushawaken}
            \State Perform the \Push\ operation. \label{alg:framework:push}
            \State $\Unlock(\mathcal{L})$. \label{alg:framework:pushunlock}
        \ElsIf {$M_i$ is \Pull\ (say, $A$ in nodes $u,v$ is pulling $B$ into node $v$)}
            \State \Write$(B, \awaken, \true)$. \label{alg:framework:pullawaken}
            \State \Unlock\ each amoebot in $\mathcal{L}$ (except $B$) that is adjacent to node $v$ but not to node $u$. \label{alg:framework:pullunlock1}
            \State Perform the \Pull\ operation. \label{alg:framework:pull}
            \State \Unlock\ each remaining amoebot in $\mathcal{L}$.'' \label{alg:framework:pullunlock2}
        \EndIf
    \EndIndent
    \State \Return $\mathcal{A'} = \{[\alpha' : g' \to ops']\}$.
\end{algorithmic}
\end{algorithm}

Algorithm $\alg'$ only contains one action $\alpha' : g' \to ops'$ where $g'$ requires that an amoebot's activity bit is set to \true\ (Step~\ref{alg:framework:guard}).
If $\alpha'$ is enabled for an amoebot $A$, $A$ first attempts to \Lock\ itself and its persistent neighbors (Step~\ref{alg:framework:lock}).
Given that it locks successfully, there are two cases.
If $A.\awaken = \true$, then some amoebot must have previously changed the neighborhood of $A$ without being able to update the corresponding neighbors' activity bits (Steps~\ref{alg:framework:expandawaken}, \ref{alg:framework:writes}, \ref{alg:framework:pushawaken}, or \ref{alg:framework:pullawaken}).
So $A$ updates the intended activity bits to \true, resets $A.\awaken$, releases its locks, and aborts (Steps~\ref{alg:framework:awaken}--\ref{alg:framework:awakenabort}).
Otherwise, $A$ obtains the necessary information to evaluate the guards of all actions in algorithm $\alg$ (Steps~\ref{alg:framework:guards}--\ref{alg:framework:evaluate}).
If no action of $\alg$ is enabled for $A$, $A$ sets $A.\activity$ to \false, releases its locks, and aborts; this disables $\alpha'$ for $A$ until some future change occurs in its neighborhood (Step~\ref{alg:framework:disabledabort}).
Otherwise, $A$ chooses any enabled action and executes its compute phase in private memory (Step~\ref{alg:framework:compute1}) to determine which \Write\ and movement operations, if any, it wants to perform (Step~\ref{alg:framework:compute2}).

Before enacting these operations (thereby updating the system's configuration) amoebot $A$ must be certain that no operation of $\alpha'$ will fail.
It has already passed its first point of failure: the \Lock\ operation in Step~\ref{alg:framework:lock}.
But the \Expand\ operation of $\alpha'$ may also fail if it conflicts with some other concurrent expansion (Step~\ref{alg:framework:expand}).
In either case, $A$ handles the failure, releases any locks it obtained (if any), and aborts (Steps~\ref{alg:framework:lockabort} and~\ref{alg:framework:expandabort}).
As we will show in Lemma~\ref{lem:failure}, these are the only two operations of $\alpha'$ that can fail.
Provided neither of these failures occur, $A$ can now perform operations that---without locks on its neighbors---could otherwise interfere with its neighbors' actions or be difficult to undo.
This begins with $A$ setting the activity bits of all its locked neighbors to \true\ since it is about to cause activity in its neighborhood (Step~\ref{alg:framework:active}).
It then enacts the \Write\ operations it decided on during its computation, writing updates to its own public memory and the public memories of its neighbors.
Since writes to its neighbors can change what amoebots in its 2-neighborhood see, it must also set the awaken bits of the neighbors it writes to to \true\ (Step~\ref{alg:framework:writes}).

The remainder of the framework handles movements and releases locks.
If $A$ did not want to move or it intended to \Expand---which, recall, it already did in Step~\ref{alg:framework:expand}---it can simply release all its locks (Step~\ref{alg:framework:nomove}).
If $A$ wants to contract, it must first release its locks on the neighbors it is contracting away from; it can then \Contract\ and, once contracted, release its remaining locks (Step~\ref{alg:framework:contractunlock1}--\ref{alg:framework:contractunlock2}).
If $A$ wants to perform a \Push\ handover, it does so and then releases all its locks (Steps~\ref{alg:framework:pushawaken}--\ref{alg:framework:pushunlock}).
Finally, pull handovers are handled similarly to contractions:
$A$ first releases its locks on the neighbors it is disconnecting from; it can then \Pull\ and, once contracted, release its remaining locks (Steps~\ref{alg:framework:pullawaken}--\ref{alg:framework:pullunlock2}).

\subsection{Analysis} \label{subsec:serialanalysis}

In this section, we prove the following result regarding the concurrency control framework.

\begin{theorem} \label{thm:concurrencycontrol}
    Let $\alg$ be any amoebot algorithm satisfying Conventions~\ref{conv:valid}--\ref{conv:expandrobust} and $\alg'$ be the amoebot algorithm produced from $\alg$ by the concurrency control framework (Algorithm~\ref{alg:framework}).
    Let $C_0$ be any (legal) initial system configuration for $\alg$ and let $C_0'$ be its extension for $\alg'$ that adds $A.\activity = \true$ and $A.\awaken = \false$ for all amoebots $A$.
    If every sequential execution of $\alg$ starting in $C_0$ terminates, every asynchronous execution of $\alg'$ starting in $C_0'$ also terminates.
    Moreover, if $C'$ is the final configuration of some asynchronous execution of $\alg'$ starting in $C_0'$, then there exists a sequential execution of $\alg$ starting in $C_0$ with final configuration $C$ that is identical to $C'$, modulo amoebots' activity and awaken bits.
\end{theorem}

Informally, this theorem shows that the concurrency control framework only permits asynchronous outcomes that could have occurred in the sequential setting, provided algorithm $\alg$ always terminates in the sequential setting and satisfies the three conventions.

This analysis has three parts.
First, we show that asynchronous executions of $\alg'$ can be ``sanitized'' of ``irrelevant'' events without changing the system's final configuration (Observation~\ref{obs:isolation}--Lemma~\ref{lem:sanitize}).
Second, we show that any sanitized asynchronous execution of $\alg'$ can be transformed into a sequential execution of $(\alg^E)'$, the framework-applied expansion-robust version of $\alg$, again without changing the final configuration (Lemmas~\ref{lem:dag}--\ref{lem:serializability}).
Finally, we leverage the expansion-robustness of $\alg$ (Convention~\ref{conv:expandrobust}) to show that any final configuration reached by a sequential execution of $(\alg^E)'$ is also reachable by a sequential execution of $\alg$ (Lemmas~\ref{lem:seqfinite}--\ref{lem:seqoutcome}).
Combining these results after showing asynchronous executions of $\alg'$ terminate (Lemma~\ref{lem:progress}) yields the theorem.

We first analyze algorithm $\alg'$ under asynchronous executions.
Recall from Section~\ref{subsec:modelalgs} that although each amoebot executes at most one action at a time and executes that action's operations sequentially to completion, asynchronous executions allow arbitrarily many amoebots to execute actions simultaneously.
An \textit{asynchronous schedule} is an assignment of precise timing by a wall-clock to every \textit{event} in an asynchronous execution; i.e., every message sending and receipt, variable update in public memory, movement start and end, and operation failure.
We emphasize that this wall-clock timing is only used for this analysis and is unavailable to the amoebots.
In keeping with Sections~\ref{subsec:modeloperations} and~\ref{subsec:modelalgs}, we make no assumptions on timing other than (\textit{i}) the delay between every message's sending and receipt as well as every movement's start and end must be positive, and (\textit{ii}) the time taken by every operation execution---and, by extension, every action execution---must be finite.
W.l.o.g., we may assume that any two events either occur simultaneously or are at least one time unit apart.
We also assume, w.l.o.g., that at any time before the asynchronous schedule has terminated, there is at least one active amoebot; note that any positive delay during which all amoebots are inactive could be truncated so that the last action execution's end time coincides with the next action execution's start time without changing the system configuration.
In addition to timing, an asynchronous schedule specifies the operations executed, all messages' contents, and variable values accessed and updated; i.e., all details except private computations.

To ensure that an asynchronous schedule captures the actual system behavior of an amoebot system under an asynchronous adversary, we introduce the concept of validity.
An asynchronous schedule is \textit{valid} if there is an asynchronous execution of (enabled) actions producing the same events (w.r.t.\ timing and content) as in the given asynchronous schedule.
In the remainder of our analysis, whenever we refer to an asynchronous schedule, we assume its timing is in the control of an adversary constrained only by validity.

We begin with an observation that follows immediately from $\alg'$ and Convention~\ref{conv:phases}.

\begin{observation} \label{obs:isolation}
    Whenever an amoebot $B$ is locked by an amoebot $A$ in an execution of $\alg'$, only $A$ can initiate a movement with or update the public memory of $B$.
\end{observation}

Next, we identify the points of failure in action $\alpha'$ of $\alg'$.

\begin{lemma} \label{lem:failure}
    In an execution of action $\alpha'$, only the \Lock\ and \Expand\ operations can fail.
\end{lemma}
\begin{proof}
    The first operation $A$ executes is the \Lock\ operation which may fail, as claimed.
    If it fails, $\alpha'$ catches the \texttt{lock-failure} and aborts, so no further operations are executed.
    Supposing the initial \Lock\ operation succeeds, let $\mathcal{L}_A$ denote the set of amoebots locked by $A$.
    Recall from Section~\ref{subsubsec:operationscomms} that a \Read\ or \Write\ operation by $A$ can only fail if $A$ is accessing a variable in the public memory of an amoebot $B \neq A$ that is disconnected from $A$ during that operation's execution.
    By Observation~\ref{obs:isolation}, no amoebot in $\mathcal{L}_A$ can change its shape outside of a movement operation initiated by $A$.
    By inspection of $\alpha'$, $A$ only executes \Read\ and \Write\ operations involving amoebots in $\mathcal{L}_A$ and does so before its movement operation; thus, they must succeed.
    Finally, \Unlock\ operations cannot fail because they only involve locked amoebots, and \Connected\ operations always succeed.

    It remains to consider the movement operations, all of which are determined by the execution of an enabled action $\alpha \in \alg$.
    An \Expand\ operation may fail, as claimed.
    A \Contract\ operation by $A$ only fails if $A.\shape \neq \textsc{expanded}$ or $A$ is already involved in a handover.
    By Convention~\ref{conv:valid}, this contraction would succeed if all other amoebots were inactive, so $A$ must have been expanded when it evaluated the guard of $\alpha$.
    Action $\alpha'$ does not contain any operations that change the shape of $A$ between the guard evaluations and this \Contract\ operation, and by Observation~\ref{obs:isolation} no other action executions could have involved $A$ in a handover and changed its shape since $A \in \mathcal{L}_A$.
    Thus, $A$ is expanded and cannot be involved in a handover when starting this contraction, so the \Contract\ operation succeeds.
    
    A \Pull\ operation by $A$ with a neighbor $B$ only fails if $A.\shape \neq \textsc{expanded}$, $B.\shape \neq \textsc{contracted}$, $A$ and $B$ are not connected, or $A$ or $B$ is already involved in another handover.
    By Convention~\ref{conv:valid}, this pull handover would succeed if all other amoebots were inactive, so $A$ must have been expanded, $B$ must have been contracted, and $A$ and $B$ must have been connected when $A$ evaluated the guard of $\alpha$.
    Once again, $\alpha'$ does not contain any operations that change the shape of $A$ between the guard evaluations and this \Pull\ operation, and by Observation~\ref{obs:isolation} neither $A$ nor $B$ can be involved in another handover or could have changed their shape since $A, B \in \mathcal{L}_A$.
    So $A$ is expanded, $B$ is contracted, $A$ and $B$ are neighbors, and neither $A$ nor $B$ are involved in another handover when starting this pull handover, so the \Pull\ operation succeeds.
    An analogous argument holds for \Push\ operations.
    
    Therefore, in an execution of $\alpha'$, only the \Lock\ or \Expand\ operations can fail.
\end{proof}

We say that an amoebot is \textit{$\alg$-enabled} if it has at least one enabled action $\alpha \in \alg$ and is \textit{$\alg$-disabled} otherwise.
An execution of $\alpha'$ by an amoebot $A$ is \textit{relevant} in an asynchronous schedule of $\alg'$ if all its operations succeed and either $A.\awaken = \true$ or $A$ is $\alg$-enabled in $\alpha'$.
The next two lemmas show that we can \textit{sanitize} any asynchronous schedule of $\alg'$ by removing all events associated with \textit{irrelevant} executions of $\alpha'$---i.e., those with at least one failed operation or that have $A.\awaken = \false$ and are $\alg$-disabled---without changing the system's final configuration.

\begin{lemma} \label{lem:sanitizelocks}
    Let $\sched$ be any asynchronous schedule of $\alg'$ and let $\sched_L$ be the asynchronous schedule obtained from $\sched$ by removing all events except those associated with \Lock\ and \Unlock\ operations and successful movements.
    Then $\sched_L$ is valid w.r.t.\ its \Lock\ operations.
    Moreover, for any set $S$ of successful \Lock\ operations in $\sched_L$, the asynchronous schedule $\sched_S$ obtained from $\sched_L$ by removing all events associated with \Lock\ operations not in $S$ is valid and all \Lock\ operations of $S$ are successful and lock the same amoebots they did in $\sched_L$.
\end{lemma}
\begin{proof}
    First consider the asynchronous schedule $\sched_L$ containing the events of $\sched$ associated with all \Lock, \Unlock, and successful movement operations.
    The only movements included in $\sched$ that are not present in $\sched_L$ are failed expansions. 
    However, a failed expansion does not introduce any new connections or disconnections.
    Thus, since $\sched$ is valid w.r.t.\ its \Lock\ operations and the validity of a \Lock\ operation depends only on other \Lock\ operations and amoebots' connections, $\sched_L$ must also be valid w.r.t.\ its \Lock\ operations.
    
    Consider any set $S$ of successful \Lock\ operations in $\sched_L$ and let $\sched_S$ be the asynchronous schedule obtained from $\sched_L$ by removing all events associated with \Lock\ operations not in $S$.
    Since any \Lock\ operation in $S$ was successful in $\sched_L$, then by the \Lock\ operation's mutual exclusion property all amoebots $A$ it attempted to lock must have had $A.\lock = \bot$.
    Removing other \Lock\ operations cannot cause amoebots' $\lock$ variables to be $\neq \bot$.
    Thus, \Lock\ operations in $S$ remain successful in $\sched_S$.
    Any \Lock\ operation in $S$ must also lock the same amoebots in $\sched_S$ as it did in $\sched_L$ since this depends only on connectivity---not on other \Lock\ operations---and the movement operations that control connectivity are identical in $\sched_S$ and $\sched_L$.
    Therefore, since $\sched_L$ is valid w.r.t.\ its \Lock\ operations, so must $\sched_S$.
\end{proof}

\begin{lemma} \label{lem:sanitize}
    Let $\sched$ be any asynchronous schedule of $\alg'$ and let $\sched^*$ be its sanitized version keeping only the events associated with relevant executions of $\alpha'$ in $\sched$.
    Then $\sched^*$ is a valid asynchronous schedule that changes the system configuration exactly as $\sched$ does except w.r.t.\ amoebots' activity bits, which have the property that the set of amoebots $A$ with $A.\activity = \true$ in $\sched^*$ is always a superset of those in $\sched$.
\end{lemma}
\begin{proof}
    By Lemma~\ref{lem:failure}, only \Lock\ or \Expand\ operations can fail in an execution of $\alpha'$, implying three types of irrelevant executions of $\alpha'$ by an amoebot $A$: those whose \Lock\ operation fails, those whose \Lock\ operation succeeds but that have $A.\awaken = \false$ and are $\alg$-disabled, and those whose \Lock\ operation succeeds but whose \Expand\ operation fails.
    By Lemma~\ref{lem:sanitizelocks}, when removing events associated with irrelevant executions of $\alpha'$ from $\sched$ to obtain $\sched^*$, all successful \Lock\ operations in $\sched$ remain valid and successful in $\sched^*$ and lock the same amoebots as they did in $\sched$.
    Thus, the only change an irrelevant execution of $\alpha'$ could have made is setting an $\alg$-disabled amoebot's activity bit to \false.
    This implies that the set of amoebots $A$ with $A.\activity = \true$ in $\sched^*$ is always a superset of those in $\sched$ and thus any relevant action execution of $\alpha'$ in $\sched$ remains enabled in $\sched^*$.
    
    Since relevant executions of $\alpha'$ only issue \Read\ and \Write\ operations to the executing amoebot or its locked neighbors, the success and identical outcome of all \Lock\ operations in $\sched^*$ ensures that all \Read\ and \Write\ operations in $\sched^*$ also succeed.
    Moreover, because irrelevant executions of $\alpha'$ never perform \Write\ operations, all \Read\ and \Write\ operations in $\sched^*$ must access or update the same variable values as they did in $\sched$ since the event timing is preserved.
    \Connected\ operations in $\sched^*$ are also guaranteed to return the same results as in $\sched$ since failed \Expand\ operations discarded from $\sched$ do not change amoebot connectivity.
    
    It remains to show that all movement operations in $\sched^*$ are successful.
    Any \Contract, \Pull, or \Push\ operations in $\sched^*$ must have succeeded in $\sched$, implying that they were unaffected by any failed \Expand\ operations in $\sched$ that are now removed.
    So the only movement operations in $\sched^*$ that could have interacted with failed \Expand\ operations in $\sched$ are concurrent \Expand\ operations that contended with failed \Expand\ operations for the same nodes.
    But the fact that these \Expand\ operations are in $\sched^*$ implies that they succeeded in $\sched$, and thus must also succeed in $\sched^*$ when all contending expansions are removed.
\end{proof}

It thus suffices to study algorithm $\alg'$ under sanitized asynchronous schedules.
Our next goal is to map any sanitized asynchronous schedule $\sched$ of $\alg'$ to a sequential schedule that produces the same final system configuration as $\sched$.
Any asynchronous schedule already totally orders the updates to any single variable in an amoebot's public memory and the occupancy of any single node; here, we focus on ordering entire action executions.
Denote the (relevant) executions of $\alpha'$ in $\sched$ as pairs $(A_i, \alpha_i')$, where amoebot $A_i$ executes $\alpha_i'$.
Construct a directed graph $D$ with nodes $\{(A_1, \alpha_1'), \ldots, (A_k, \alpha_k')\}$ representing all executions of $\alpha'$ in $\sched$ and directed edges $(A_i, \alpha_i') \to (A_j, \alpha_j')$ for $i \neq j$ if and only if one of the following hold:
\begin{enumerate}
    \item Both $(A_i, \alpha_i')$ and $(A_j, \alpha_j')$ lock some amoebot $B$ in their \Lock\ operations and $(A_j, \alpha_j')$ is the first execution to lock $B$ after $B$ is unlocked by $(A_i, \alpha_i')$.  \label{dag:locksame}

    \item The nodes occupied by $A_i$ at the start of $\alpha_i'$ and by $A_j$ at the start of $\alpha_j'$ are adjacent and $(A_i, \alpha_i')$ is the last execution of $A_i$ to execute an action of $\alg$ before the \Lock\ operation of $(A_j, \alpha_j')$ completes.  \label{dag:overlap}
    
    \item $(A_j, \alpha_j')$ is the first execution to \Expand\ into some node $v$ after $v$ is vacated by a \Contract\ operation in $(A_i, \alpha_i')$.  \label{dag:move}
\end{enumerate}

\begin{lemma} \label{lem:dag}
    The directed graph $D$ corresponding to the executions of $\alpha'$ in a sanitized asynchronous schedule of $\alg'$ is a directed, acyclic graph (DAG).
\end{lemma}
\begin{proof}
    We will show that for any edge $(A_i, \alpha_i') \to (A_j, \alpha_j')$ in $D$, $(A_i, \alpha_i')$ completes its \Lock\ operation before $(A_j, \alpha_j')$ does.
    This immediately implies that $D$ is acyclic; otherwise, the \Lock\ operations of any two executions in a cycle of $D$ must complete both before and after each other, a contradiction.
    
    First suppose that $(A_i, \alpha_i') \to (A_j, \alpha_j')$ is an edge in $D$ by Rule~\ref{dag:locksame}, i.e., both executions lock an amoebot $B$ and $(A_j, \alpha_j')$ is the first execution to lock $B$ after $B$ is unlocked by $(A_i, \alpha_i')$.
    Clearly, $A_j$ can only lock $B$ after $A_i$ has unlocked $B$ and $A_i$ can only unlock $B$ after it locks $B$ in its own \Lock\ operation.
    Since these operations all involve message transfers requiring positive time, $(A_i, \alpha_i')$ must complete its \Lock\ operation before $(A_j, \alpha_j')$ does.
    
    Next suppose that $(A_i, \alpha_i') \to (A_j, \alpha_j')$ is an edge in $D$ by Rule~\ref{dag:overlap}, i.e., the nodes occupied by $A_i$ and $A_j$ at the start of their respective actions are adjacent and $(A_i, \alpha_i')$ is the last execution of $A_i$ to execute an action of $\alg$ before the \Lock\ operation of $(A_j, \alpha_j')$ completes.
    Any execution of an action of $\alg$ in $(A_i, \alpha_i')$ must start after its \Lock\ operation completes; thus, $(A_i, \alpha_i')$ must complete its \Lock\ operation before $(A_j, \alpha_j')$ does.
    
    Finally, suppose that $(A_i, \alpha_i') \to (A_j, \alpha_j')$ is an edge in $D$ by Rule~\ref{dag:move}, i.e., $(A_j, \alpha_j')$ is the first execution to \Expand\ into some node $v$ after $v$ is vacated by a \Contract\ operation in $(A_i, \alpha_i')$.
    It suffices to consider the case where the \Lock\ operation of $(A_i, \alpha_i')$ does not lock $A_j$; otherwise, there exists a directed path of Rule~\ref{dag:locksame} edges from $(A_i, \alpha_i')$ to $(A_j, \alpha_j')$ in $D$ and the first case proves the claim.
    For $(A_i, \alpha_i')$ to not lock $A_j$, $A_j$ cannot be a neighbor of $A_i$ at the time the \Lock\ operation of $(A_i, \alpha_i')$ starts.
    We know that the \Lock\ operation of $(A_i, \alpha_i')$ is successful, so $A_i$ is locked and occupies $v$ until the start of its \Contract\ operation out of $v$.
    But $A_j$ must occupy a node adjacent to $v$ at the start of $(A_j, \alpha_j')$ and must succeed in its own \Lock\ operation in order to \Expand\ into $v$.
    Thus, the \Lock\ operation of $(A_j, \alpha_j')$ cannot complete until after $A_i$ has started contracting out of $v$, which occurs strictly after the completion of the \Lock\ operation of $(A_i, \alpha_i')$.
\end{proof}

The following lemma compares the outcome of any sanitized asynchronous schedule of $\alg'$ to a schedule where an execution of $\alpha'$ corresponding to a sink in the DAG $D$ is removed.
A property that will become important shortly is whether any \Expand\ operation in the removed execution of $\alpha'$ could still be executed if it were placed in a different schedule.
Formally, we say a system configuration $C$ is \textit{expansion-compatible} with an execution $(A_i, \alpha_i')$ if either $(A_i, \alpha_i')$ does not perform an \Expand\ operation or the \Expand\ operation executed by $A_i$ in $\alpha_i'$ would succeed in $C$.

\begin{lemma}  \label{lem:serialremove}
    Consider any sanitized asynchronous schedule $\sched$ of $\alg'$ and let $(A_i, \alpha_i')$ be any sink in the corresponding DAG $D$.
    Let $\sched_i^-$ be the asynchronous schedule obtained by removing all events associated with $(A_i, \alpha_i')$ from $\sched$.
    Then $\sched_i^-$ is valid and the final configuration reached by $\sched^-_i$ is expansion-compatible with $(A_i, \alpha_i')$ and identical to that of $\sched$ except for the amoebots locked by $(A_i, \alpha_i')$ in $\sched$, which appear exactly as they did just after the \Lock\ operation of $(A_i, \alpha_i')$ completed in $\sched$.
\end{lemma}
\begin{proof}
    If $(A_i, \alpha_i')$ is the only execution in $\sched$, then its removal yields an empty schedule $\sched_i^-$ that trivially satisfies the lemma.
    So consider any action execution $(A_j, \alpha_j')$ in $\sched$ with $j \neq i$.
    We first show that $(A_j, \alpha_j')$ must remain enabled in $\sched_i^-$; i.e., $A_j.\activity = \true$ at the time of this execution.
    This must have been the case in $\sched$, so the only way for $(A_j, \alpha_j')$ to not be enabled in $\sched_i^-$ is if $(A_i, \alpha_i')$ was responsible for enabling it in $\sched$.
    But $(A_i, \alpha_i')$ could only have updated $A_j.\activity$ to \true\ if $(A_i, \alpha_i')$ locked $A_j$, implying that $(A_j, \alpha_j')$ could not have started until after $(A_i, \alpha_i')$ unlocked $A_j$.
    Thus, there must exist a directed path of Rule~\ref{dag:locksame} edges in $D$ from $(A_i, \alpha_i')$ to $(A_j, \alpha_j')$, contradicting the assumption that $(A_i, \alpha_i')$ is a sink.

    We next show that $(A_j, \alpha_j')$ remains valid in $\sched_i^-$.
    Let $\mathcal{L}_j(\sched)$ (resp., $\mathcal{L}_j(\sched_i^-)$) denote the set of amoebots locked by $(A_j, \alpha_j')$ in $\sched$ (resp., in $\sched_i^-$); we begin by showing $\mathcal{L}_j(\sched) = \mathcal{L}_j(\sched_i^-)$.
    First suppose that there is a directed path in $D$ from $(A_j, \alpha_j')$ to $(A_i, \alpha_i')$.
    By the proof of Lemma~\ref{lem:dag}, $(A_j, \alpha_j')$ must complete its \Lock\ operation before $(A_i, \alpha_i')$ does, implying that $(A_j, \alpha_j')$ completes its \Lock\ operation before $(A_i, \alpha_i')$ completes any operation.
    Since the timing of $(A_j, \alpha_j')$ in $\sched$ is preserved in $\sched^-_i$, it follows that $\mathcal{L}_j(\sched) = \mathcal{L}_j(\sched_i^-)$.
    Now suppose that there is no directed path in $D$ from $(A_j, \alpha_j')$ to $(A_i, \alpha_i')$.
    Then the amoebots locked by $(A_j, \alpha_j')$ and $(A_i, \alpha_i')$ in $\sched$ must be disjoint by Rule~\ref{dag:locksame}, so certainly $(A_j, \alpha_j')$ can lock any amoebot in $\sched_i^-$ that it did in $\sched$; i.e., $\mathcal{L}_j(\sched) \subseteq \mathcal{L}_j(\sched_i^-)$.
    But suppose to the contrary that $(A_j, \alpha_j')$ is able to lock some additional amoebot $B$ in $\sched_i^-$ that it did not lock in $\sched$.
    This is only possible if $(A_i, \alpha_i')$ caused $B$ to move out of the neighborhood of $A_j$ in $\sched$, either directly via a handover or indirectly by enabling some action of $B$ involving a movement.
    In either case, $A_i$ must have locked $B$ before $A_j$ did, implying the existence of a directed path of Rule~\ref{dag:locksame} edges in $D$ from $(A_i, \alpha_i')$ to $(A_j, \alpha_j')$.
    This once again contradicts the assumption that $(A_i, \alpha_i')$ is a sink.
    So in any case, $(A_j, \alpha_j')$ locks the same set of amoebots in $\sched$ and $\sched_i^-$.
    
    After completing its \Lock\ operation, $(A_j, \alpha_j')$ does one of two things.
    If $A_j.\awaken = \true$, then it updates the activity bits of all the amoebots it locked to \true, updates its own awaken bit to \false, releases its locks, and aborts.
    Since $\mathcal{L}_j(\sched) = \mathcal{L}_j(\sched_i^-)$ and timing is preserved, these updates occur identically in $\sched$ and $\sched_i^-$.
    
    Otherwise, if $A_j.\awaken = \false$, $A_j$ evaluates the guards of actions in $\alg$; recall that these depend only on the positions, shapes, and public memories of the locked amoebots.
    Suppose to the contrary that there is an amoebot $B$ locked by $A_j$ whose position, shape, or public memory is different in $\sched_i^-$ than it was in $\sched$.
    Then $(A_i, \alpha_i')$ must have locked $B$ to perform the corresponding update in $\sched$, implying that there is a directed path of Rule~\ref{dag:locksame} edges from $(A_i, \alpha_i')$ to $(A_j, \alpha_j')$ in $D$, contradicting the assumption that $(A_i, \alpha_i')$ is a sink.
    So the outcomes of the guard evaluations must be identical in $\sched$ and $\sched_i^-$.
    
    Since $(A_j, \alpha_j')$ is in the sanitized schedule $\sched$, it must be relevant, and thus $A_j$ must be $\alg$-enabled in $\alpha_j'$.
    Whichever enabled action of $\alg$ is executed, any \Write, \Contract, \Pull, or \Push\ operations involved must occur identically in $\sched$ and $\sched_i^-$ since the locked amoebots and their positions, shapes, and public memories are the same in both schedules.
    The only remaining possibility is that $(A_j, \alpha_j')$ causes $A_j$ to \Expand\ into an adjacent node $v$ in $\sched$ that is occupied in $\sched_i^-$, causing the \Expand\ operation to fail in $\sched_i^-$.
    This implies that $(A_i, \alpha_i')$ causes $A_i$ to \Contract\ out of $v$.
    But then $(A_i, \alpha_i') \to (A_j, \alpha_j')$ must be a Rule~\ref{dag:move} edge in $D$, contradicting the assumption that $(A_i, \alpha_i')$ is a sink.
    Thus, we conclude that $\sched_i^-$ is valid and all action executions $(A_j, \alpha_j')$ for which $j \neq i$ execute identically in $\sched$ and $\sched_i^-$.
    
    Next, we show that the final configuration $C_i^-$ reached by $\sched_i^-$ is expansion-compatible with $(A_i, \alpha_i')$.
    Suppose to the contrary that $(A_i, \alpha_i')$ performs a successful \Expand\ operation into a node $v$ in $\sched$ but the same expansion would fail in $C_i^-$.
    This is only possible if $v$ is occupied by another amoebot in $C_i^-$.
    Since all executions other than $(A_i, \alpha_i')$ are valid and execute identically in $\sched$ and $\sched_i^-$, another amoebot can only have come to occupy $v$ in $C_i^-$ if $A_i$ vacated $v$ in some later execution in $\sched$.
    But $A_i$ can only change its shape if it is locked, contradicting the assumption that $(A_i, \alpha_i')$ is a sink in $D$ by Rule~\ref{dag:locksame}.
    So $v$ must be unoccupied in $C_i^-$ and thus $C_i^-$ is expansion-compatible with $(A_i, \alpha_i')$.
    
    It remains to show that the amoebots in $\mathcal{L}_i(\sched)$---i.e., those locked by $(A_i, \alpha_i')$ in $\sched$---appear in $\sched_i^-$ exactly as they did after the \Lock\ operation of $(A_i, \alpha_i')$ in $\sched$.
    But this follows immediately from the assumption that $(A_i, \alpha_i')$ is a sink: for an execution $(A_j, \alpha_j')$ with $j \neq i$ to change the position, shape, or public memory of an amoebot $B \in \mathcal{L}_i(\sched)$, it would first have to lock $B$, implying that $(A_i, \alpha_i') \to (A_j, \alpha_j')$ is a directed edge in $D$.
\end{proof}

Lemma~\ref{lem:serialremove} allows us to prove the following central result.
Here, we consider the expansion-robust variant $\alg^E$ of $\alg$ and the algorithm $(\alg^E)'$ produced from $\alg^E$ by the concurrency control framework.
We denote the sole action of $(\alg^E)'$ as $(\alpha^E)'$.
Given the initial configuration $C_0$ of $\alg$, configuration $C_0^E$ is its extension with expand flags $A.\xflag_p = \false$ for all amoebots $A$ and ports $p$; the initial configuration $(C_0^E)'$ of $(\alg^E)'$ further extends $C_0^E$ by adding $A.\activity = \true$ and $A.\awaken = \false$ for all amoebots $A$.

\begin{lemma} \label{lem:serializability}
    For any finite sanitized asynchronous schedule $\sched$ of $\alg'$ starting in $C_0'$, there exists a sequential schedule of $(\alg^E)'$ starting in $(C_0^E)'$ that reaches a final configuration that is identical to that of $\sched$, modulo amoebots' expand flags, with the exception that the set of amoebots $A$ with $A.\activity = \true$ or $A.\awaken = \true$ is a superset of those in the final configuration reached by $\sched$.
\end{lemma}
\begin{proof}
    Consider any finite sanitized asynchronous schedule $\sched$ of $\alg'$ starting in $C_0'$ and let $D$ be its corresponding DAG (Lemma~\ref{lem:dag}).
    We argue by induction on $k$, the number of executions of $\alpha'$ in $\sched$, that any sequential ordering of the executions of $\alpha'$ in $\sched$ consistent with a topological ordering of $D$ can be extended to a sequential schedule $\bar{\sched}$ of $(\alg^E)'$ starting in $(C_0^E)'$ satisfying the lemma.
    Specifically, we construct $\bar{\sched}$ by replacing executions of $\alpha'$ in $\sched$ that execute some action $\alpha_i \in \alg$ with corresponding executions of $(\alpha^E)'$ that execute action $\alpha_i^E \in \alg^E$.
    We then suitably pad $\bar{\sched}$ with executions of $(\alpha^E)'$ that execute action $\alpha_0^E$ (as defined in Algorithm~\ref{alg:expandrobust}) so that the set of amoebots $A$ with $A.\activity = \true$ or $A.\awaken = \true$ in the final configuration reached by $\bar{\sched}$ is a superset of those in the final configuration reached by $\sched$.
    
    The lemma trivially holds for $k = 0$, so suppose the lemma holds for any sanitized asynchronous schedule of $\alg'$ with $k \geq 0$ executions of $\alpha'$.
    Let $\sched$ be any sanitized asynchronous schedule of $\alg'$ starting in $C_0'$ consisting of $k + 1$ executions of $\alpha'$, let $C$ be the final configuration it reaches, and let $(A_i, \alpha_i')$ be any sink in the corresponding DAG $D$.
    By Lemma~\ref{lem:serialremove}, the sanitized asynchronous schedule $\sched_i^-$ obtained by removing all events associated with $(A_i, \alpha_i')$ from $\sched$ is valid and reaches a final configuration $C_i^-$ that is expansion-compatible with $(A_i, \alpha_i')$ and is identical to $C$ except for the amoebots locked by $(A_i, \alpha_i')$ in $\sched$, which appear exactly as they did just after the \Lock\ operation of $(A_i, \alpha_i')$ completed in $\sched$.
    By the induction hypothesis, there exists a sequential schedule $\bar{\sched}_i$ of $(\alg^E)'$ starting in $(C_0^E)'$ that reaches a final configuration $\bar{C}_i$ identical to $C_i^-$ (modulo amoebots' expand flags) with the exception that the set of amoebots $A$ with $A.\activity = \true$ or $A.\awaken = \true$ in $\bar{C}_i$ is a superset of those in $C_i^-$.
    This implies that $(A_i, (\alpha^E)_i')$ is enabled in $\bar{C}_i$ since $(A_i, \alpha_i')$ was enabled in $C_i^-$ and they have the same guard: $A_i.\activity = \true$.
    
    The amoebots $\mathcal{L}_i(\sched)$ locked by $(A_i, \alpha_i')$ in $\sched$ must still be neighbors of $A_i$ in $\bar{C}_i$ (i.e., at the end of $\bar{\sched}_i$) by Lemma~\ref{lem:serialremove} and the induction hypothesis, but $A_i$ may also have additional neighbors in $\bar{C}_i$ that were not originally present at the time of its \Lock\ operation in $\sched$.
    Thus, we have $\mathcal{L}_i(\sched) \subseteq \mathcal{L}_i(\bar{\sched}_i)$.
    There are three cases for the behavior of $(A_i, \alpha_i')$; in each, we construct a sequential schedule $\bar{\sched}$ by combining $\bar{\sched}_i$, the execution $(A_i, (\alpha^E)_i')$, and possibly additional executions of $(\alpha^E)'$ involving $\alpha_0^E$ whose final configuration satisfies the lemma.
    
    \smallskip
    
    \noindent \textit{Case 1.} $A_i.\awaken = \true$ both at the start of $(A_i, \alpha_i')$ in $\sched$ and at the end of $\bar{\sched}_i$.
    Let $\bar{\sched}$ be the sequential schedule obtained by appending $(A_i, (\alpha^E)_i')$ to the end of $\bar{\sched}_i$.
    Then in $\bar{\sched}$, $(A_i, (\alpha^E)_i')$ updates $B.\activity$ to $\true$ for all amoebots $B$ that it locks, updates $A_i.\awaken$ to $\false$, releases its locks, and aborts---just as $(A_i, \alpha_i')$ does in $\sched$.
    Since $\mathcal{L}_i(\sched) \subseteq \mathcal{L}_i(\bar{\sched}_i)$, the only difference between the final configurations of $\sched$ and $\bar{\sched}$ (other than amoebots' expand flags) is that the latter may have additional amoebots with their activity or awaken bits set to \true, so the lemma holds.
    
    \smallskip
    
    \noindent \textit{Case 2.} $A_i.\awaken = \false$ at the start of $(A_i, \alpha_i')$ in $\sched$ but $A_i.\awaken = \true$ at the end of $\bar{\sched}_i$.
    Let $\bar{\sched}$ be the sequential schedule obtained by activating $A_i$ twice at the end of $\bar{\sched}_i$.
    The first activation has the same effect as Case 1, potentially yielding more amoebots with their activity or awaken bits set to \true.
    It also resets the awaken bit of $A_i$, yielding $A_i.\awaken = \false$ in both $\sched$ and $\bar{\sched}_i + (A_i, (\alpha^E)_i')$.
    We address this in the following case.
    
    \smallskip
    
    \noindent \textit{Case 3.} $A_i.\awaken = \false$ both at the start of $(A_i, \alpha_i')$ in $\sched$ and at the end of $\bar{\sched}_i$.
    Since $(A_i, \alpha_i')$ is an execution of $\sched$, a sanitized schedule, we know that $(A_i, \alpha_i')$ is relevant and thus must have an action $\alpha_j \in \alg$ in $\sched$ that is enabled by the amoebots $\mathcal{L}_i(\sched)$ locked in $\alpha_i'$.
    Intuitively, we would like to construct the sequential schedule $\bar{\sched}$ by appending $(A_i, (\alpha^E)_i')$ to the end of $\bar{\sched}_i$, where the execution of $(\alpha^E)_i'$ involves the corresponding action $\alpha_j^E \in \alg^E$.
    However, because $\bar{\sched}_i$ involves expand flags and $\mathcal{L}_i(\sched) \subseteq \mathcal{L}_i(\bar{\sched}_i)$, it is not immediately obvious that $\alpha_j^E$ is enabled at the end of $\bar{\sched}_i$ and can be executed to satisfy the lemma.
    
    To this end, we first show that any amoebot $A_\ell \in \mathcal{L}_i(\bar{\sched}_i) \setminus \mathcal{L}_i(\sched)$---i.e., those locked by $(A_i, (\alpha^E)_i')$ at the end of $\bar{\sched}_i$ but not by $(A_i, \alpha_i')$ in $\sched$---would be ignored in any guard evaluation and execution of $\alpha_j^E$ at the end of $\bar{\sched}_i$ due to expand flags.
    Such an $A_\ell$ can only exist if there was some time $t$ during the \Lock\ operation of $(A_i, \alpha_i')$ in $\sched$ at which $A_i$ and $A_\ell$ were not connected.
    But $A_\ell \in \mathcal{L}_i(\bar{\sched}_i)$ implies that $A_\ell$ later became a neighbor of $A_i$, so consider the first event in $\sched$ after time $t$ at which $A_i$ and $A_\ell$ are connected.
    This event must correspond to $A_\ell$ completing an expansion or handover and connecting to $A_i$, so in the corresponding action execution in $\bar{\sched}_i$, $A_\ell$ must have updated the expand flag of any new port now connected to $A_i$ to \true\ (see Lines~\ref{alg:expandrobust:handover1},~\ref{alg:expandrobust:handover2}, and~\ref{alg:expandrobust:expand} of Algorithm~\ref{alg:expandrobust}).
    Any such expand flag can only be reset to \false\ in $\bar{\sched}_i$ if $A_\ell$ or $A_i$ execute another action in $\sched$ after their connection event (see Line~\ref{alg:expandrobust:resetown} or~\ref{alg:expandrobust:resetnbr} of Algorithm~\ref{alg:expandrobust}, respectively).
    But $A_\ell$ (resp., $A_i$) cannot execute another action in $\sched$ because Rule~\ref{dag:overlap} (resp., Rule~\ref{dag:locksame}) of the DAG $D$ would imply $(A_i, \alpha_i')$ is not a sink in $D$, a contradiction.
    Thus, any port $p$ of any amoebot $A_\ell \in \mathcal{L}_i(\bar{\sched}_i) \setminus \mathcal{L}_i(\sched)$ connected to $A_i$ must have $A_\ell.\xflag_p = \true$ at the end of $\bar{\sched}_i$.
    
    We have established that if $\alpha_j \in \alg$ was enabled in $\sched$ for execution $(A_i, \alpha_i')$, then any additional neighbors $\mathcal{L}_i(\bar{\sched}_i) \setminus \mathcal{L}_i(\sched)$ locked by $A_i$ at the end of $\bar{\sched}_i$ cannot cause $\alpha_j^E \in \alg^E$ to be disabled because their expand flags are \true\ and they are thus ignored.
    However, we must show the opposite for any original neighbor $A_\ell \in \mathcal{L}_i(\sched)$ at the end of $\bar{\sched}_i$, i.e., that its expand flags do not cause $A_i$ to ignore it and thus possibly disable $\alpha_j^E$.
    This situation can be easily prevented using the $\alpha_0^E \in \alg^E$ action as follows.
    For any port $p$ of any amoebot $A_\ell \in \mathcal{L}_i(\sched)$ connected to $A_i$ with $A_\ell.\xflag_p = \true$, append an execution $(A_\ell, (\alpha^E)_\ell')$ to the end of $\bar{\sched}_i$ that involves an execution of $\alpha_0^E$ resetting $A_\ell.\xflag_p$ to \false.
    Complete the construction of the desired sequential schedule $\bar{\sched}$ by appending the execution $(A_i, (\alpha^E)_i')$ involving the execution of $\alpha_j^E$.
    We have shown that this final execution in $\bar{\sched}$ considers exactly the same neighborhood as $(A_i, \alpha_i')$ did in $\sched$, and thus because $\alpha_j$ is enabled in $\sched$ so is $\alpha_j^E$ in $\bar{\sched}$.
    This further guarantees that their respective executions make the same updates to the original variables of $\alg$ and, by the expansion-compatibility ensured by Lemma~\ref{lem:serialremove}, the same movement.
    The only differences between the final configurations reached by $\sched$ and $\bar{\sched}$ are (\textit{i}) the added executions $(A_\ell, (\alpha^E)_\ell')$ in $\bar{\sched}$ for executing $\alpha_0^E$ might set additional activity and awaken bits to \true, and (\textit{ii}) the final execution $(A_i, (\alpha^E)_i')$ involving the execution of $\alpha_j^E$ will set the activity bit of any amoebot in $\mathcal{L}_i(\bar{\sched}_i) \setminus \mathcal{L}_i(\sched)$ to \true.
    But these differences are exactly those allowed by the lemma, completing the induction.
\end{proof}

We now turn to the analysis of $\alg'$ under sequential executions.
Define a \textit{sequential schedule} $\sched = ((A_1, \alpha_1), (A_2, \alpha_2), \ldots)$ as the sequence of actions executed in a sequential execution, where $\alpha_i$ is the $i$-th action of $\alg$ executed by the system and $A_i$ is the amoebot that executed it.
For a sequential schedule to be \textit{valid}, $\alpha_i$ must be enabled for $A_i$ in the configuration produced by executions $(A_1, \alpha_1), \ldots, (A_{i-1}, \alpha_{i-1})$, for all $i \geq 1$.
Certainly, sequential schedules obfuscate various details that were made explicit in asynchronous schedules; e.g., they ignore the precise timing of message transmissions and movements.
Although a single sequential schedule may in fact represent many possible sequential executions, this abstraction suffices for our purposes because the resulting system configurations are well-defined.

We first argue that sequential executions of $\alg'$ terminate.

\begin{lemma} \label{lem:seqfinite}
    If every sequential schedule of $\alg$ starting in $C_0$ is finite, then every sequential schedule of $\alg'$ starting in $C_0'$ is also finite.
\end{lemma}
\begin{proof}
    Suppose to the contrary that there exists an infinite sequential schedule $\sched$ of $\alg'$ starting in configuration $C_0'$.
    When ignoring the handling of amoebots' activity and awaken bits, any execution of action $\alpha'$ of $\alg'$ either makes no change to the system configuration or makes changes identical to those of some action $\alpha \in \alg$.
    First suppose that $\sched$ contains an infinite number of executions of $\alpha'$ executing actions of $\alg$.
    Then by constructing a sequential schedule comprising only these $\alg$ action executions, we obtain an infinite schedule of $\alg$ starting in $C_0$, a contradiction.
    
    Suppose instead that $\sched$ contains only a finite number of executions of $\alpha'$ executing actions of $\alg$.
    Since there are only a finite number of such executions, there must exist a time $t$ after which no amoebot is $\alg$-enabled and the remaining infinite executions of $\alpha'$ only involve updates to amoebots' activity and awaken bits.
    Any activation of an amoebot $A$ with $A.\awaken = \true$ results in $A$ setting the activity bits of its neighbors to \true---of which there can be at most a finite number $\Delta$ that depends on the assumed space variant---and resetting $A.\awaken$ to \false\ (Steps~\ref{alg:framework:awaken}--\ref{alg:framework:awakenabort}).
    Otherwise, an activation of $A$ with $A.\awaken = \false$ must result in $A$ resetting $A.\activity$ to \false\ since it is not $\alg$-enabled (Step~\ref{alg:framework:disabledabort}).
    Then the potential function $\Phi(C) = \sum_A (I_{A.\activity} + (\Delta + 1)I_{A.\awaken})$ over system configurations $C$ where $I_{A.\activity} \in \{0, 1\}$ (resp., $I_{A.\awaken} \in \{0, 1\}$) is equal to $1$ if and only if $A.\activity = \true$ (resp., $A.\awaken = \true$) is both lower bounded by $0$ and strictly decreasing after time $t$.
    Therefore, $\sched$ can only contain a finite number of executions of $\alpha'$ only involving updates to amoebots' activity and awaken bits, a contradiction of $\sched$ being infinite.
\end{proof}

We next establish a crucial property for characterizing configurations reachable by $\alg'$.

\begin{lemma} \label{lem:seqenabled}
    Consider any sequential schedule $\sched$ of $\alg'$ starting in $C_0'$.
    Any amoebot that is $\alg$-enabled in the final configuration reached by $\sched$ either (\textit{i}) is $\alg'$-enabled or (\textit{ii}) has an $\alg'$-enabled neighbor $B$ with $B.\awaken = \true$.
\end{lemma}
\begin{proof}
    Argue by induction on the length of $\sched = ((A_1, \alpha_1), \ldots, (A_k, \alpha_k))$.
    If $k = 0$, then the lemma trivially holds since all amoebots $A$ initially have $A.\activity = \true$ in $C_0'$ and thus are all $\alg'$-enabled.
    So suppose the lemma holds for schedules of $\alg'$ starting in $C_0'$ with any length $k \geq 0$, and consider any such schedule $\sched_{k+1} = ((A_1, \alpha_1), \ldots, (A_{k+1}, \alpha_{k+1}))$ with length $k + 1$.
    For $1 \leq i \leq k+1$, let $C_i'$ be the final configuration reached by the subschedule $\sched_i = ((A_1, \alpha_1), \ldots, (A_i, \alpha_i))$ of $\sched_{k+1}$.
    Consider any $\alg$-enabled amoebot $A$ in $C_{k+1}'$.
    
    We first suppose that $A$ was also $\alg$-enabled in $C_k'$.
    By the induction hypothesis, there are two cases to consider.
    If $A$ is $\alg'$-enabled in $C_k'$, then the only scenario in which $A.\activity$ is updated to \false\ is if $A = A_{k+1}$ and $A$ is not $\alg$-enabled (Step~\ref{alg:framework:disabledabort}), contrary to our supposition.
    So $A$ must also be $\alg'$-enabled in $C_{k+1}'$, satisfying (\textit{i}).
    Otherwise, $A$ must have an $\alg'$-enabled neighbor $B$ with $B.\awaken = \true$ in $C_k'$.
    The only scenario in which $B.\awaken$ is updated to \false\ is if $B = A_{k+1}$ and $B$ sets all of its neighbors' activity bits, including that of $A$, to \true\ (Steps~\ref{alg:framework:awaken}--\ref{alg:framework:awakenabort}).
    So either $B$ satisfies (\textit{ii}) by remaining an $\alg'$-enabled neighbor with $B.\awaken = \true$ or $A$ is $\alg'$-enabled in $C_{k+1}'$, satisfying (\textit{i}).
    
    Now suppose that $A$ was not $\alg$-enabled in $C_k'$; i.e., the execution of action $\alpha_{k+1}$ by amoebot $A_{k+1}$ causes a change in the neighborhood of $A$ such that $A$ becomes $\alg$-enabled in $C_{k+1}'$.
    Note that because $A$ was not $\alg$-enabled in $C_k'$, we must have $A_{k+1} \neq A$.
    If $A$ and $A_{k+1}$ were neighbors in $C_k'$, then $A_{k+1}$ must update $A.\activity$ to \true\ during its execution of $\alpha_{k+1}$ (Step~\ref{alg:framework:active}), satisfying (\textit{i}).
    Otherwise, if $A$ and $A_{k+1}$ were not neighbors in $C_k'$, there are still two ways $A_{k+1}$ could change the neighborhood of $A$ by executing $\alpha_{k+1}$.
    First, $A_{k+1}$ could move into the neighborhood of $A$ via an \Expand\ or \Push; in this case, $A_{k+1}$ remains $\alg'$-enabled and updates its own awaken bit to \true\ (Steps~\ref{alg:framework:expandawaken} and~\ref{alg:framework:pushawaken}), satisfying (\textit{ii}).
    Second, $A_{k+1}$ could update the memory of a neighbor $B$ of $A$ via a \Write; in this case, $A_{k+1}$ must also update $B.\activity$ and $B.\awaken$ to \true\ (Steps~\ref{alg:framework:active} and~\ref{alg:framework:writes}), also satisfying (\textit{ii}).
\end{proof}

The following lemma concludes our analysis of sequential executions.

\begin{lemma} \label{lem:seqoutcome}
    For any configuration $C'$ in which some sequential execution of $\alg'$ starting in $C_0'$ terminates, there exists a sequential execution of $\alg$ starting in $C_0$ that terminates in a configuration $C$ identical to $C'$, modulo activity and awaken bits.
\end{lemma}
\begin{proof}
    Consider any valid sequential schedule $\sched'$ of $\alg'$ starting in $C_0'$ under which $\alg'$ terminates and let $C'$ be the configuration it terminates in.
    As in the proof of Lemma~\ref{lem:seqfinite}, the executions of action $\alpha'$ in $\sched'$ involving $\alg$ action executions form a valid sequential schedule $\sched$ of $\alg$ starting in $C_0$ that makes the same system configuration changes as $\sched'$ w.r.t.\ the variables used in $\alg$.
    So $\sched$ reaches a configuration $C$ that is equivalent to $C'$ modulo amoebots' activity and awaken bits.
    Moreover, $\sched$ must terminate in $C$; otherwise, there exists an $\alg$-enabled amoebot in $C$ that, by Lemma~\ref{lem:seqenabled}, implies there exists an $\alg'$-enabled amoebot in $C'$, contradicting our supposition that $\alg'$ terminates in $C'$.
\end{proof}

It remains to show that all asynchronous schedules of $\alg'$ are finite in a sense that they only require a finite amount of time.

\begin{lemma} \label{lem:progress}
    If every sequential schedule of $\alg$ starting in $C_0$ is finite, then every asynchronous schedule of $\alg'$ starting in $C_0'$ is also finite.
\end{lemma}
\begin{proof}
    Suppose to the contrary that there exists an infinite asynchronous schedule $\sched$ of $\alg'$ starting in $C_0'$.
    First suppose that $\sched$ contains only a finite number of relevant action executions.
    Then there exists an earliest time $t$ after which no event associated with a relevant action execution is ever scheduled.
    Time $t$ is well-defined because (\textit{i}) every operation---and, by extension, every action execution---terminates in finite time and (\textit{ii}) there can be at most a finite number of irrelevant action executions initiated before time $t$ due to the fact that there are a finite number of amoebots, each amoebot executes at most one action per time, and any non-simultaneous events in $\sched$ are at least one time unit apart.
    Since $\sched$ is infinite and there always exists at least one active amoebot, there must exist an infinite number of action executions initiated after time $t$ and they must all be irrelevant.
    Recall that, by Lemma~\ref{lem:failure}, there are three types of irrelevant executions: those whose \Lock\ operation fails, those whose \Lock\ operation succeeds but that have $A.\awaken = \false$ and are $\alg$-disabled, and those whose \Lock\ operation succeeds but whose \Expand\ operation fails.
    
    It is easy to see that there must exist an execution of $\alpha'$ initiated after time $t$ whose \Lock\ operation succeeds; otherwise, all action executions initiated after time $t$ fail in their \Lock\ operation, a violation of the \Lock\ operation's deadlock freedom property.
    
    We next argue that some execution of $\alpha'$ initiated after time $t$ whose \Lock\ operation succeeds has $A.\awaken = \false$ and is $\alg$-enabled.
    Certainly, no execution of $\alpha'$ initiated after time $t$ with a successful \Lock\ operation could have $A.\awaken = \true$ as this execution would be relevant, contradicting our assumption on $t$.
    Any execution of $\alpha'$ that succeeds in its \Lock\ operation but is $\alg$-disabled sets its amoebot's activity bit to \false, disabling $\alpha'$.
    With a finite number of amoebots, there cannot be an infinite number of such executions.
    
    So consider any execution $(A, \alpha')$ initiated after time $t$ that succeeds in its \Lock\ operation, has $A.\awaken = \false$, and is $\alg$-enabled.
    This execution is irrelevant by supposition, so by Lemma~\ref{lem:failure}, its \Expand\ operation (say, into an adjacent node $v$) must fail.
    Convention~\ref{conv:valid} ensures that $A$ could not have called \Expand\ if it was expanded or if $v$ was occupied at the time of the corresponding guard evaluation, and $A$ cannot be involved in a movement initiated by some other amoebot because it is locked.
    The only way the \Expand\ operation of $(A, \alpha')$ could fail is if another amoebot $B$ successfully moves into $v$ during an execution $(B, \alpha')$ that is concurrent with $(A, \alpha')$.
    But if $(B, \alpha')$ succeeds in its movement operation, then all its operation executions must succeed by Lemma~\ref{lem:failure}; therefore, $(B, \alpha')$ is a relevant execution with an event occurring after time $t$, again contradicting our assumption on $t$.
    
    We conclude that $\sched$ must in fact contain an infinite number of relevant action executions.
    Moreover, when ordering these relevant action executions by the time their \Lock\ operations complete, there is at most a finite number of time units---and thus a finite number of irrelevant action executions---between any two consecutive relevant action executions.
    Thus, every relevant action execution has a well-defined, finite start time.
    
    Since $\sched$ contains an infinite number of relevant action executions, its sanitized version $\sched^*$ is also infinite.
    By Lemma~\ref{lem:sanitize} (which also holds for infinite schedules), $\sched^*$ is a valid asynchronous schedule that changes the system configuration exactly as $\sched$ does, except w.r.t.\ amoebots' activity bits.
    Let $D$ be the infinite DAG corresponding to $\sched^*$ (Lemma~\ref{lem:dag}).
    We argue next that Lemmas~\ref{lem:serialremove} and \ref{lem:serializability} apply to any snapshot of $\sched^*$ consistent with $D$.
    
    Consider the schedule $\hat{\sched}$ obtained by selecting the first $T \geq 1$ relevant action executions from $\sched^*$ ordered by the time their \Lock\ operations complete; if multiple action executions complete their \Lock\ operations simultaneously, we may assume any unique, canonical ordering of these action executions.
    Since all edges of the DAG $D$ of $\sched^*$ are forward in time w.r.t.\ completions of \Lock\ operations, $\hat{\sched}$ forms a consistent snapshot of $\sched^*$: for any edge $(A_i, \alpha_i') \to (A_j, \alpha_j')$ in $D$ with execution $(A_j, \alpha_j')$ contained in $\hat{\sched}$, we must have that $(A_i, \alpha_i')$ is also in $\hat{\sched}$.
    This snapshot property ensures that any memory accesses, contractions, and handovers execute in $\hat{\sched}$ in the same way as in $\sched^*$ since these only depend on the amoebots locked in the \Lock\ operations.
    Moreover, if $\hat{\sched}$ contains the first execution $(A_j, \alpha_j')$ to \Expand\ into a position after it is vacated during an execution $(A_i, \alpha_i')$ in $\sched^*$, then $(A_i, \alpha_i') \to (A_j, \alpha_j')$ is an edge in $D$ by DAG Rule~\ref{dag:move} and thus $(A_i, \alpha_i')$ is also contained in $\hat{\sched}$ because it is a consistent snapshot.
    Hence, $\hat{\sched}$ is a valid asynchronous schedule of $\alg'$ starting in $C_0'$.
    By Lemmas~\ref{lem:serialremove} and \ref{lem:serializability}, $\hat{\sched}$ can be mapped to a valid sequential schedule of $(A^E)'$ starting in $(C_0^E)'$ that contains at least $T$ action executions.
    
    This immediately implies that if there exists an infinite asynchronous schedule of $\alg'$ starting in $C_0'$, then there must also exist an infinite sequential schedule of $(A^E)'$ starting in $(C_0^E)'$.
    Otherwise, there exists a value of $T$ for which the above conversion fails, a contradiction.
    But this contradicts our original supposition: if every sequential schedule of $\alg$ starting in $C_0$ is finite, then every sequential schedule of $\alg^E$ starting in $C_0^E$ is finite by the termination condition of Convention~\ref{conv:expandrobust}, which in turn implies that every sequential schedule of $(A^E)'$ starting in $(C_0^E)'$ is finite by Lemma~\ref{lem:seqfinite}.
    This concludes the proof.
\end{proof}

We are now ready to prove Theorem~\ref{thm:concurrencycontrol}, concluding our analysis.

\begin{proof}[Proof of Theorem~\ref{thm:concurrencycontrol}]
    Every sequential execution of $\alg$ starting in $C_0$ terminates by supposition, so every asynchronous execution of $\alg'$ starting in $C_0'$ also terminates by Lemma~\ref{lem:progress}. 
    Consider any asynchronous schedule $\sched$ of $\alg'$ starting in $C_0'$ and let $C'$ be the configuration it terminates in.
    By Lemma~\ref{lem:sanitize}, the sanitized asynchronous schedule $\sched^*$ obtained from $\sched$ is valid and terminates in a configuration $C^*$ that is identical to $C'$, except $C^*$ may contain additional amoebots with \true\ activity bits.
    By Lemma~\ref{lem:serializability}, there exists a sequential schedule $\bar{\sched}$ of $(\alg^E)'$ starting in $(C_0^E)'$ that terminates in a configuration $(C^E)'$ that is identical to $C^*$, except $(C^E)'$ contains amoebots' expand flags and may also have additional amoebots with \true\ activity or awaken bits.
    Applying Lemma~\ref{lem:seqoutcome} to $\alg^E$ implies that there exists some sequential schedule of $\alg^E$ starting in $C_0^E$ that terminates in a configuration $C^E$ that is identical to $(C^E)'$, modulo amoebots' activity and awaken bits.
    Finally, because $\alg$ satisfies Convention~\ref{conv:expandrobust} by supposition, the correctness condition of expansion-robustness states that there exists a sequential execution of $\alg$ starting in $C_0$ that terminates in a configuration $C$ that is identical to $C^E$, modulo amoebots' expand flags.
    Therefore, $C$ and $C'$ are identical, modulo amoebots' activity and awaken bits, concluding the proof.
\end{proof}

\section{Discussion and Future Work} \label{sec:discuss}


An immediate application of the canonical amoebot model and its hierarchy of assumption variants is a systematic comparison of existing amoebot algorithms and their assumptions.
For example, when comparing two recent amoebot algorithms for leader election using the canonical hierarchy, we find that among other problem-specific differences, Bazzi and Briones~\cite{Bazzi2019-stationarydeterministic} assume an asynchronous adversary and common chirality while Emek et al.~\cite{Emek2019-deterministicleader} assume a sequential adversary and assorted orientations.
Such comparisons will provide valuable and comprehensive understanding of the state of amoebot literature and will facilitate clearer connections to related models of programmable matter.

The canonical amoebot model should also be extended to address \textit{fault tolerance} and \textit{self-stabilizing} algorithms.
This work assumed that all amoebots are reliable, though crash faults have been previously considered in the amoebot model for specific problems~\cite{Daymude2021-bioinspiredenergy,DiLuna2018-linerecovery}.
Faulty amoebot behavior is especially challenging for lock-based concurrency control mechanisms which are prone to deadlock in the presence of crash faults.
Additional modeling efforts will be needed to introduce a stable family of fault assumptions.

Finally, further study is needed on the design of concurrent amoebot algorithms.
Amoebots' communication and movement raise many issues of concurrency, ranging from conflicts of movement to operating based on stale information.
Our analysis of the \algHex\ algorithm produced one set of algorithm-agnostic invariants that yield correct asynchronous behavior without the use of locks (Lemmas~\ref{lem:hexsequential}--\ref{lem:hexsuccess}) while our concurrency control framework gives another set of sufficient conditions for obtaining correct behavior under an asynchronous adversary when using locks (Conventions~\ref{conv:valid}--\ref{conv:expandrobust}).

Of the three conventions used by the concurrency control framework, expansion-robustness (Convention~\ref{conv:expandrobust}) is the most restrictive and technically difficult to verify, though it is easier to understand and verify than the original ``monotonicity'' convention~\cite{Daymude2021-canonicalamoebot} that it replaced.
The framework's analysis relies on expansion-robustness to show that when an action execution is moved from its timing in an asynchronous schedule into the future where it is not concurrent with any other execution, it produces the same system configuration that it did originally, regardless of any new amoebots that may have moved into its neighborhood in the meantime.
In that light, it is easy to see that \textit{stationary} algorithms that do not use movement are trivially expansion-robust (Observation~\ref{obs:stationarycompatible}).
These include many of the existing algorithms for leader election~\cite{Bazzi2019-stationarydeterministic,Daymude2017-improvedleader,Derakhshandeh2015-leaderelection,DiLuna2020-shapeformation,Gastineau2019-distributedleader} and the recent algorithm for energy distribution~\cite{Daymude2021-bioinspiredenergy}.
However, many interesting collective behaviors for programmable matter require movement.
We proved that the \algHex\ algorithm is expansion-robust and compatible with the concurrency control framework (Theorem~\ref{thm:hexcompatible}).
Future work should investigate whether this is also true of other existing amoebot algorithms.

We emphasize that expansion-robustness is not simply a technicality of our approach but rather a general phenomenon for asynchronous amoebot systems.
Imagine a cycle alternating between contracted amoebots and empty positions and an asynchronous execution where all amoebots, having no neighbors, expand concurrently.
This forms a cycle of expanded amoebots.
However, any sequence of these expansions would result in at least one amoebot seeing an already expanded neighbor at the start of its action execution, which may prohibit its expansion and stop the system from reaching the original outcome (an expanded cycle).

This discussion highlights two open questions.
Do there exist amoebot algorithms that are not correct under an asynchronous adversary but are compatible with our concurrency control framework, establishing the necessity of lock-based approaches to concurrency control?
What are the necessary conditions for amoebot algorithm correctness in spite of asynchrony, both with and without locks?
We are hopeful that our approaches to concurrent algorithm design combined with answers to these open problems will advance the analysis of existing and future algorithms for programmable matter in the concurrent setting.

\bibliographystyle{plainurl}
\bibliography{ref}

\clearpage

\appendix

\section{Appendix: Amoebot Operation Pseudocode} \label{app:pseudocode}

In this appendix, we give formal distributed pseudocode for the amoebot operations.
Algorithm~\ref{alg:operationscomms} details the communication operations (Section~\ref{subsubsec:operationscomms}) and Algorithms~\ref{alg:operationsmoves} and~\ref{alg:operationshandovers} detail the movement operations (Section~\ref{subsubsec:operationsmoves}).
One possible implementation of the concurrency control operations (Section~\ref{subsubsec:operationsconcurrency}) is given in~\cite{Daymude2022-localmutual}.

\begin{algorithm}[tbh]
\caption{Communication Operations for Amoebot $A$} \label{alg:operationscomms}
\begin{algorithmic}[1]
    \Function{Connected}{$p$}
        \If {there is a neighbor connected via port $p$} \Return \true.
        \Else {} \Return \false.
        \EndIf
    \EndFunction
\end{algorithmic}
\begin{algorithmic}[1]
    \Function{Connected}{ }
        \State Let $k$ be the number of edges incident to the node(s) $A$ occupies.
        \State Snapshot the connectivity status of each port $p \in \{1, \ldots, k\}$.
        \State Let $c_p \gets N_i$ if neighbor $N_i$ is connected via port $p$ and $c_p \gets \false$ otherwise.
        \State \Return $[c_0, \ldots, c_{k-1}] \in \{N_1, \ldots, N_8, \false\}^k$.
    \EndFunction
\end{algorithmic}
\begin{algorithmic}[1]
    \Function{Read}{$p, x$}
        \State On being called:
            \Indent
                \If {$p = \bot$} \Return the value of $x$ in the public memory of $A$; success.
                \ElsIf {\Call{Connected}{$p$}} enqueue \texttt{read\_request}$(x)$ in the message buffer on $p$.
                \Else {} \textbf{throw} \texttt{disconnect-failure}.
                \EndIf
            \EndIndent
        \State On receiving \texttt{read\_request}$(x)$ via port $p'$:
            \Indent
                \State Let $x_{val}$ be the value of $x$ in the public memory of $A$.
                \State Enqueue \texttt{read\_ack}$(x, x_{val})$ in the message buffer on $p'$.
            \EndIndent
        \State On receiving \texttt{read\_ack}$(x, x_{val})$ via port $p$:
            \Indent
                \State \Return $x_{val}$; success.
            \EndIndent
        \State On disconnection via port $p$:
            \Indent
                \State \textbf{throw} \texttt{disconnect-failure}.
            \EndIndent
    \EndFunction
\end{algorithmic}
\begin{algorithmic}[1]
    \Function{Write}{$p, x, x_{val}$}
        \State On being called:
            \Indent
                \If {$p = \bot$} update the value of $x$ in the public memory of $A$ to $x_{val}$; \Return success.
                \ElsIf {\Call{Connected}{$p$}} enqueue \texttt{write\_request}$(x, x_{val})$ in the message buffer on $p$.
                \Else {} \textbf{throw} \texttt{disconnect-failure}.
                \EndIf
            \EndIndent
        \State On \texttt{write\_request}$(x, x_{val})$ being sent:
            \Indent
                \State \Return success.
            \EndIndent
        \State On disconnection via port $p$:
            \Indent
                \State \textbf{throw} \texttt{disconnect-failure}.
            \EndIndent
        \State On receiving \texttt{write\_request}$(x, x_{val})$ via port $p'$:
            \Indent
                \State Update the value of $x$ in the public memory of $A$ to $x_{val}$.
            \EndIndent
    \EndFunction
\end{algorithmic}
\end{algorithm}

\begin{algorithm}
\caption{Movement Operations for Amoebot $A$} \label{alg:operationsmoves}
\begin{algorithmic}[1]
    \Function{Contract}{$v$}
        \State On being called:
            \Indent
                \If {$A.\shape \neq \textsc{expanded}$} \textbf{throw} \texttt{shape-failure}.
                \ElsIf {$A$ is involved in a handover} \textbf{throw} \texttt{handover-failure}.
                \Else {} release all connections via ports on $v$ and begin contracting out of $v$.
                \EndIf
            \EndIndent
        \State On completing the contraction:
            \Indent
                \State Update $A.\shape \gets \textsc{contracted}$; \Return success.
            \EndIndent
    \EndFunction
\end{algorithmic}
\begin{algorithmic}[1]
    \Function{Expand}{$p$}
        \State Let $v_p$ denote the node that port $p$ faces.
        \State On being called:
            \Indent
                \If {$A.\shape \neq \textsc{contracted}$} \textbf{throw} \texttt{shape-failure}.
                \ElsIf {$A$ is involved in a handover} \textbf{throw} \texttt{handover-failure}.
                \ElsIf {\Call{Connected}{$p$}} \textbf{throw} \texttt{occupied-failure}.
                \Else {} begin expanding into $v_p$.
                \EndIf
            \EndIndent
        \State On collision with another amoebot:
            \Indent
                \State Perform contention resolution.
            \EndIndent
        \State On failing contention resolution:
            \Indent
                \State \textbf{throw} \texttt{occupied-failure}.
            \EndIndent
        \State On completing the expansion or on succeeding in contention resolution:
            \Indent
                \State Establish connections with any new neighbors adjacent to $v_p$.
                \State Update $A.\shape \gets \textsc{expanded}$; \Return success.
            \EndIndent
    \EndFunction
\end{algorithmic}
\end{algorithm}

\begin{algorithm}[p]
\caption{Movement Operations for Amoebot $A$ (cont.)} \label{alg:operationshandovers}
\begin{algorithmic}[1]
    \Function{Pull}{$p$}
        \State Let $v_p$ denote the node that port $p$ faces.
        \State On being called:
            \Indent
                \If {$A.\shape \neq \textsc{expanded}$} \textbf{throw} \texttt{shape-failure}.
                \ElsIf {$A$ is involved in a handover} \textbf{throw} \texttt{handover-failure}.
                \ElsIf {$\neg$\Call{Connected}{$p$}} \textbf{throw} \texttt{disconnect-failure}.
                \Else {} enqueue \texttt{pull\_request}$()$ in the message buffer on $p$.
                \EndIf
            \EndIndent
        \State On receiving \texttt{pull\_request}$()$ via port $p'$:
            \Indent
                \If {$A.\shape = \textsc{contracted}$ and $A$ is not involved in a move} set $m' \gets\texttt{pull\_ack}()$.
                \Else {} set $m' \gets\texttt{pull\_nack}()$.
                \EndIf
                \State Enqueue $m'$ in the message buffer on $p'$.
            \EndIndent
        \State On sending \texttt{pull\_ack}$()$:
            \Indent
                \State Begin expanding into $v_p$.
            \EndIndent
        \State On completing the expansion into $v_p$:
            \Indent
                \State Establish connections with any new neighbors adjacent to $v_p$.
                \State Update $A.\shape \gets \textsc{expanded}$.
            \EndIndent
        \State On receiving \texttt{pull\_ack}$()$ via port $p$:
            \Indent
                \State Release all connections via ports on $v_p$ except $p$ and begin contracting out of $v_p$.
            \EndIndent
        \State On receiving \texttt{pull\_nack}$()$ via port $p$ or on a disconnection via port $p$:
            \Indent
                \State \textbf{throw} \texttt{nack-failure}.
            \EndIndent
        \State On completing the contraction out of $v_p$:
            \Indent
                \State Update $A.\shape \gets \textsc{contracted}$; \Return success.
            \EndIndent
    \EndFunction
\end{algorithmic}
\begin{algorithmic}[1]
    \Function{Push}{$p$}
        \State Let $v_p$ denote the node that port $p$ faces.
        \State On being called:
            \Indent
                \If {$A.\shape \neq \textsc{contracted}$} \textbf{throw} \texttt{shape-failure}.
                \ElsIf {$A$ is involved in a handover} \textbf{throw} \texttt{handover-failure}.
                \ElsIf {$\neg$\Call{Connected}{$p$}} \textbf{throw} \texttt{disconnect-failure}.
                \Else {} enqueue \texttt{push\_request}$()$ in the message buffer on $p$.
                \EndIf
            \EndIndent
        \State On receiving \texttt{push\_request}$()$ via port $p'$:
            \Indent
                \If {$A.\shape = \textsc{expanded}$ and $A$ is not involved in a move} set $m' \gets\texttt{push\_ack}()$.
                \Else {} set $m' \gets\texttt{push\_nack}()$.
                \EndIf
                \State Enqueue $m'$ in the message buffer on $p'$.
            \EndIndent
        \State On sending \texttt{push\_ack}$()$:
            \Indent
                \State Release all connections via ports on $v_p$ except $p$ and begin contracting out of $v_p$.
            \EndIndent
        \State On completing the contraction out of $v_p$:
            \Indent
                \State Update $A.\shape \gets \textsc{contracted}$.
            \EndIndent
        \State On receiving \texttt{push\_ack}$()$ via port $p$:
            \Indent
                \State Begin expanding into $v_p$.
            \EndIndent
        \State On receiving \texttt{push\_nack}$()$ via port $p$ or on a disconnection via port $p$:
            \Indent
                \State \textbf{throw} \texttt{nack-failure}.
            \EndIndent
        \State On completing the expansion into $v_p$:
            \Indent
                \State Establish connections with any new neighbors adjacent to $v_p$.
                \State Update $A.\shape \gets \textsc{expanded}$; \Return success.
            \EndIndent
    \EndFunction
\end{algorithmic}
\end{algorithm}

\clearpage

\section{Appendix: Expansion Contention Resolution} \label{app:expandcontend}

Recall that when an amoebot's expansion collides with another movement, it must perform contention resolution such that exactly one contending amoebot succeeds in its expansion while all others fail.
In this appendix, we detail and analyze one possible implementation of such a contention resolution scheme inspired by randomized backoff mechanisms for contention resolution in wireless networks~\cite{Bender2005-adversarialcontention,Cali2000-ieee80211,Capetanakis1979-treealgorithms}.
We need one additional assumption: all amoebots know an upper bound $T$ on the time required for an amoebot to complete any movement.
For simplicity, we will assume geometric space (i.e., the triangular lattice $\Gtri$), though this mechanism would generalize to any bounded degree graph.

\begin{figure}[h]
    \centering
    \includegraphics[width=.6\textwidth]{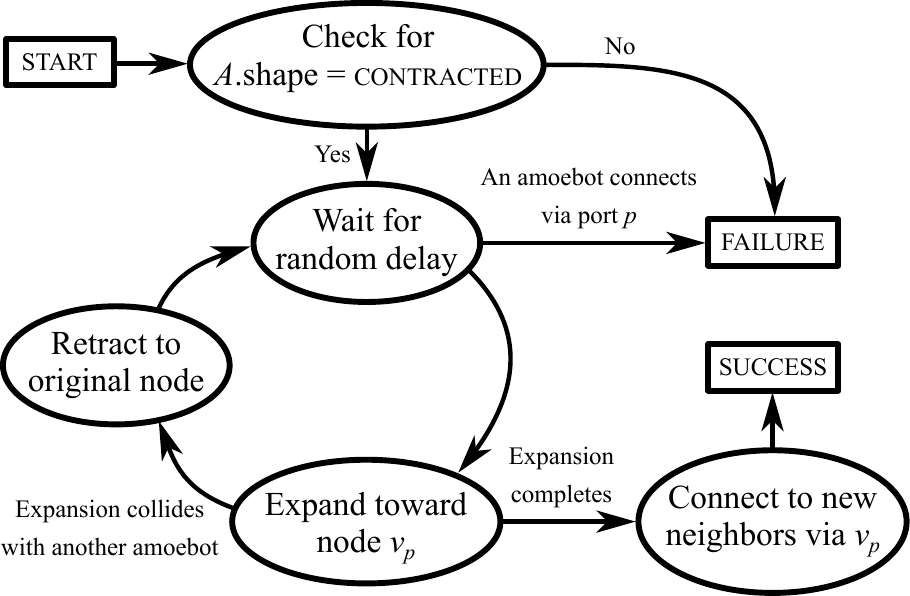}
    \caption{Execution flow of the \Expand\ operation with contention resolution for the calling amoebot $A$.}
    \label{fig:expandcontend}
\end{figure}

\begin{algorithm}[tbh]
\caption{\Expand\ Operation with Contention Resolution for Amoebot $A$} \label{alg:expandcontend}
\begin{algorithmic}[1]
    \Function{Expand}{$p$}
        \State Let $v_p$ denote the node that port $p$ faces.
        \State On being called:
            \Indent
                \If {$A.\shape \neq \textsc{contracted}$} \textbf{throw} \texttt{shape-failure}.
                \ElsIf {$A$ is involved in a handover} \textbf{throw} \texttt{handover-failure}.
                \Else {} wait for a delay of $0$.
                \EndIf
            \EndIndent
        \State After waiting for a delay:
            \Indent
                \If {$\neg$\Call{Connected}{$p$}} begin expanding into $v_p$.
                \Else {} \textbf{throw} \texttt{occupied-failure}.
                \EndIf
            \EndIndent
        \State On collision with another amoebot:
            \Indent
                \State Retract back out of $v_p$ and wait for a delay chosen u.a.r.\ from $[5T, 10T]$.
            \EndIndent
        \State On connection via port $p$:
            \Indent
                \State \textbf{throw} \texttt{occupied-failure}.
            \EndIndent
        \State On completing the expansion:
            \Indent
                \State Establish connections with any new neighbors adjacent to $v_p$.
                \State Update $A.\shape \gets \textsc{expanded}$; \Return success.
            \EndIndent
    \EndFunction
\end{algorithmic}
\end{algorithm}

The execution flow of our contention resolution mechanism is shown in \figtext~\ref{fig:expandcontend} and its pseudocode is given in Algorithm~\ref{alg:expandcontend}.
When $A$ detects a collision, it retracts to its original node and retries its expansion after waiting for a delay chosen uniformly at random from $[5T, 10T]$, where $T$ is an upper bound on the time required for an amoebot to complete an expansion or retraction.
In the remainder of this section, we verify the following claim.

\begin{lemma} \label{lem:expandcontend}
    Suppose a set of amoebots are contending to expand into the same node of $\Gtri$.
    If each amoebot waits for a delay chosen uniformly at random from $[5T, 10T]$ before its expansion attempt, then exactly one contender succeeds in $\bigo{\log n}$ attempts w.h.p.\footnote{An event occurs \textit{with high probability} (w.h.p.) if it occurs with probability at least $1 - 1/n^c$, where $n$ is the number of amoebots in the system and $c > 0$ is a constant.}
\end{lemma}
\begin{proof}
    We first bound the probability that two amoebots $A_1$ and $A_2$ collide in their respective expansion attempts into the same node.
    For each amoebot $A_i \in \{A_1, A_2\}$, let $t_i$ denote the start of its expansion attempt, $d_i$ denote its random delay, and $e_i$ denote the duration of its expansion if it were to succeed.
    The start time $t_i$ and expansion duration $e_i$ are fixed a priori by the adversary while the delay $d_i$ is chosen uniformly at random from the interval $[5T, cT]$, where $c > 5$ is a constant.
    So, in summary, amoebot $A_i \in \{A_1, A_2\}$ is waiting in the time interval $[t_i, t_i + d_i)$ and is expanding in the interval $[t_i + d_i, t_i + d_i + e_i]$.
    Thus, the expansions of amoebots $A_1$ and $A_2$ collide if and only if:
    \begin{align*}
        & [t_1 + d_1, t_1 + d_1 + e_1] \cap [t_2 + d_2, t_2 + d_2 + e_2] \neq \emptyset \\
        &\iff (t_1 + d_1 + e_1 \geq t_2 + d_2) \wedge (t_1 + d_1 \leq t_2 + d_2 + e_2) \\
        &\iff t_2 - t_1 - e_1 \leq d_1 - d_2 \leq t_2 - t_1 + e_2
    \end{align*}
    This implies:
    \begin{align*}
        & \Pr{\text{the expansions of $A_1$ and $A_2$ collide} \; | \; t_1, t_2, e_1, e_2} \\
        &= \Pr{t_2 - t_1 - e_1 \leq d_1 - d_2 \leq t_2 - t_1 + e_2} \\
        &= \Pr{d_1 - d_2 \leq t_2 - t_1 + e_2} - \Pr{d_1 - d_2 \leq t_2 - t_1 - e_1}
    \end{align*}

    Delays $d_1$ and $d_2$ are both uniform random variables over the interval $[5T, cT]$, so the difference $d_1 - d_2$ follows the symmetric triangular distribution with lower bound $(5 - c)T$, upper bound $(c - 5)T$, and mode $0$.
    W.l.o.g., suppose $t_1 < t_2$.
    There are two cases: when $t_2 - t_1 - e_1 \leq 0$ and when $t_2 - t_1 - e_1 > 0$.
    If we have $t_2 - t_1 - e_1 \leq 0$, then:
    \begin{align*}
        & \Pr{d_1 - d_2 \leq t_2 - t_1 + e_2} - \Pr{d_1 - d_2 \leq t_2 - t_1 - e_1} \\
        &= 1 - \frac{((c - 5)T - (t_2 - t_1 + e_2))^2}{((c - 5)T - (5 - c)T)((c - 5)T - 0)} - \frac{(t_2 - t_1 - e_1 - (5 - c)T)^2}{((c - 5)T - (5 - c)T)(0 - (5 - c)T)} \\
        &= \frac{2(c - 5)^2T^2 - ((c - 5)T - t_2 + t_1 - e_2)^2 - ((c - 5)T + t_2 - t_1 - e_1)^2}{2(c - 5)^2T^2} \\
        &= \frac{2(c - 5)^2T^2 - 2(c - 5)^2T^2 + 2(c - 5)Te_2 + 2(c - 5)Te_1 - 2t_2^2 + 4t_2t_1 - 2t_2e_2}{2(c - 5)^2T^2} \\
        &\phantom{=} \; + \frac{2t_2e_1 - 2t_1^2 + 2t_1e_2 - 2t_1e_1 - e_2^2 - e_1^2}{2(c - 5)^2T^2} \\
        &= \frac{2(c - 5)T(e_1 + e_2) - 2(t_2 - t_1)(e_2 - e_1) - 2(t_2 - t_1)^2 - e_1^2 - e_2^2}{2(c - 5)^2T^2} \\
        &< \frac{4(c - 5)T^2 + 2(c + 1)T^2}{2(c - 5)^2T^2} \\
        &= \frac{3(c - 3)}{(c - 5)^2},
    \end{align*}
    which is a constant probability when $c > \frac{13 + \sqrt{33}}{2} \approx 9.373$.
    The upper bound follows from:
    \begin{itemize}
        \item Since $T$ is the upper bound on the time required for an expansion, $e_1 + e_2 \leq 2T$.

        \item We assumed that $t_1 < t_2$, but we also have that if $t_2 > t_1 + d_1 + e_1$, then there cannot be a collision.
        Thus, $t_2 - t_1 \leq d_1 + e_1 \leq cT + T$ is a necessary condition for a collision.
        We also have that $-T \leq e_2 - e_1 \leq T$, so we conclude that $-2(t_2 - t_1)(e_2 - e_1) \leq 2(c + 1)T^2$.

        \item The last three numerator terms are all nonpositive, and thus can be upper bounded by $0$.
    \end{itemize}

    In the second case, if we have $t_2 - t_1 - e_1 > 0$, then:
    \begin{align*}
        & \Pr{d_1 - d_2 \leq t_2 - t_1 + e_2} - \Pr{d_1 - d_2 \leq t_2 - t_1 - e_1} \\
        &= 1 - \frac{((c - 5)T - (t_2 - t_1 + e_2))^2}{((c - 5)T - (5 - c)T)((c - 5)T - 0)} - 1 + \frac{((c - 5)T - (t_2 - t_1 - e_1))^2}{((c - 5)T - (5 - c)T)((c - 5)T - 0)} \\
        &= \frac{((c - 5)T - t_2 + t_1 + e_1)^2 - ((c - 5)T - t_2 + t_1 + e_2)^2}{2(c - 5)^2T^2} \\
        &= \frac{2(c - 5)Te_1 - 2(c - 5)Te_2 - 2t_2e_1 + 2t_2e_2 + 2t_1e_1 - 2t_1e_2 + e_1^2 - e_2^2}{2(c - 5)^2T^2} \\
        &= \frac{2(c - 5)T(e_1 - e_2) - 2(t_2 - t_1)(e_1 - e_2) + e_1^2 - e_2^2}{2(c - 5)^2T^2} \\
        &< \frac{2(c - 5)T^2 + 2(c + 1)T^2 + T^2}{2(c - 5)^2T^2} \\
        &= \frac{4c - 7}{2(c - 5)^2},
    \end{align*}
    which is a constant probability when $c > 6 + \sqrt{15/2} \approx 8.739$.
    Therefore, in any case, the probability that the expansions of $A_1$ and $A_2$ collide when their delays are drawn uniformly at random from the interval $[5T, cT]$ is at most a constant $p \in (0, 1)$ when $c > 9.373$.

    Due to the structure of the triangular lattice $\Gtri$, at most six amoebots may be concurrently expanding into the same node.
    We now establish that pairwise collisions of any of these amoebots' expansions are independent.
    Given each expansion attempt's starting time and expansion duration---which are fixed by the asynchronous execution---the interval of expansion is entirely determined by the delay.
    Since each delay is drawn independently and uniformly from $[5T, cT]$, each pair of expansions' time intervals and thus also their collision is independent.
    So, fixing an amoebot $A_1$,
    \begin{align*}
        & \Pr{\text{an expansion of $A_1$ succeeds} \; | \; t_1, e_1} \\
        &= \Pr{\text{the expansions of $A_1$ and $A_i$ do not collide} \; | \; t_1, t_i, e_1, e_i, \forall i \neq 1} \\
        &= \prod_{i \neq 1} (1 - \Pr{\text{the expansions of $A_1$ and $A_i$ collide} \; | \; t_1, t_i, e_1, e_i}) \\
        &> (1 - p)^5,
    \end{align*}
    which is a constant probability since $p$ is a constant probability.

    In order to amplify this success probability for the desired w.h.p.\ result, we must establish independence of expansion attempts.
    We have already shown that pairwise collisions of amoebots' expansions are independent, but this is insufficient to establish the independence of subsequent expansion attempts.
    In particular, $A_1$ and $A_2$ may collide while concurrently attempting to expand, causing them both to retract before reattempting their expansions.
    A third amoebot $A_3$ could then expand and collide with $A_1$ or $A_2$ while they are retracting, causing $A_3$ to also retract; a fourth amoebot $A_4$ could then expand and collide with $A_3$ while it retracts, and so on.
    In the worst case, if the expansions of $A_1$ and $A_2$ collide at time $t$, these cascading expansion-retraction collisions can continue until time $t + 5T$; this occurs if all retractions take the maximum time $T$ and each amoebot $A_i$ (for $i = 3, \ldots, 6$, since there are at most six competing amoebots) collides with retracting amoebot $A_{i-1}$ at the last possible moment.
    However, it is impossible for these cascading collisions to continue after $t + 5T$: the earliest an amoebot could reattempt its expansion is after time $t + 5T$ if $A_1$ or $A_2$ immediately retracted after colliding at time $t$ and then sampled the minimum possible delay, $5T$.
    Therefore, the expansion attempt of an amoebot $A_i$ is  independent of any of its previous attempts.
    So we have:
    \begin{align*}
        & \Pr{\text{no amoebot successfully expands after $k$ attempts}} \\
        &\leq \Pr{\text{$A_1$ collides in all $k$ expansion attempts}} \\
        &= \Pr{\text{$A_1$ collides in its $i$-th expansion attempt, $\forall i = 1, \ldots, k$}} \\
        &= \prod_{i=1}^k \left(1 - \Pr{\text{$A_1$ succeeds in its $i$-th expansion attempt} \; | \; t_1^i, e_1^i}\right) \\
        &< \left(1 - (1 - p)^5\right)^k
    \end{align*}
    Setting $k = \ln n / (1 - p)^5$, we have the probability that no amoebot successfully expands after $k$ attempts is at most:
    \[\left(1 - (1 - p)^5\right)^k \leq \exp\left\{-(1 - p)^5 \cdot \frac{\ln n}{(1 - p)^5}\right\} = \frac{1}{n}\]

    Once an amoebot's expansion succeeds, it connects to all its new neighbors causing any contending expansions to immediately fail.
    Therefore, we conclude that exactly one amoebot will successfully expand in at most $\ln n = \bigo{\log n}$ attempts with high probability.
\end{proof}

\end{document}